%% file: main.tex
\documentclass[a4paper,11pt]{article}
\input{JiaDef}

\begin{document}

\flushbottom

\title{\fontsize{15pt}{15pt}Weak Hopf non-invertible symmetry-protected topological spin liquid and lattice realization of (1+1)D symmetry topological field theory}

\author[a,b,\emaillogo]{Zhian Jia\orcidlink{0000-0001-8588-173X}}

\affiliation[a]{Centre for Quantum Technologies, National University of Singapore, SG 117543, Singapore}
\affiliation[b]{Department of Physics, National University of Singapore, SG 117543, Singapore}
\affiliation[\emaillogo]{Email: \href{mailto:giannjia@foxmail.com}{giannjia@foxmail.com}}
%\emailAdd{giannjia@foxmail.com}

\abstract{We propose weak Hopf symmetry as a general framework to explore (1+1)D topological phases that exhibit non-invertible symmetries. Inspired by the Symmetry Topological Field Theory (SymTFT) description of quantum phases with non-invertible symmetry, we construct a lattice model by introducing two distinct topological boundary conditions for a weak Hopf lattice gauge theory. One boundary encodes the topological symmetry information, while the other incorporates the non-topological dynamics. The resulting model is termed the cluster ladder model. We demonstrate that the cluster state model is a special case of this broader class of lattice models exhibiting weak Hopf symmetry $ H \times \hat{H} $, where $ H $ is a weak Hopf algebra and $ \hat{H} $ is its dual weak Hopf algebra. On a closed manifold, the symmetry reduces to $ \operatorname{Cocom}(H) \times \operatorname{Cocom}(\hat{H}) $, corresponding to the cocommutative subalgebras of $ H \times \hat{H} $. An essential weak Hopf sub-symmetry is $ \operatorname{Cocom}(H) \times \mathsf{Rep}(H) $, which, in the finite group case, reduces to the familiar symmetry $ G \times \mathsf{Rep}(G) $. To exactly solve the lattice model, we introduce a weak Hopf tensor network. Furthermore, we demonstrate how to construct the lattice realization of an arbitrary fusion category symmetry $ \EuScript{S} $ via combining Tannaka–Krein reconstruction or weak Hopf tube algebra and the cluster ladder model.}


\maketitle

\section{Introduction}

%\subsection{Overview}
The symmetry-protected topological (SPT) phases \cite{Haldane1983nonlinear,haldane1983continuum,Affleck1987AKLT,Gu2009SPT,Chen2011classification,Pollmann2012symmetry,Chen20112dSPT,Pollmann2012detection,Schuch2011classifyingSPT,Chen2012SPTboson,Zeng2015} have garnered significant attention in recent years. The ground state of an SPT phase is short-range entangled, and its topological order is characterized as an invertible topological order \cite{Wen2017zoo,freed2021reflection}. This implies that, on any closed manifold, SPT phases have a unique ground state.
Distinct SPT phases with a given symmetry cannot be smoothly deformed into each other without undergoing a phase transition, provided the deformation preserves the symmetry. However, if the symmetry is completely broken during the deformation, the SPT states can all be smoothly deformed into a product state.
The SPT phases are of particular interest due to their boundary behavior, the boundary of a SPT state is either gapless or degenerate. 
The boundary effective theory of a non-trivial SPT state always exhibits a pure gauge anomaly or a mixed gauge-gravity anomaly with respect to the symmetry group \cite{Wen2013SPT}.
The classification of \( d \)-dimensional bosonic SPT phases with an internal symmetry group \( G \) is described by the group cohomology \( H^{d+1}(G, U(1)) \) \cite{Wen2013SPT, Chen2013SPTclassification}\footnote{In this work, \( d \) denotes the spatial dimension, while \( D \) represents the spacetime dimension.}. However, it is known that some bosonic SPT phases cannot be captured by this cohomology-based framework \cite{Vishwanath2013physics, Wang2013SPT, Burnell2014spt}. This limitation has motivated a broader classification approach based on cobordism theory, in which SPT phases are characterized by the cobordism group \( \operatorname{Hom}(\Omega_{d+1}^G, U(1)) \). Here, \( \Omega_{d+1}^G \) denotes the bordism group of \((d+1)\)-dimensional manifolds with a \( G \)-structure \cite{kapustin2014SPT, kapustin2014bosonicSPT, freed2021reflection, yonekura2019cobordism}.
For fermionic SPT phases, a complete classification of non-interacting cases has been achieved via K-theory \cite{kitaev2009periodic, ryu2010topological}. In the case of interacting fermionic SPT phases, alternative frameworks based on group supercohomology \cite{Gu2014SPTfermion} and cobordism theory \cite{kapustin2015fermionic} have been proposed.
In Ref.~\cite{kong2020classification}, a classification based on minimal modular extensions of unitary braided fusion categories is proposed, which classifies both bosonic and fermionic SPT phases.

On the other hand, significant advancements have been made in the study of generalized symmetries, which encompass new and novel forms of symmetries. Notable examples include high-form symmetries \cite{gaiotto2015generalized, kapustin2017higher, gomes2023introduction, Bhardwaj2024lecture}, non-invertible symmetries \cite{cordova2022snowmass, brennan2023introduction, mcgreevy2023generalized, luo2023lecture, shao2024whats, SchaferNameki2024ICTP, Bhardwaj2024lecture, delcamp2024higher}, fusion category symmetries \cite{Frohlich2004kramers, fuchs2007topological, frohlich2010defect, bhardwaj2018finite, chang2019topological, thorngren2019fusion, thorngren2021fusion, komargodski2021symmetries, inamura2022lattice, inamura2023fermionization}, algebraic higher symmetries \cite{Kong2020algebraic}, and Hopf and weak Hopf symmetries \cite{bais2003hopf, Buerschaper2013a, chen2021ribbon, ciamprone2024weakquasihopfalgebrasctensor,jia2023boundary, Jia2023weak, jia2024weakTube, jia2024generalized}. Category theory provides a general framework for us to investigate these generalized symmetries. The SPT phases protected by these generalized symmetries have attracted considerable attention \cite{inamura2021topological,Kong2020algebraic, SchaferNameki2024ICTP, huang2023topologicalholo, bhardwaj2024lattice, freed2024topSymTFT, gaiotto2021orbifold, bhardwaj2023generalizedcharge, apruzzi2023symmetry, bhardwaj2024gappedphases, Zhang2024anomaly, Ji2020categoricalsym}.

The concept of categorical symmetry can be understood through a holographic perspective \cite{Kong2020algebraic, Ji2020categoricalsym, chatterjee2024TopHolo}: an \( n \)-dimensional categorical symmetry corresponds to a topological order in one higher dimension. This approach, now referred to as symmetry topological field theory (SymTFT) or topological holography, provides a comprehensive framework for understanding non-invertible or categorical symmetries in both gapped and gapless phases \cite{Kong2020algebraic, SchaferNameki2024ICTP, huang2023topologicalholo, bhardwaj2024lattice, freed2024topSymTFT, gaiotto2021orbifold, bhardwaj2023generalizedcharge, apruzzi2023symmetry, bhardwaj2024gappedphases, Zhang2024anomaly, Ji2020categoricalsym, bhardwaj2024clubsandwich, bhardwaj2024hassediagramsgaplessspt}.

Consider an \( n \)-dimensional topological order,
the $(n-1)$-dimensional defects can be regarded as 1-morphisms, $(n-2)$-dimensional defects between $(n-1)$-dimensional defects can be regarded as 2-morphisms, and so on. Thus, the corresponding topological order is described by a fusion $n$-category \cite{kong2014braided,kong2017boundary,Kong2020algebraic,johnson2022classification}.
We can combine all $d$-dimensional topological orders to form a category $\mathbf{TO}_{d}$. For this category, we need to include domain walls between different topological orders and the stacking operation of topological orders to ensure that $\mathbf{TO}_{d}$ is closed.
From Hamiltonian point of view, consider the moduli space $\mathfrak{M}_{nd}$ formed by all $nd$ gapped liquid Hamiltonians. An elements in $\pi_0(\mathfrak{M}_{nd})$ is a topological order \cite{zeng2015quantum,Swingle2016,Hsin2023moduli}.
Two topological orders are considered equivalent if they can be smoothly deformed into each other through a continuous path of Hamiltonians without closing the energy gap.
These equivalence classes of topological orders are described first fundamental group $\pi_1(\mathfrak{M})$. From algebraic point of view, two topological orders are equivalent if they can be connected by some invertible domain wall, the invertible domain walls are classified by  $\pi_1(\mathfrak{M})$.
The $(n-1)$-dimensional defects can fuse but not braid. If we exclude these defects, we obtain a braided fusion $(n-1)$-category of excitations. This braided fusion category is the looping of the $n$-fusion category of topological order \cite{Kong2020algebraic,johnson2022classification}.
A special type of topological order is called an invertible topological order $\eC$, which is invertible under the stacking operation. Specifically, there exists a topological order $\eD$ such that the stacking of $\eC$ and $\eD$ yields the trivial topological order: $\eC \boxtimes \eD \simeq \eI$. The SPT phases possess exclusively invertible topological order, which guarantees the unique ground state of the SPT phase on any closed manifold.

For an \( n \)-dimensional system, its categorical symmetry is characterized by a fusion \( n \)-category \( \mathcal{S} \). From a topological holography perspective, this fusion \( n \)-symmetry \( \mathcal{S} \) corresponds to an \( (n+1) \)-dimensional topological order, which is described by the Drinfeld center \( \mathcal{Z}(\mathcal{S}) \). Using the SymTFT sandwich construction \cite{Kong2020algebraic, SchaferNameki2024ICTP, huang2023topologicalholo, bhardwaj2024lattice, freed2024topSymTFT, gaiotto2021orbifold, bhardwaj2023generalizedcharge, apruzzi2023symmetry, bhardwaj2024gappedphases, Zhang2024anomaly, Ji2020categoricalsym,bhardwaj2024clubsandwich,bhardwaj2024hassediagramsgaplessspt}, this relationship can be understood as follows: the symmetry \( \mathcal{S} = \mathcal{B}_{\text{sym}} \) represents a topological boundary condition of the topological field theory \( \mathcal{Z}(\mathcal{S}) \). Additionally, there may exist another boundary, \( \mathcal{B}_{\text{phys}} \), of \( \mathcal{Z}(\mathcal{S}) \), which encodes the non-topological aspects of the system. The boundary \( \mathcal{B}_{\text{phys}} \) can be either gapped or gapless.
For an \( n = 1 \)-dimensional system, the categorical symmetry is described by a fusion category \( \mathcal{S} \) (fusion 1-category). The symmetry is considered anomaly-free if there exists a 1-dimensional \( \mathcal{S} \)-SPT phase; otherwise, it is referred to as anomalous. For an anomaly-free \( \mathcal{S} \), the \( \mathcal{S} \)-SPT phases are classified by fiber functors \( F: \mathcal{S} \to \mathsf{Vect} \), which are tensor functors mapping \( \mathcal{S} \) to the category of finite-dimensional vector spaces \cite{thorngren2019fusion}.

An interesting and challenging question is whether it is possible to construct a lattice model that realizes a given \( \mathcal{S} \)-SPT phase \cite{aasen2016topological, aasen2020topological, inamura2022lattice, Inamura2024fusionSuface, eck2024generalizationskitaevshoneycombmodel, fechisin2023noninvertible, bhardwaj2024lattice, bhardwaj2024gappedphases, Feiguin2007interacting, trebst2008short, buican2017anyonic, Lootens2023dualityHamiltonian, Lootens2024duality, Seiberg2024noninvertible, Seifnashri2024cluster, inamura202411dsptphasesfusion, cordova2024noninvertiblesymmetriesfinitegroup}. Despite the numerous examples of \( \mathcal{S} \)-symmetric phases, a systematic approach for constructing such models remains elusive. In this work, we address this question by exploring it within the framework of weak Hopf symmetry.

For a $(1+1)$D system, the categorical symmetry is described by a (multi)fusion category $\mathcal{S}$. It has been proven that any (multi)fusion category can be realized as the representation category of a weak Hopf algebra~\cite{szlachanyi2000finite,ostrik2003module,etingof2005fusion}. Motivated by this result, we introduce weak Hopf symmetries and use them to construct lattice models by generalizing the cluster state model.

The cluster state model is a widely studied example of lattice models exhibiting SPT phases. For the standard $\mathbb{Z}_2$ cluster state, the Hamiltonian takes the form  
\begin{equation}\label{eq:ClusterHamiltonianNormal}
    H_{\mathrm{cluster}} = -\sum_i Z_{i-1} X_i Z_{i+1},
\end{equation}
which is known to host $\mathbb{Z}_2 \times \mathbb{Z}_2$ SPT order~\cite{son2012topological}.
By applying Hadamard gates to all even vertices of the lattice, one can map this system to the Calderbank-Shor-Steane (CSS)-type cluster state~\cite{brell2015generalized,fechisin2023noninvertible,jia2024generalized}. The stabilizers for the transformed state are of the CSS-type, specifically $X_{i-1} X_i X_{i+1}$ for odd sites ($i \in 2\mathbb{N}+1$) and $Z_{i-1} Z_i Z_{i+1}$ for even sites ($i \in 2\mathbb{N}$). This leads to a modified Hamiltonian of the form
\begin{equation}\label{eq:ClusterHamiltonianCSS}
    H = -\sum_{i \in 2\mathbb{N}} Z_{i-1} Z_i Z_{i+1} - \sum_{i \in 2\mathbb{N}+1} X_{i-1} X_i X_{i+1}.
\end{equation}
Recent studies have revealed that the cluster state model exhibits non-invertible global symmetries. In particular, the symmetry can be described by the fusion category $\Rep(D_8)$, where $D_8$ is the dihedral group of order 8~\cite{Seifnashri2024cluster}.  

Since the qubit-based cluster state corresponds to systems valued in the Abelian group $\mathbb{Z}_2$, a natural generalization is to consider systems defined over arbitrary finite groups $G$~\cite{brell2015generalized}. This generalization involves introducing higher-dimensional analogs of Pauli operators, similar to those used in Kitaev's quantum double model~\cite{Kitaev2003,albert2021spin}.  
The SPT order and non-invertible symmetries of such finite-group cluster states were recently analyzed in Ref.~\cite{fechisin2023noninvertible}. Using tensor network techniques, it was demonstrated that the symmetry of a $G$-valued cluster state is described by $G \times \Rep(G)$, where $\Rep(G)$ is the representation category of $G$.  
We generalize these results to the Hopf algebra setting in Ref.~\cite{jia2024generalized}. 
Furthermore, Ref.~\cite{jia2024generalized} demonstrates that the cluster state model can be understood as a specific realization of a quantum double ladder model, which in turn provides a pathway for extending the construction to more general weak Hopf algebras (and even to weak quasi-Hopf algebras).

In this work, we systematically construct lattice models that realize an \( \eS \)-symmetric phase (more precisely the phase with $\eS$ fusion algebra symmetry). To achieve this, we introduce a general framework for weak Hopf symmetric phases and demonstrate how such phases can be realized via SymTFT. We also show that the (multi)fusion symmetry can be interpreted as a dual weak Hopf symmetry.
Building on the observation that a cluster state SPT phase can be viewed as a thin quantum double model with one smooth and one rough boundary \cite{jia2024generalized}, we introduce the \emph{weak Hopf cluster state model}. Inspired by this construction, we further propose a generalized model, termed the \emph{weak Hopf cluster ladder model}, where the two topological boundaries are chosen to be distinct and do not allow any non-trivial anyonic tunneling channels between them. 
A crucial observation is that the symmetry of the (1+1)D cluster ladder model is given by \( H \times \hat{H} \), where \( H \) represents a weak Hopf symmetry and \( \hat{H} \) is its dual weak Hopf symmetry. Notably, the \( \mathbb{Z}_2 \times \mathbb{Z}_2 \) symmetry of the \( \mathbb{Z}_2 \)-cluster state, the \( G \times \Rep(G) \) symmetry of the \( G \)-cluster state model, and the \( \Cocom(H) \times \Rep(H) \) symmetry are all special cases of this general framework.
Also note that our model is quite different from that based on anyonic chains \cite{Feiguin2007interacting, trebst2008short, buican2017anyonic, bhardwaj2024lattice, bhardwaj2024gappedphases}, where the Hilbert space is constructed from the fusion tree and thus does not have a well-defined local space or tensor product structure \footnote{While a tensor product structure can be assigned to each edge, followed by a projection to enforce the two labels on the edge to be identical, this approach for introducing tensor product structure is not so natural.}.

The paper is organized as follows. In Section~\ref{sec:groupcase}, we use finite groups as examples to illustrate the key ideas of our lattice model. For the $\mathbb{Z}_p$ group with $p$ a prime number, there exists a single type of cluster ladder model, which corresponds to the $\mathbb{Z}_p$ cluster state model. In contrast, for the $S_3$ group, two distinct types of cluster ladder models that realize SPT phases arise: one corresponds to the cluster state model, while the other exhibits topological boundary conditions that differ from the conventional smooth and rough boundaries.
In Section~\ref{sec:WHAsymmetry}, we provide a detailed discussion of weak Hopf symmetries, introducing the necessary mathematical formalism and results. We begin with the case of invertible symmetries, which correspond to group algebra Hopf symmetries, and then extend the analysis to general weak Hopf symmetries and comodule algebra symmetries. Furthermore, we demonstrate how fusion category symmetries can be interpreted as dual weak Hopf symmetries. Finally, we explore the implications of weak Hopf symmetries for single-qudit systems.  
Section~\ref{sec:tensor-network} introduces weak Hopf tensor network states, which serve as a powerful tool for analyzing and solving our lattice model.
In Section~\ref{sec:latticeI}, we present the weak Hopf generalization of the cluster state model. The model is solved using weak Hopf tensor network states. We explore the weak Hopf symmetry of the model on both closed and open manifolds, demonstrating that there is a $H \times \hat{H}$ symmetry for open manifolds, while for closed manifolds, the symmetry is given by $\Cocom(H) \times \Cocom(\hat{H})$.  
Section~\ref{sec:latticeII} extends the cluster state model to incorporate general boundary conditions. The lattice model is constructed using comodule algebras and can also be solved via weak Hopf tensor network states.  
Finally, Section~\ref{sec:Discussion} provides a discussion of the results and outlines several open problems for future research.

\section{Warm-up example: finite-group quantum cluster ladder}
\label{sec:groupcase}
In Ref.~\cite{jia2024generalized}, it is demonstrated that the cluster state model is indeed a quantum double ladder model with rough and smooth boundaries. The aim of this paper is to generalize the construction presented therein to arbitrary weak Hopf algebras, encompassing all possible non-invertible symmetries in (1+1)D.  
Before delving into the cluster ladder model realization of non-invertible symmetry for the most general weak Hopf case, let us first consider a warm-up example of the finite-group cluster ladder model, focusing specifically on the \(\mathbb{Z}_p\) and \(S_3\) quantum cluster ladder models. The finite-group cluster state model is also discussed in Refs.~\cite{brell2015generalized,fechisin2023noninvertible}. The macroscopic theory of the corresponding SymTFT is given in Ref.~\cite{huang2023topologicalholo}. Here, we emphasize the comodule algebra approach and discuss it in a more general setting via the so-called quantum cluster ladder model.

The group-based Pauli X operators are defined as 
\begin{equation}
    \XR_g|h\rangle = |gh\rangle,\quad \XL_g|h\rangle =|hg^{-1}\rangle.
\end{equation}
For some irreducible representation $\Gamma \in \Rep(G)$, the group-based Pauli Z operators are defined by
\begin{equation}
    Z_{\Gamma} = \sum_{g\in G} \Gamma(g) \otimes |g\rangle \langle g|,\quad Z^{\ddagger}_{\Gamma}=\sum_{g\in G} \Gamma(g^{-1})|g\rangle \langle g|.
\end{equation}
If we take trace of the representation matrix, we obtain
\begin{equation}
    \ZR_{\chi_{\Gamma}}=\sum_{g\in G} \chi_{\Gamma}(g)|g\rangle\langle g|, \quad \ZL_{\chi_{\Gamma}}=\sum_{g\in G} \chi_{\Gamma}(g^{-1})|g\rangle\langle g|,
\end{equation}
where $\chi_{\Gamma}$ is character of $\Gamma \in \Rep(G)$.

For a quantum double model $D(G)$ of finite group $G$, the topological boundary conditions are determined by a pair $(K,\omega)$ with $K$ a subgroup of $G$ and $\omega$ a two-cocycle \cite{Kitaev2003,Beigi2011the,Cong2017,Bombin2008family}.
The topological boundary conditions are equivalently classified by the comodule algebras over Hopf algebra $\Cbb[G]$ \cite{jia2023boundary,Jia2023weak}. Here for simplicity, we will only consider the case with trivial 2-cocycle.
The 1d cluster ladder is defined as the following lattice (with periodic boundary condition at two ends):
\begin{equation}\label{eq:ClusterLadderG}
\begin{aligned}
\begin{tikzpicture}
    % Define the number of rungs in the ladder
    \def\n{5}
    % Define the size of each square
    \def\s{1} 
    % Draw the shaded background lattice
    \fill[green!20] (0, 0) rectangle (\n*\s+\s, \s); % Rectangle covering the whole background
    % Draw the ladder with arrows in the middle of each edge
    \foreach \i in {0,...,\n} {
        % Draw solid bottom edges with arrows in the middle pointing right
        \draw[-stealth, line width=1.0pt,blue, midway] (\i*\s, 0) -- (\i*\s+\s, 0);
        % Draw dotted top edges with arrows in the middle pointing right
        \draw[-stealth, line width=1.0pt,red, midway] (\i*\s, \s) -- (\i*\s+\s, \s);
        % Draw upward ladder edges with arrows in the middle pointing up
        \draw[-stealth,line width=1.0pt, midway] (\i*\s, 0) -- (\i*\s, \s);
    }
    % Draw the right-most vertical ladder edge with an arrow in the middle pointing up
    \draw[-stealth, midway,line width=1.0pt] (\n*\s+\s, 0) -- (\n*\s+\s, \s);
          \draw[-stealth, white, line width=2pt, midway] (6, 0) -- (6, 1.02);
\end{tikzpicture}
\end{aligned}
\end{equation}
For the symmetry boundary (blue edges), we place $K=\Cbb[M]$ on each edge. For the physical boundary (red edges), we place $J=\Cbb[N]$ on each edge. In the bulk (black edges), we assign $W=\Cbb[G]$ to each edge.
This means that the total Hilbert space is $\mathcal{H}_{\rm tot}=\cH_{\rm bk}\otimes \cH_{\rm sym} \otimes \cH_{\rm phys}$ with $\cH_{\rm bk}=\otimes_j \Cbb[G]$, $\cH_{\rm sym} = \otimes_k \Cbb[M]$, and $\cH_{\rm phys}=\otimes_l \Cbb[N]$.

The SymTFT bulk operator is the face operator which is defined as
\begin{equation}
    \Bf_f = \sum_{\Gamma\in \operatorname{Irr}(G)}  \frac{d_{\Gamma}}{|G|}\Bf_f^{\chi_{\Gamma}}
\end{equation}
where 
\begin{equation}\label{eq:Bf}
	\Bf_f^{\chi_{\Gamma}}
	\big{|}	\begin{aligned}
		\begin{tikzpicture}
                    \fill[green!20] (-0.5, -0.5) rectangle ++(1,1); % Rectangle 
			\draw[-stealth, line width=1.0pt,red, midway] (-0.5,0.5) -- (0.5,0.5);
		\draw[-stealth,black] (-0.5,-0.5) -- (-0.5,0.5); 
		\draw[-stealth,black] (0.5,-0.5) -- (0.5,0.5); 
		\draw[-stealth, line width=1.0pt,blue, midway] (-0.5,-0.5) -- (0.5,-0.5); 
			\draw [fill = black] (0,0) circle (1.2pt);
			\node[ line width=0.2pt, dashed, draw opacity=0.5] (a) at (0.75,0){$x_1$};
			\node[ line width=0.2pt, dashed, draw opacity=0.5] (a) at (-0.75,0){$x_3$};
			\node[ line width=0.2pt, dashed, draw opacity=0.5] (a) at (0,-0.7){$x_4$};
			\node[ line width=0.2pt, dashed, draw opacity=0.5] (a) at (0,0.7){$x_2$};
		\end{tikzpicture}
	\end{aligned}   \big{ \rangle}     
	= 
	\chi_{\Gamma}(x_1^{-1}x_2x_3x_4^{-1})
	\big{|}	\begin{aligned}
	\begin{tikzpicture}
            \fill[green!20] (-0.5, -0.5) rectangle ++(1,1); % Rectangle 
		\draw[-stealth, line width=1.0pt,red, midway] (-0.5,0.5) -- (0.5,0.5);
		\draw[-stealth,black] (-0.5,-0.5) -- (-0.5,0.5); 
		\draw[-stealth,black] (0.5,-0.5) -- (0.5,0.5); 
		\draw[-stealth, line width=1.0pt,blue, midway] (-0.5,-0.5) -- (0.5,-0.5); 
		\draw [fill = black] (0,0) circle (1.2pt);
		\node[ line width=0.2pt, dashed, draw opacity=0.5] (a) at (0.75,0){$x_1$};
		\node[ line width=0.2pt, dashed, draw opacity=0.5] (a) at (-0.75,0){$x_3$};
		\node[ line width=0.2pt, dashed, draw opacity=0.5] (a) at (0,-0.7){$x_4$};
		\node[ line width=0.2pt, dashed, draw opacity=0.5] (a) at (0,0.7){$x_2$};
	\end{tikzpicture}
\end{aligned}   \big{ \rangle} .
\end{equation}
The symmetry boundary vertex operator is defined as
\begin{equation}\label{eq:Av}
	\Av_{v_s}
\big{|}	\begin{aligned}
		\begin{tikzpicture}
                    \fill[green!20] (-0.5, 0) rectangle ++(1,0.5); % Rectangle 
			\draw[-stealth, line width=1.0pt,blue, midway] (-0.5,0) -- (0,0);
			\draw[-stealth,black] (0,0) -- (0,0.5); 
			\draw[-stealth, line width=1.0pt,blue, midway] (0,0) -- (0.5,0); 
		%	\draw[-latex,black] (0,-0.5) -- (0,0); 
			\draw [fill = black] (0,0) circle (1.2pt);
			\node[ line width=0.2pt, dashed, draw opacity=0.5] (a) at (0.7,0){$x_1$};
			\node[ line width=0.2pt, dashed, draw opacity=0.5] (a) at (-0.7,0){$x_3$};
			%\node[ line width=0.2pt, dashed, draw opacity=0.5] (a) at (0,-0.7){$x_4$};
			\node[ line width=0.2pt, dashed, draw opacity=0.5] (a) at (0,0.7){$x_2$};
		\end{tikzpicture}
	\end{aligned}   \big{ \rangle}     
= 
\frac{1}{|M|} \sum_{g\in M}
\big{|}	\begin{aligned}
	\begin{tikzpicture}
                        \fill[green!20] (-0.5, 0) rectangle ++(1,0.5); % Rectangle 
		\draw[-stealth, line width=1.0pt,blue, midway] (-0.5,0) -- (0,0);
		\draw[-stealth,blue,black] (0,0) -- (0,0.5); 
		\draw[-stealth, line width=1.0pt,blue, midway] (0,0) -- (0.5,0); 
		%\draw[-latex,black] (0,-0.5) -- (0,0); 
		\draw [fill = black] (0,0) circle (1.2pt);
		\node[ line width=0.2pt, dashed, draw opacity=0.5] (a) at (1,0){$x_1g^{-1}$};
		\node[ line width=0.2pt, dashed, draw opacity=0.5] (a) at (-1,0){$gx_3$};
		%\node[ line width=0.2pt, dashed, draw opacity=0.5] (a) at (0,-0.7){$gx_4$};
		\node[ line width=0.2pt, dashed, draw opacity=0.5] (a) at (0,0.75){$x_2g^{-1}$};
	\end{tikzpicture}
\end{aligned}   \big{ \rangle} .
\end{equation}
The physical boundary vertex operator is defined similarly. Both the face and vertex operators are equipped with a counterclockwise ordering (around the face and around the vertex). Alternatively, we can also define the model with a clockwise ordering.
The Hamiltonian is of the form
\begin{equation}
    \mathbb{H}_G^{M,N}=\mathbb{H}_{\rm bk}+\mathbb{H}^M_{\rm sym} + \mathbb{H}^N_{\rm phys} 
\end{equation}
where $\mathbb{H}_{\rm bk}=\sum_f \Bf_f$, $\mathbb{H}_{\rm sym} =\sum_{v_s}\Av_{v_s}$ and $\mathbb{H}_{\rm phys}=\sum_{v_p} \Av_{v_p}$. This model realize the SymTFT with fusion category symmetry give by the boundary excitations over the symmetry boundary.

The topological charges of the 2d bulk are labeled by the irreducible representations of the quantum double \( D(G) \) \cite{Kitaev2003,Beigi2011the,Cong2017,Bombin2008family,jia2023boundary,Jia2023weak}. These representations are classified by two components: the conjugacy class \( [g] \) of \( G \), determining the magnetic charge, and the irreducible representations \( \pi \) of the centralizer \( C_G([g]) \), labeling the electric charge.
For symmetry boundaries $K=\Cbb[M]$, the boundary excitation is described by the fusion category 
\begin{equation}
    \eB_{\rm sym}={^{\Cbb[G]}_K}\Mod_K ,
\end{equation}
which is the category of \(\Cbb[G]\)-covariant \(K|K\)-bimodules \cite{andruskiewitsch2007module,jia2023boundary,Jia2023weak}. The boundary excitations are equivalently classified by pairs \((S, \pi)\), where \(S \in M\backslash G/M\) is a double coset and \(\pi\) is an irreducible representation of \(M \cap r M r^{-1}\), with \(r \in S\) \cite{Cong2017}. For physical boundaries $J=\Cbb[N]$, the excitations can be characterized in a similar manner.
Since we have chosen both boundaries to be topological boundaries, the cluster ladder model exhibits the following symmetry \footnote{This reflects the fact that, in a lattice model, one cannot distinguish which topological boundary is a symmetry boundary, as both boundaries are on equal footing.}:
\begin{equation}
    \eB_{\rm sym} \times \eB_{\rm phys} = {^{\Cbb[G]}_K}\Mod_K \times {_J}\Mod^{\Cbb[G]}_J,
\end{equation}
where \(\Cbb[G]\) is placed on the left and right, respectively, to emphasize that \(K\) is a left \(\Cbb[G]\)-comodule algebra and \(J\) is a right \(\Cbb[G]\)-comodule algebra.

The CSS-type cluster state model is a special case of the quantum cluster ladder model, featuring one rough boundary and one smooth boundary~\cite{jia2024generalized}. 
For this case, \( M = G \) and \( N = \{1\} \). The symmetry boundary excitation becomes \(\eB_{\rm sym} = \Rep(G)\), while the physical boundary excitation becomes \(\Vect_G\), the category of $G$-graded vector spaces. 
Notice that for the rough boundary edge, the boundary space reduces to \(\Cbb\), allowing us to effectively remove it. Additionally, \(\Av_{v_p}\) becomes the identity operator, so we can also remove the boundary vertex operator. The lattice then simplifies to:
\begin{equation}\label{eq:ClusterLatticeG}
\begin{aligned}
\begin{tikzpicture}
    % Define the number of rungs in the ladder
    \def\n{5}
    % Define the size of each square
    \def\s{1}
        % Draw the shaded background lattice
    \fill[green!20] (0, 0) rectangle (\n*\s+\s, \s); % Rectangle covering the whole background
    % Draw the ladder with arrows in the middle of each edge
    \foreach \i in {0,...,\n} {
        % Draw solid bottom edges with arrows in the middle pointing right
        \draw[-stealth, line width=1.0pt,blue, midway] (\i*\s, 0) -- (\i*\s+\s, 0);
        % Draw dotted top edges with arrows in the middle pointing right
        \draw[dotted, line width=1.0pt,red, midway] (\i*\s, \s) -- (\i*\s+\s, \s);
        % Draw upward ladder edges with arrows in the middle pointing up
        \draw[-stealth,line width=1.0pt, midway] (\i*\s, 0) -- (\i*\s, \s);
    }
    % Draw the right-most vertical ladder edge with an arrow in the middle pointing up
    \draw[-stealth, midway,line width=1.0pt] (\n*\s+\s, 0) -- (\n*\s+\s, \s);
          \draw[-stealth, white, line width=2pt, midway] (6, 0) -- (6, 1.02);
\end{tikzpicture}
\end{aligned}
\end{equation}
and the corresponding Hamiltonian becomes 
\begin{equation}
    \Hbb_{G}^{G,\{1\}}= -\sum_f \Bf_f^G-\sum_{v_s: \text{sym}} \Av^{G}_{v_s}.
\end{equation}
The symmetry of the model is 
\begin{equation}
     \eB_{\rm sym} \times \eB_{\rm phys} = \Rep(G)\times \Vect_G.
\end{equation}
Notice that choosing different configurations of orientations for each edge results in distinct cluster state models, and choosing a clockwise local ordering also yields a different model. However, all these models realize the same phase. The models discussed in Refs.~\cite{fechisin2023noninvertible,jia2024generalized} are examples of the clockwise model. In this paper, we follow the conventions in Refs.~\cite{Kitaev2003,jia2023boundary,Jia2023weak} by using the counterclockwise local ordering. This will be elaborated in more detail later when discussing weak Hopf cluster ladder models.

\subsection{Abelian cluster ladder model}

Consider the \(\mathbb{Z}_p\) quantum double model, where \(p\) is a prime number.\footnote{We consider this case since there are only two possible topological boundary conditions; the discussion can be generalized to a general Abelian group straightforwardly.}
The topological charges are given by irreducible representations of \( D(\mathbb{Z}_p) \):  
\begin{equation}
    a_{[g],\pi_h} = e^g \otimes m^h=:\varepsilon^{g,h}, g,h\in \Zbb_p.
\end{equation}  
Since \( \mathbb{Z}_p \) is Abelian, \( [g] \) consists of a single element \( \{g\} \) and \( C_{\mathbb{Z}_p}([g]) = \mathbb{Z}_p \). The irreducible representations of the Abelian group are also labeled by elements \( h \in \mathbb{Z}_p \).
The fusion rule reads
\begin{equation}
\varepsilon^{g,h}\otimes \varepsilon^{k,l}=  (e^g \otimes m^h) \otimes (e^{k} \otimes m^{l})= e^{g+k}\otimes m^{h+l}=\varepsilon^{g+k,h+l},
\end{equation}
namely, they form $\Zbb_p\times \Zbb_p$ fusion algebra. In a good choice of gauge, the fusion F-symbols are trivial (i.e., they equal to 1 when allowed by fusion), and the braiding R-symbols read
\begin{equation}
    R_{\varepsilon^{g,h},\varepsilon^{k,l}}=\exp(\frac{2\pi i}{p} gl).
\end{equation}

There are two possible topological boundary conditions of the \(D(\mathbb{Z}_p)\) quantum double model:
\begin{itemize}
    \item \emph{Rough boundary} (\emph{Dirichlet boundary}): This corresponds to the trivial subgroup \(M_r = \{0\}\)\footnote{Here we use \(0\) to denote the identity of the group.}. 
The Lagrangian algebra is \(\mathcal{A}_r = \mathbb{1} \oplus e \oplus \cdots \oplus e^{p-1}\) meaning all electric charges condense on this boundary. The boundary excitation is given by the unitary fusion category \(\Vect_{\mathbb{Z}_p} \simeq \Rep(\Cbb[\Zbb_p]^{\vee})\), where $\Cbb[\Zbb_p]^{\vee}$ denotes the dual algebra of $\Cbb[\Zbb_p]$.
    \item \emph{Smooth boundary} (\emph{Neumann boundary}): This corresponds to the subgroup \(N_s = \mathbb{Z}_p\). The Lagrangian algebra is \(\mathcal{A}_s = \mathbb{1} \oplus m \oplus \cdots \oplus m^{p-1}\), meaning all magnetic fluxes condense on this boundary. The boundary excitation is given by \(\Rep(\mathbb{Z}_p)\simeq \Rep(\Cbb[\Zbb_p])\).
\end{itemize}
In both cases, the associated 2-cocycles are trivial. The triviality of the 2-cocycle can be understood as follows:  
the case of \( H^2(\mathbb{Z}_1, \mathbb{C}^\times) = 0 \) is straightforwardly trivial. For \( H^2(\mathbb{Z}_p, \mathbb{C}^\times) \), we use the relation  
\(
H^n(G, \mathbb{Z}) \simeq H^{n-1}(G, \mathbb{Q}/\mathbb{Z}) \simeq H^{n-1}(G, \mathbb{C}^\times), \quad \text{for all } n \geq 2,
\)  
along with the fact that \( H^3(\mathbb{Z}_p, \mathbb{Z}) = 0 \).

For the Abelian group \(\mathbb{Z}_p\), the Weyl-Heisenberg \(X\) operator is given by
\[
X_g = \sum_{h \in G} |h+g\rangle \langle h|,
\]
where \(X_g = \XR_g\) and \(X_g^{\dagger} = \XL_g\). Recall that all irreducible representations (irreps) of \(\mathbb{Z}_p\) are one-dimensional, and \(\chi_{\Gamma}\) is a complex character function, with \(\Gamma_k\) labeled by a group element \(k \in \mathbb{Z}_p\). In this way, \(\chi_k(h) = \exp(2\pi i k h / p)\).  
The Weyl-Heisenberg \(Z\) operator is defined as  
\[
Z_k = \sum_{h \in G} \chi_k(h) |h\rangle \langle h|,
\]
where \(Z_k = \ZR_{\Gamma_k}\) and \(Z_k^{\dagger} = \ZL_{\Gamma_k}\).

The face operator in this case can be simplified as  $\Bf_f^{\chi_{\Gamma}}= \ZR_{k}\otimes \ZL_{k} \otimes \ZL_{k}$:
\begin{equation}\label{eq:Bf}
	\Bf_f^{\chi_{\Gamma}}
	\big{|}	\begin{aligned}
		\begin{tikzpicture}
                    \fill[green!20] (-0.5, -0.5) rectangle ++(1,1); % Rectangle 
	%	\draw[-stealth, line width=1.0pt,red, midway] (-0.5,0.5) -- (0.5,0.5);
		\draw[-stealth,black] (-0.5,-0.5) -- (-0.5,0.5); 
		\draw[-stealth,black] (0.5,-0.5) -- (0.5,0.5); 
		\draw[-stealth, line width=1.0pt,blue, midway] (-0.5,-0.5) -- (0.5,-0.5); 
			\draw [fill = black] (0,0) circle (1.2pt);
			\node[ line width=0.2pt, dashed, draw opacity=0.5] (a) at (0.75,0){$x_3$};
			\node[ line width=0.2pt, dashed, draw opacity=0.5] (a) at (-0.75,0){$x_1$};
			\node[ line width=0.2pt, dashed, draw opacity=0.5] (a) at (0,-0.7){$x_2$};
			%\node[ line width=0.2pt, dashed, draw opacity=0.5] (a) at (0,0.7){$x_2$};
		\end{tikzpicture}
	\end{aligned}   \big{ \rangle}     
	= 
	\chi_{\Gamma}(x_1x_2^{-1}x_3^{-1})
	\big{|}	\begin{aligned}
	\begin{tikzpicture}
            \fill[green!20] (-0.5, -0.5) rectangle ++(1,1); % Rectangle 
	%	\draw[-stealth, line width=1.0pt,red, midway] (-0.5,0.5) -- (0.5,0.5);
		\draw[-latex,black] (-0.5,-0.5) -- (-0.5,0.5); 
		\draw[-latex,black] (0.5,-0.5) -- (0.5,0.5); 
		\draw[-stealth, line width=1.0pt,blue, midway] (-0.5,-0.5) -- (0.5,-0.5); 
		\draw [fill = black] (0,0) circle (1.2pt);
		\node[ line width=0.2pt, dashed, draw opacity=0.5] (a) at (0.75,0){$x_3$};
		\node[ line width=0.2pt, dashed, draw opacity=0.5] (a) at (-0.75,0){$x_1$};
		\node[ line width=0.2pt, dashed, draw opacity=0.5] (a) at (0,-0.7){$x_2$};
		\node[ line width=0.2pt, dashed, draw opacity=0.5] (a) at (0,0.7){$x_2$};
	\end{tikzpicture}
\end{aligned}   \big{ \rangle} .
\end{equation}
Thus, the face operator  is given by
\[
\Bf_f = \frac{1}{p} \sum_{k \in \mathbb{Z}_p} \ZR_k \otimes \ZL_k \otimes \ZL_k.
\]
The symmetry boundary operator becomes
\[
\Av_{v_s} = \frac{1}{p} \sum_{g \in \mathbb{Z}_p} \XL_g \otimes \XL_g \otimes \XR_g.
\]
This model is nothing but the cluster state model \cite{jia2024generalized, fechisin2023noninvertible, brell2015generalized}.
The symmetry of the model is \(\Rep(\mathbb{Z}_p) \times \Vect_{\mathbb{Z}_p}\), with the fusion algebra \(\operatorname{Gr}(\Rep(\mathbb{Z}_p)) = \operatorname{Gr}(\Vect_{\mathbb{Z}_p}) = \mathbb{Z}_p \) (we use $\operatorname{Gr}$ to denote Grothendieck group), which gives the usual \(\mathbb{Z}_p \times \mathbb{Z}_p\) cluster state model.

\subsection{$S_3$ cluster ladder model}

For non-Abelian group $S_3= \langle \sigma, \tau: \sigma^2=\tau^3=1, \sigma \tau = \tau^{-1}\sigma \rangle $ (we take $\sigma=(12)$ and $\tau=(123)$), the corresponding quantum double model $D(S_3)$ has eight topological charges described in the following table:
\begin{align*}
\begin{array}{|c|ccc|cc|ccc|}
\hline
      & A & B & C & D & E  & F & G & H   \\
\hline
\text{conjugacy class} ~[g] &  [1] &  [1] &  [1] &  [\sigma]  & [\sigma] &  [\tau] &  [\tau]  &  [\tau]    \\
 \text{irrep of the centralizer}     & \mathbb{1} & \mathrm{sign} & \pi & \omega_2^0 & \omega_2^1 &  \omega_3^0 & \omega_3^1 & \omega_3^2 \\
\hline
\end{array}
\end{align*}
Notice $C_{S_3}([1])=S_3$, we use $\mathbb{1}$ to denote the trivial irrep, $\mathrm{sign}$ to denote the sign representation, and $\pi$ to denote the two-dimensional representation of $S_3$.
For \( C_{S_3}([\sigma]) = \mathbb{Z}_2 \), we use \(\omega_2^0\) and \(\omega_2^1\) to denote the two irreducible representations (irreps) of \(\mathbb{Z}_2\) (\(\omega_2\) is the second root of unity).  
For \( C_{S_3}([\tau]) = \mathbb{Z}_3 \), we use \(\omega_3^k\), where \(k = 0, 1, 2\), to denote its three irreps (\(\omega_3\) is the third root of unity). Their fusion rules are given in Table~\ref{tab:DS3-fusion}  \cite{Beigi2011the,Cui2015universal}.

\begin{table}
\resizebox{\textwidth}{!}{
\begin{tabular}{|c|c|c|c|c|c|c|c|c|}
\hline $\otimes$ &$A$ &$B$ &$C$ &$D$ &$E$ &$F$ &$G$ &$H$\\ \hline
$A$ &$A$ &$B$ &$C$ &$D$ &$E$& $F$ &$G$ &$H$\\ \hline
$B$ &$B$ &$A$ &$C$& $E$ &$D$ &$F$ &$G$ &$H$\\ \hline
$C$ &$C$ &$C$ &$A\oplus B\oplus C$& $D\oplus E$ &$D\oplus E$ & $G\oplus H$& $F\oplus H$ &$F\oplus G$\\ \hline
\multirow{2}{*}{$D$} &\multirow{2}{*}{$D$} &\multirow{2}{*}{$E$} &\multirow{2}{*}{$D\oplus E$}& $A\oplus C\oplus  F$ & $B\oplus C\oplus F$ & \multirow{2}{*}{$D\oplus E$} & \multirow{2}{*}{$D\oplus E$} & \multirow{2}{*}{$D\oplus E$} \\
& & & & $\oplus G\oplus H$ & $\oplus G\oplus H$ & & &  \\ \hline
\multirow{2}{*}{$E$} &\multirow{2}{*}{$E$}& \multirow{2}{*}{$D$}& \multirow{2}{*}{$D\oplus E$} & $B\oplus C\oplus F$ & $A\oplus C\oplus F$ & \multirow{2}{*}{$D\oplus E$} &\multirow{2}{*}{$D\oplus E$} & \multirow{2}{*}{$D\oplus E$} \\
& & & & $\oplus G\oplus H$ & $\oplus G\oplus H$ & & &  \\  \hline
$F$ &$F$ & $F$& $G\oplus H$& $D\oplus E$ & $D\oplus E$ & $A\oplus B\oplus F$ & $H\oplus C$ & $G\oplus C$ \\ \hline
$G$ &$G$ & $G$& $F\oplus H$ & $D\oplus E$ & $D\oplus E$ & $H\oplus C$ & $A\oplus B\oplus G$ & $F\oplus C$ \\ \hline
$H$ &$H$ & $H$& $F\oplus G$ & $D\oplus E$ & $D\oplus E$ & $G\oplus C$ & $F\oplus C$ & $A\oplus B\oplus H$\\\hline
\end{tabular}
}
\caption{Fusion rules of $D(S_3)$ quantum double model.}\label{tab:DS3-fusion}
\end{table}

The $S$ and $T$ matrices of $D(S_3)$ topological phase are of the form:
\begin{equation}
\label{eq:ds3-S}
S = \frac{1}{6}
\begin{bmatrix}
1 & 1 & 2 & 3 & 3 & 2 & 2 & 2 \\
1 & 1 & 2 & -3 & -3 & 2 & 2 & 2 \\
2 & 2 & 4 & 0 & 0 & -2 & -2 & -2 \\
3 & -3 & 0 & 3 & -3 & 0 & 0 & 0 \\
3 & -3 & 0 & -3 & 3 & 0 & 0 & 0 \\
2 & 2 & -2 & 0 & 0 & 4 & -2 & -2 \\
2 & 2 & -2 & 0 & 0 & -2 & -2 & 4 \\
2 & 2 & -2 & 0 & 0 & -2 & 4 & -2 \\
\end{bmatrix},
\end{equation}
\begin{equation}
\label{eq:ds3-T}
{T} = \text{diag}(1,1,1,\omega_2^0,\omega_2^1,\omega_3^0,\omega_3^1,\omega_3^2).
\end{equation}
Notice that the first row of the \( S \)-matrix represents the quantum dimensions of all topological charges, while the \( T \)-matrix encodes the topological spins of these charges. The fusion rule \( N_{XY}^Z \) can be determined from the \( S \)-matrix using the Verlinde formula:  
\begin{align}\label{eq:verlinde}
    N_{XY}^{Z} = \sum_{U \in \Irr (D(G))} \frac{S_{XU} S_{YU} S_{ZU}^{\ast}}{S_{\mathbb{1}U}},
\end{align}  
where the summation runs over all irreducible representations \( U \), and \( \mathbb{1} \) denotes the trivial charge.
For more details of modular data, see Refs.~\cite{bakalov2001lectures,wang2010topological,etingof2016tensor}.

There are four types of topological boundary conditions \cite{Cong2017}:
\begin{itemize}
    \item \emph{Rough boundary}: This corresponds to the trivial subgroup \( M_r = \{1\} \). The Lagrangian algebra is \( \mathcal{A}_r = A \oplus B \oplus 2C \). The boundary excitation is given by \( \Vect_{S_3} \).
    \item \emph{Smooth boundary}: This corresponds to the trivial subgroup \( M_s = S_3 \). The Lagrangian algebra is \( \cA_s = A \oplus D \oplus F \). The boundary excitation is given by \( \Rep(S_3) \).
    \item \emph{\( \Zbb_2 \) boundary}: This corresponds to the subgroup \( \Zbb_2 \) of \( S_3 \). The Lagrangian algebra is \( \cA_2 = A \oplus C \oplus D \).
    \item \emph{\( \Zbb_3 \) boundary}: This corresponds to the subgroup \( \Zbb_3 \) of \( S_3 \). The Lagrangian algebra is \( \cA_3 = A \oplus B \oplus 2F \).
\end{itemize}

For the cluster ladder lattice, if there is no anyonic tunneling channel between the two boundaries, the ground state degeneracy of the model is 1. Such a model is a candidate for a non-invertible SPT phase. This includes the following models:
\begin{enumerate}
    \item \emph{The smooth and rough boundary pair:} $\Hbb[\cA_s,\cA_r]$. This corresponds to the $S_3$ cluster state model. The symmetry is given by $\Rep(S_3) \times \Vect_{S_3}$.
    
    \item \emph{The $\Zbb_2$ and $\Zbb_3$ boundary pair:} $\Hbb[\cA_2,\cA_3]$. This model is entirely different from the cluster state model. Its symmetry is given by $\Rep(D(S_3))_{\cA_2} \times \Rep(D(S_3))_{\cA_3}$, where we use the notation $\Rep(D(S_3))_{\cA}$ to denote the category of $\cA$-modules in $\Rep(D(S_3))$. Equivalently, the symmetry can be described as ${^{\Cbb[S_3]}_K}\Mod_K \times {_J}\Mod^{\Cbb[S_3]}_J$, where $K = \Cbb[\Zbb_2]$ and $J = \Cbb[\Zbb_3]$.
\end{enumerate}
The lattice model with $\Rep(S_3)$ symmetry has been discussed carefully based on the anyonic chain in Ref.~\cite{bhardwaj2024lattice}. Our model here can be regarded as an alternative lattice realization.

For other pairings, we also provide the corresponding lattice Hamiltonian, e.g., $\Hbb[\cA_r,\cA_2]$, $\Hbb[\cA_s,\cA_3]$, etc. Due to the existence of anyonic tunneling channels between the two boundaries, the symmetry is either fully or partially broken.

% duality!!!!!!!!!!

\section{Weak Hopf non-invertible symmetry}
\label{sec:WHAsymmetry}
From a modern perspective, symmetries can be characterized by the algebraic structure of topological defects~\cite{gaiotto2015generalized,inamura2021topological,Kong2020algebraic,SchaferNameki2024ICTP,huang2023topologicalholo,bhardwaj2024lattice,freed2024topSymTFT,gaiotto2021orbifold,bhardwaj2023generalizedcharge,apruzzi2023symmetry,bhardwaj2024gappedphases,Zhang2024anomaly,Ji2020categoricalsym,chatterjee2024TopHolo,chang2019topological,bhardwaj2018finite}.
Symmetries associated with non-invertible topological defects are known as non-invertible symmetries. For (1+1)D system, the symmetry is described by a fusion $1$-category, viz., the usual fusion category, thus it's also called fusion category symmetry \cite{bhardwaj2018finite,chang2019topological,thorngren2019fusion}.
Since any unitary (multi)fusion category is equivalent to the representation category $\Rep(H)$ for some weak Hopf algebra $H$, it is natural for us to investigate the symmetry in the framework of Hopf and weak Hopf algebras~\cite{Kitaev2003,bais2003hopf,meusburger2017kitaev,Buerschaper2013a,girelli2021semidual,jia2023boundary,Jia2023weak,jia2024weakTube,jia2024generalized}. These symmetries, characterized by quantum algebras such as Hopf, quasi-Hopf and weak Hopf algebras, as well as module and comodule algebras over (weak, quasi-) Hopf algebras, constitute a significant class of non-invertible symmetries~\cite{cordova2022snowmass,brennan2023introduction,mcgreevy2023generalized,luo2023lecture,shao2024whats,SchaferNameki2024ICTP,Bhardwaj2024lecture}.

As we will see, in many cases, by fusion category symmetry \(\eC \simeq \Rep(H)\), people usually refer to the fusion algebra symmetry, whose algebraic structure is the fusion ring (Grothendieck ring) \(\operatorname{Gr}(\eC)\) of the category \(\eC\) (or more generally, \(R(\eC) = \operatorname{Gr}(\eC) \otimes_{\Zbb} \Cbb\)). This fusion ring can be naturally embedded into the dual weak Hopf algebra \(\hat{H}\). This suggests that weak Hopf symmetry provides a comprehensive framework for understanding non-invertible fusion algebra symmetries in (1+1)D systems.

\subsection{Weak Hopf algebra symmetry and weak Hopf comodule algebra symmetry}

Here, we will review the definitions of various quantum algebras used in this work and discuss the symmetries characterized by these algebras. For further details on quantum algebras, see, e.g., Refs.~\cite{majid2000foundations,abe2004hopf,montgomery1993hopf,kassel2012quantum}.
Notice that, unless otherwise specified, all algebras and categories in this work are over the complex number field \(\mathbb{C}\).

\subsubsection{Hopf symmetry}
A complex Hopf algebra is a complex vector space $H$ equipped with the five structure linear morphisms: multiplication $\mu: H \otimes H \to H$, unit $\eta: \mathbb{C} \to H$, comultiplication $\Delta: H \to H \otimes H$, counit $\varepsilon: H \to \mathbb{C}$, and antipode $S: H \to H$. These maps satisfy the following conditions:
    \begin{enumerate}
        \item $(H, \mu, \eta)$ is an algebra, meaning $\mu \circ (\mu \otimes \id) = \mu \circ (\id \otimes \mu)$ and $\mu \circ (\eta \otimes \id) = \id = \mu \circ (\id \otimes \eta)$.
        
        \item $(H, \Delta, \varepsilon)$ is a coalgebra, meaning $(\Delta \otimes \id) \circ \Delta = (\id \otimes \Delta) \circ \Delta$ and $(\varepsilon \otimes \id) \circ \Delta = \id = (\id \otimes \varepsilon) \circ \Delta$.
        
        \item $(H, \mu, \eta, \Delta, \varepsilon)$ forms a bialgebra, where $\Delta$ and $\varepsilon$ are algebra homomorphisms, and equivalently, $\mu$ and $\eta$ are coalgebra homomorphisms.
        
        \item The antipode $S$ satisfies $\mu \circ (S \otimes \id) \circ \Delta = \eta \circ \varepsilon = \mu \circ (\id \otimes S) \circ \Delta$.
    \end{enumerate}
We will use the notation $\mu(x \otimes y) = xy$ and $\eta(1) = 1_H$ for multiplication and unit. We will also adopt the Sweedler's notation $\Delta(u) = \sum_{(u)} u^{(1)} \otimes u^{(2)} = \sum_{i} u_i^{(1)} \otimes u_i^{(2)}$. The comultiplication law ensures that $(\Delta \otimes \id) \circ \Delta(u) = (\id \otimes \Delta) \circ \Delta(u) = \sum_{(u)} u^{(1)} \otimes u^{(2)} \otimes u^{(3)}$. 
In general, we define $\Delta_1 = \Delta$ and $\Delta_n = (\id \otimes \cdots \otimes \id \otimes \Delta) \circ \Delta_{n-1}$, then $\Delta_n(u) = \sum_{(u)} u^{(1)} \otimes \cdots \otimes u^{(n+1)}$. From the definition above, we have the following useful identity:
\begin{align}
    \sum_{(x)} \varepsilon(x^{(1)}) x^{(2)} = x = \sum_{(x)} x^{(1)} \varepsilon(x^{(2)}).
\end{align}
The swap operation is defined as $\tau(x \otimes y) = y \otimes x$, by which one denotes $\mu^{\rm op} = \mu \circ \tau$ and $\Delta^{\rm op} = \tau \circ \Delta$. The opposite Hopf algebra $H^{\rm op}$ is defined as $(H, \mu^{\rm op}, \eta, \Delta, \varepsilon, S^{-1})$, and the coopposite weak Hopf algebra $H^{\rm cop}$ is defined as $(H, \mu, \eta, \Delta^{\rm op}, \varepsilon, S^{-1})$. For more details about our notation and convention, see~\cite{jia2023boundary, Jia2023weak, jia2024generalized}.

\begin{definition}[Hopf symmetry]
    A quantum system \((\mathbb{H}, \mathcal{H})\) is said to exhibit Hopf symmetry if and only if the Hilbert space \(\mathcal{H}\) supports a representation of a Hopf algebra, given by \(W_{\bullet}: H \to \operatorname{End}(\mathcal{H})\), and satisfies the following conditions:  
    (i) \([W_h, \mathbb{H}] = 0\) for all \(h \in H\);  
    (ii) all eigenspaces \(\mathcal{V}\) of \(\mathbb{H}\) are invariant under the \(H\)-action.  
    This symmetry will also be referred to as Hamiltonian symmetry to distinguish it from state symmetry.
\end{definition}

In certain cases (especially when discussing weak Hopf symmetry), it is useful to discuss symmetry within the vacuum space. Specifically, this refers to an operator \(W_h\) that commutes with the Hamiltonian \(\mathbb{H}\) only within the ground state space, i.e.,  
\[
[W_h, \mathbb{H}]_{\mathcal{V}_{\mathrm{GS}}} = 0, \quad \forall h \in H.
\]  
The ground state symmetry is generally larger than the Hamiltonian symmetry defined above.

All group symmetries are special cases of Hopf symmetries. A key example of a non-group Hopf algebra symmetry is quantum double symmetry $D(G)$.
Consider a $(2+1)$D discrete gauge theory, such as Gauge-Higgs theory, where a Higgs field breaks the continuous gauge group $U$ to a finite group $G$~\cite{Propitius1995, bais1992quantum}. In the low-energy or long-distance limit, the theory becomes topological, with excitations characterized by fusion and braiding data. These data are captured by the representation category $\mathsf{Rep}(D(G))$, corresponding to the quantum double of $G$.
This formalism is generalized in \cite{bais2003hopf}, where the unbroken symmetry is described by a Hopf algebra $H$, of which the group algebra is a special case. 
Hopf algebraic symmetry has broad applications in quantum gravity~\cite{Bais2002quantum, majid2000foundations, delcamp2017fusion}, conformal field theory (CFT)~\cite{fuchs1995affine}, topological quantum field theory (TQFT)~\cite{meusburger2021hopf, meusburger2017kitaev}, and the quantum Hall effect~\cite{Slingerland2001quantum}. 
The Hopf quantum double model provides a general framework for realizing Hopf symmetry~\cite{Kitaev2003, Buerschaper2013a, buerschaper2013electric, meusburger2017kitaev, meusburger2021hopf, koppen2020defects, girelli2021semidual, voss2021defects, chen2021ribbon, jia2023boundary}, while certain string-net models exhibit Hopf symmetry~\cite{Kitaev2012boundary, jia2024weakTube}, among other examples.

\paragraph{Invertible Hopf symmetry} 
A special type of Hopf symmetry is the conventional invertible group symmetry, where \( H = \mathbb{C}[G] \). The group elements form a basis, with multiplication induced by group multiplication. The comultiplication and counit are defined as \( \Delta(g) = g \otimes g \) and \( \varepsilon(g) = 1 \) for all basis elements \( g \in G \). The antipode is given by \( S(g) = g^{-1} \) for each \( g \in G \), and the \( * \)-operation (to be defined later for more general weak Hopf algebras) is specified by \( g^* = g^{-1}\).

Recall that for invertible symmetry $G$, $U_g^{\dagger}=U_{g^*}=U_{g^{-1}}$. $U_g$'s commute with Hamiltonian $\mathbb{H}$ for all $g\in G$. The total Hilbert space supports a representation of $G$, and all eigenspace of $\mathbb{H}$ is invariant under the $G$-action. When regarded as a group algebra symmetry $\Cbb[G]$, we say $(\mathbb{H}, \mathcal{H})$ possesses a $\Cbb[G]$ symmetry means that (i) $[U_x, \mathbb{H}]=0$ for all $x\in \Cbb[G]$; (ii) any eigenspace $\mathcal{V}$ of $\mathbb{H}$ is invariant under the $G$-action, note that this a result of condition (i).

A natural question now arises: under what circumstances is a Hopf symmetry invertible? We provide the following formal definition:

\begin{definition}
    A (weak) Hopf symmetry $H$ is called \emph{invertible} if there exists a group $G$ such that $H \cong \Cbb[G]$ as Hopf algebras.
\end{definition}

Notice that the elements in \(G\) are all invertible, but the elements in \(\mathbb{C}[G]\) are generally not invertible.  
Determining whether a given Hopf symmetry is invertible is generally a challenging task. It is clear that group algebra symmetries must be cocommutative.  
For pointed cocommutative Hopf symmetries\footnote{A Hopf algebra \(H\) over a field \(\mathbb{F}\) is called pointed if all its simple left or right comodules are one-dimensional. See Ref.~\cite{andruskiewitsch2010classification} for more details.}, a useful result is the Cartier-Konstant-Milnor-Moore theorem~\cite{milnor1965structure,cartier1962groupes,sweedler1969hopf,etingof2016tensor}.

For a (weak) Hopf algebra \( H \), an element \( g \in H \) is called a group-like element if and only if \( \Delta(g) = g \otimes g \) and \( \varepsilon(g) = 1 \). The set of all group-like elements of \( H \) is denoted by \( \mathcal{G}(H) \); it forms a group when the antipode \( S \) is invertible. If $x\in H$ satisfies $\Delta(x)=x\otimes 1+1\otimes x$, we call $x$ a primitive element. If $x$ is primitive, then we have $\varepsilon(x)=0$. The set of all primitive element of $H$ is denoted as $\mathcal{P}(H)$, and  $\mathcal{P}(H)$ is a Lie algebra with $[x,y]:=xy-yx$. For group algebra $\Cbb[G]$, it is clear that $\mathcal{P}(\Cbb[G])=\{0\}$ and $\mathcal{G}(\Cbb[G])=G$.

The Cartier-Konstant-Milnor-Moore theorem~\cite{milnor1965structure,cartier1962groupes,sweedler1969hopf,etingof2016tensor} claims that, for pointed cocommutative Hopf algebras $H$, we have 
   \begin{equation}
       H\cong U(\mathcal{P}(H)) \# \Cbb[\mathcal{G}(H)],
   \end{equation}
   where $\mathcal{P}(H)$ is the set of primitive elements of $H$, $U(\mathcal{P}(H)) $ is the  universal enveloping algebra of $\mathcal{P}(H)$, and $\mathcal{G}(H)$ is the set of group elements in $H$; `\#' represents smash product (see, e.g., \cite{montgomery1993hopf}). 

Since \( \mathcal{P}(H) = 0 \) for a finite-dimensional complex Hopf algebra \( H \) \cite[Exercise 4.2.16]{dascalescu2000hopf}, a direct consequence of the Cartier–Kostant–Milnor–Moore theorem is that for finite-dimensional pointed cocommutative Hopf algebras, we have \( H \cong \mathbb{C}[\mathcal{G}(H)] \). Therefore, all finite-dimensional  pointed cocommutative Hopf symmetries are invertible symmetries.

\paragraph{Comodule algebra symmetry}
As shown in Refs.~\cite{jia2023boundary,Jia2023weak,jia2024weakTube}, (weak) Hopf symmetry characterizes 2d topologically ordered phases, while 1d gapped boundaries exhibit comodule algebra symmetry (or equivalently, module algebra symmetry). These symmetries are fundamental to our construction of non-invertible SPT lattice models. We now review their definitions.

Let $H$ be a Hopf algebra, and let $(K, \mu_K, \eta_K)$ be an algebra. If $K$ is a left $H$-comodule with coaction $\beta_K: K \to H \otimes K$ \footnote{This means that $(\id \otimes \beta) \circ \beta = (\Delta \otimes \id )\circ \beta $ and $(\varepsilon \otimes \id)\circ \beta =\id$.}, satisfying the conditions:
\begin{equation}
    \beta_K(xy) = \beta_K(x) \beta_K(y), \quad \beta_K(1_K) = 1_H \otimes 1_K,
\end{equation}
then $K$ is called a \emph{left $H$-comodule algebra}. A \emph{right $H$-comodule algebra} can be defined analogously.
For left and right comodules, we adopt Sweedler's notation for left and right coactions as follows:
\begin{align}
    \beta(x) = \sum_{(x)} x^{[-1]} \otimes x^{[0]}, 
    \beta(x) = \sum_{(x)} x^{[0]} \otimes x^{[1]}.
\end{align}

Let $H$ be a Hopf algebra and $M$ an algebra. If $M$ is a left $H$-module and satisfies the following conditions:
\begin{equation}
    h \triangleright (xy) = \sum_{(h)} (h^{(1)} \triangleright x)(h^{(2)} \triangleright y), \quad h \triangleright 1_M = \varepsilon(h) 1_M,
\end{equation}
for all $h \in H$ and $x, y \in M$, then $M$ is called a \emph{left $H$-module algebra}. A \emph{right $H$-module algebra} is defined similarly.

From a macroscopic perspective, the description of the 1d gapped boundary of the quantum double model in terms of comodule algebra and module algebra is equivalent. Henceforth, following the convention in Refs.~\cite{jia2023boundary,Jia2023weak,jia2024weakTube}, we adopt the comodule algebra description.

\subsubsection{Weak Hopf symmetries}

It is well-known that not all fusion categories can be regarded as the representation category of a Hopf algebra. Typical examples include the Haagerup category and \(\Vect_G^{\omega}\) (the category of \(G\)-graded vector spaces with a non-trivial 3-cocycle \(\omega\)). To represent an arbitrary fusion category as a representation category, it is necessary to introduce the concept of weak Hopf algebras~\cite{szlachanyi2000finite,ostrik2003module,etingof2005fusion}.

A complex weak Hopf algebra~\cite{BOHM1998weak} is a complex vector space equipped with the same structural morphisms as a Hopf algebra, but with modified compatibility conditions:

\begin{enumerate}

\item The triple \((H,\mu,\eta)\) is an algebra,  and the triple $(H,\Delta,\varepsilon)$ forms a coalgebra.

\item The quintuple $(H, \mu, \eta, \Delta, \varepsilon)$ is a weak bialgebra, meaning that the algebra structure and coalgebra structure are compatible in the following way:

(i) \emph{Multiplicativity of comultiplication:} The comultiplication map satisfies \(\Delta(xy) = \Delta(x) \cdot \Delta(y)\), explicitly,  
\begin{equation}
\sum_{(x \cdot y)} (x \cdot y)^{\cone} \otimes (x \cdot y)^{\ctwo} = \sum_{(x), (y)} x^{\cone} \cdot y^{\cone} \otimes x^{\ctwo} \cdot y^{\ctwo}.  
\end{equation}
This property holds for Hopf algebras as well.

(ii) \emph{Weak comultiplicativity of the unit:} The comultiplication of the identity element satisfies  
\begin{equation}
(\Delta \otimes \id) \circ \Delta(1_H) = (\Delta(1_H) \otimes 1_H) \cdot (1_H \otimes \Delta(1_H)) = (1_H \otimes \Delta(1_H)) \cdot (\Delta(1_H) \otimes 1_H).
\end{equation}  
This indicates that the comultiplication of the identity is not factorizable but instead takes the form of a linear combination:  
\begin{equation}
\Delta(1_H) = \sum_{(1_H)} 1_H^{\cone} \otimes 1_H^{\ctwo}.
\end{equation}  
In the case of Hopf algebras, we have the stronger condition  
\begin{equation}
\Delta(1_H) = 1_H \otimes 1_H.
\end{equation}

(iii) \emph{Weak multiplicativity of the counit:} For all \(x, y, z \in H\), the counit satisfies  
\begin{equation}
\varepsilon(x \cdot y \cdot z) = \sum_{(y)} \varepsilon(x \cdot y^{\cone}) \varepsilon(y^{\ctwo} \cdot z) = \sum_{(y)} \varepsilon(x \cdot y^{\ctwo}) \varepsilon(y^{\cone} \cdot z).   
\end{equation}
For Hopf algebras, since \(\Delta(1_H) = 1_H \otimes 1_H\), setting \(y = 1\) in the above equation gives the stronger constraint \(\varepsilon(xz) = \varepsilon(x)\varepsilon(z)\).

\item The antipode \(S\) plays the role of taking the inverse of elements, and its axioms differ significantly from those of Hopf algebras. It satisfies the following properties:

(i)  \emph{Left counit:} For all $x$ in $H$, we have 
\begin{equation}    \sum_{(x)}x^{(1)}S(x^{\ctwo})=\varepsilon_L(x):=\sum_{(1_H)}\varepsilon(1_H^{\cone}x)1_H^{\ctwo}.
\end{equation}
For Hopf algebras, we have the stronger relation \( \sum_{(x)} x^{(1)} S(x^{\ctwo}) = \varepsilon_L(x) = \varepsilon(x) 1_H \).

(ii)  \emph{Right counit:} For all $x$ in $H$, we have 
\begin{equation}    \sum_{(x)}S(x^{(1)})x^{\ctwo}=\varepsilon_R(x):=\sum_{(1_H)}1_H^{\cone}\varepsilon(x1_H^{\ctwo}).
\end{equation}
For Hopf algebra, we have the stronger relation $\sum_{(x)}S(x^{(1)})x^{\ctwo}=\varepsilon_R(x)=\varepsilon(x)1_H$.

(iii)  \emph{Antipode decomposition:} For all $x$ in $H$, we have 
\begin{equation}
    \sum_{(x)}S(x^{\cone})x^{\ctwo}S(x^{\cthree})=S(x).
\end{equation}
This condition also holds for Hopf algebras.
\end{enumerate}

Notice that the images of the left and right counit maps are denoted as \( H_L = \varepsilon_L(H) \) and \( H_R = \varepsilon_R(H) \). They serve as the tensor unit (vacuum charge or vacuum sector) of the representation category \( \Rep(H) \) of \( H \). In general, the tensor unit is not simple, and the representation category \( \Rep(H) \) is a multifusion category. When \( H \) is connected \cite{etingof2005fusion} (also called pure in Ref.~\cite{BOHM1998weak}), \( \Rep(H) \) is a fusion category.
If the dual weak Hopf algebra \( \hat{H} \) is connected, we call \( H \) coconnected \cite{Nikshych2004semisimpleWHA}. In this work, we primarily consider connected and coconnected weak Hopf algebras, although all discussions can be straightforwardly generalized.

%\begin{figure}[t]
 %   \centering
  %  \includegraphics[width=8cm]{FigHopfSymmetry.pdf}
   % \caption{An illustration of finite symmetry }
    %\label{fig:cluster-graph}
%\end{figure}

We will also assume that the weak Hopf algebra is a \( C^* \)-weak Hopf algebra. A complex weak Hopf algebra \( H \) is called simple (or indecomposable) if its underlying algebra \( (H, \mu, \eta) \) has no nontrivial subalgebras, and it is called semisimple if its underlying algebra can be written as a direct sum of simple algebras. 
A \( * \)-weak Hopf algebra \( (H, *) \) is a weak Hopf algebra \( H \) equipped with a \( C^* \)-structure \( *: H \to H \) such that \( \Delta \) is a \( * \)-homomorphism. That is, for all \( x, y \in H \) and \( c \in \mathbb{C} \),
\begin{equation}
    (x^*)^* = x, \quad (x + y)^* = x^* + y^*, \quad (xy)^* = y^* x^*, \quad (c x)^* = \bar{c} x^*, \quad \Delta(x)^* = \Delta(x^*).
\end{equation}
We also have 
\begin{equation}
    S(x^*) = S^{-1}(x)^*.
\end{equation}
A \( C^* \)-weak Hopf algebra \( (H, *) \) is one for which there exists a fully faithful \( * \)-representation \( \rho: H \to \mathbf{B}(\mathcal{H}) \) for some operator space over a Hilbert space \( \mathcal{H} \).

The weak Hopf morphism between two weak Hopf algebras $H_1$ and $H_2$ is a map $f:H_1\to H_2$ which is both an algebra and coalgebra homomorphism preserving unit and counit and which intertwines the antipodes, $f\circ S_1= S_2 \circ f$. The category of all finite dimensional weak Hopf algebras will be denoted as $\mathsf{WeakHopfAlg}_{\Cbb}$.

A weak Kac algebra is a weak Hopf algebra such that \( S^2 = \id \) ~\cite{yamanouchi1994duality,bohm2000weakII}. All finite-dimensional semisimple Hopf algebras are weak Kac algebras due to the Larson-Radford theorem (See, e.g., Ref.~\cite{kassel2012quantum}).

\begin{definition}[Weak Hopf symmetry]   
A Hamiltonian \(\mathbb{H}: \mathcal{H} \to \mathcal{H}\) is said to possess weak Hopf symmetry \(H \) if \(\mathcal{H}\) supports a representation \(W_{\bullet}: H \to \operatorname{End}(\mathcal{H})\) of the weak Hopf algebra \(H\), such that \([W_h, \mathbb{H}] = 0\) for all \(h \in H\).
In some contexts, we consider symmetries that act only on the ground states, where \([\mathbb{H}, W_h] = 0\) holds (at least) within the ground state space \(\mathcal{V}_{\rm GS}\). 
In such cases, the symmetry of the ground state space must be larger than that of the Hamiltonian, i.e., \(\operatorname{Sym}_{\rm GS} \supseteq \operatorname{Sym}_{\mathbb{H}}\). 
If the ground state space is non-degenerate, meaning \(\dim \mathcal{V}_{\rm GS} = 1\), the ground state must be an eigenstate of the symmetry operator \(W_h\).
When the weak Hopf algebra is a weak Kac algebra, we refer to the corresponding symmetry as weak Kac symmetry.
\end{definition}

To summarize, the hierarchy of symmetries is as follows:
\begin{equation}
    \text{Invertible symmetry} \subset \text{Hopf symmetry} \subset \text{Weak Kac symmetry} \subset \text{Weak Hopf symmetry}.
\end{equation}

For weak Hopf symmetry, the notion of invariance under the symmetry action differs from that of groups and Hopf algebras \cite{Jia2023weak}. A state \(|\psi\rangle \in \mathcal{H}\) is called left-invariant under the action if and only if 
\begin{equation} \label{eq:invariance}
   W_h |\psi\rangle = W_{\varepsilon_L(h)} |\psi\rangle.
\end{equation}
The notion of right-invariance is defined analogously.
For a Hopf algebra, Eq.~\eqref{eq:invariance} reduces to 
\begin{equation}
   W_h |\psi\rangle = \varepsilon(h) |\psi\rangle, \quad \forall h \in H,
\end{equation}
which means that \(|\psi\rangle\) is an eigenstate of \(W_h\) with eigenvalue \(\varepsilon(h)\).
For a group algebra \(\mathbb{C}[G]\), where \(\varepsilon(g) = 1\) for all \(g \in G\), this further simplifies to the usual definition of invariance for group symmetries.

Notice that the above definition applies to a quantum mechanical system, which can be viewed as a $(0+1)$D quantum field theory. In this case, issues of locality do not arise. However, when considering lattice models in higher dimensions, locality becomes a crucial ingredient. A promising approach to encode locality in both quantum field theory and lattice models is via quasi-local algebras (a type of AF-algebra). Yet, for non-invertible symmetries—which are intrinsically non-local, in the sense that they act over long distances—the quasi-local algebra framework appears insufficient. This suggests that a new mathematical structure may be required in order to incorporate non-invertible symmetries and locality within the AQFT framework. The locality structure is also crucial for discussing anomalies of symmetry.
We will return to the locality structure in the context of our model later on.

\paragraph{Weak Hopf comodule algebra symmetry} 
For the topological boundary condition of $2d$ weak Hopf lattice gauge theory, we require the concept of a comodule algebra over weak Hopf algebras \cite{Jia2023weak}. The axioms for the comodule algebra need to be adjusted compared to the standard Hopf algebra case.

\begin{definition}\label{defintion:comoduleAlge}
Let \(H\) be a weak Hopf algebra, an algebra \( K \) is called a left \( H \)-comodule algebra \cite{Bohmdoihopf} if there exists a comodule map \( \beta: K \to H \otimes K \) such that 
\begin{equation}
    \beta(xy) = \beta(x) \beta(y),
\end{equation}
and 
\begin{equation} \label{eq:weakComodulealge}
    (1_H \otimes x) \beta(1_{K}) = (\varepsilon_R \otimes \id_K) \circ \beta(x),
\end{equation}
for all \( x, y \in K \). We adopt Sweedler's notation \( \beta(x) = \sum_{[x]} x^{[-1]} \otimes x^{[0]} \), \( \beta_2(x) = \sum_{[x]} x^{[-2]} \otimes x^{[-1]} \otimes x^{[0]} \), etc.\footnote{We use square brackets to denote the comodule structure to distinguish it from the comultiplication of the weak Hopf algebra.} The right \( H \)-comodule algebra can be defined similarly.
\end{definition}

\begin{remark}
In literature, there are several equivalent expressions of Eq.~\eqref{eq:weakComodulealge} (written more explicitly and omitting the sum symbol $1_K^{[-1]} \otimes x 1_K^{[0]} = \varepsilon_R(x^{[-1]}) \otimes x^{[0]}$):
\begin{enumerate}
    \item $\beta(1_K)\in H_R\otimes K$;
    \item $\beta(1_K)=(\varepsilon_R \otimes \id ) \circ \beta(1_K)$;
    \item $1_H^{(1)} \otimes 1_H^{(2)}1_K^{[-1]}\otimes 1_K^{[0]}=1_H^{(1)} \otimes 1_K^{[-1]} 1_H^{(2)}\otimes 1_K^{[0]}=1_K^{[-2]}\otimes 1_K^{[-1]}\otimes 1_K^{[0]}$, where we have omitted the sum symbol;
%    \item $1_K^{[-2]}\otimes \varepsilon_R(1_K^{[-1]}) \otimes  1_K^{[0]}=$
\end{enumerate}
For the convenience of the reader, we provide the proof here. 

The equivalence of 1 and 2 is obvious since \( \varepsilon_R^2 = \varepsilon_R \). For ``\( 1 \Rightarrow 3 \)'', note that \( 1^{(1)} \otimes 1^{(2)} \in H_R \otimes H_L \) and \( hg = gh \) for all \( h \in H_L \) and \( g \in H_R \). Using the fact that for all \( h \in H_R \) we have 
\begin{equation}
   \sum_{(1)} 1^{(1)} \otimes h 1^{(2)} = \Delta(h),
\end{equation}
we arrive at the required conclusion.

For ``\( 3 \Rightarrow \) Eq.~\eqref{eq:weakComodulealge}'', note that for \( x \in K \), we have 
\begin{equation}
    \beta(x1) = x^{[-1]} 1^{[-1]} \otimes x^{[0]} 1^{[0]} = x^{[-1]} \otimes x^{[0]}.
\end{equation}
Then 
\begin{align}
    \varepsilon_R(x^{[-1]}) \otimes x^{[0]} 
    &= 1_H^{(1)} \varepsilon(x^{[-1]} 1_H^{(2)}) \otimes x^{[0]} \notag \\
    &= 1_H^{(1)} \varepsilon(x^{[-1]} 1_K^{[-1]} 1_H^{(2)}) \otimes x^{[0]} 1_K^{[0]} \notag \\
    &= 1_K^{[-2]} \varepsilon(x^{[-1]} 1_K^{[-1]}) \otimes x^{[0]} 1_K^{[0]} \notag \\
    &= 1_K^{[-1]} \otimes x 1_K^{[0]}.
\end{align}

Finally, ``Eq.~\eqref{eq:weakComodulealge} \( \Rightarrow \) 1'' is obvious.
  \end{remark}
  
Note that \( H \) is a comodule algebra over itself with the comodule map given by the comultiplication. Since \( \Delta \) is an algebra map and 
\begin{equation}
    \Delta(1_H) \in H_R \otimes H_L,
\end{equation}
the comodule algebra conditions are satisfied.
The topological boundary condition can also be characterized via the module algebra.
An algebra \(A\) is called a \emph{left} \(H\)-module algebra if \(A\) is a left \(H\)-module such that \(h \triangleright (xy) = \sum_{(h)} (h^{(1)} \triangleright x)(h^{(2)} \triangleright y)\) and \(h \triangleright 1_A = \varepsilon_L(h) \triangleright 1_A\) for all \(h \in H\) and \(x \in A\). The \emph{right} \(H\)-module algebra can be defined similarly.

\subsection{Fusion category symmetry as dual weak Hopf symmetry}

Let us now consider how to understand (multi)fusion category symmetry in the framework of weak Hopf symmetries.
By definition, a fusion category symmetry is a fusion category $\eC$ that acts on a quantum system $\mathcal{T}$, like lattice system or quantum field theory, as a symmetry.
Treating the symmetry action as an algebra, the symmetry algebra for fusion category $\eC$ is the Grothendieck ring 
$\operatorname{Gr}(\eC)$. By definition, \( \operatorname{Gr}(\eC) \) is generated by the equivalence classes of simple objects in \( \eC \) as a free \( \mathbb{Z} \)-module. The multiplication is defined by the fusion rule:  
\begin{equation}
    [a] \cdot [b] := [a \otimes b] = \sum_{c \in \Irr(\eC)} N_{ab}^c [c], \quad N_{ab}^c \in \mathbb{Z}_{\geq 0}.
\end{equation}
This is also called the fusion ring, and it is a \( \mathbb{Z}_+ \)-ring \cite{etingof2016tensor}. When \( \eC \) is a multifusion category, the unit of the ring is a nonnegative linear combination of the simple objects that constitute the tensor unit \( \mathbb{1} \) of \( \eC \). When \( \eC \) is a fusion category, \( [\mathbb{1}] \) is one of the basis elements. Hereinafter, we will primarily focus on the fusion category case, as the generalization to multifusion categories is straightforward.

The simple topological symmetry defect lines or symmetry operators $W_{[a]}$'s are labeled by basis elements  $[a] \in \operatorname{Gr}(\eC)$, which is isomorphism class of simple objects in $\eC$.
These symmetry defect line obey the fusion algebra:
\begin{equation}
    W_{[a]}W_{[b]}=W_{[a\otimes b]}=W_{\sum_c N_{ab}^c [c]}, N_{ab}^c\in \mathbb{Z}_{\geq 0}.
\end{equation}
To make the fusion algebra a \( \mathbb{C} \)-algebra, we introduce the character algebra \( R(\eC) := \operatorname{Gr}(\eC) \otimes_{\mathbb{Z}} \mathbb{C} \) \cite{Nikshych2004semisimpleWHA}. The fusion category symmetry is thus equivalently described by the character algebra \( R(\eC) \), with the symmetry operator \( W_{\alpha} \) defined for \( \alpha \in R(\eC) \).

Recall that a weak Hopf algebra $H$ is called connected  if and only if the trivial representation $\mathbb{1}$ of $H$ is a simple object. In this case, the representation category $\Rep(H)$ is a unitary fusion category. On the other hand, for any fusion category $\eC$, there exists a (not unique) weak Hopf algebra $H$ such that $\Rep(H)\simeq \eC$ as tensor categories \cite{etingof2005fusion,ostrik2003module,szlachanyi2000finite,BOHM1998weak}.
Two weak Hopf algebras $H$ and $H'$ are called Morita equivalent if and only if $\Rep(H)\simeq \Rep(H')$ as tensor categories. Thus for a given fusion category symmetry, we can determine the weak Hopf symmetry $H$ up to Morita equivalence.

The construction of a weak Hopf algebra $H$ from a given fusion category $\eC$, such that $\Rep(H)\simeq \eC$, will be presented in Section~\ref{sec:reconstrucWHA}. Now suppose we already have a weak Hopf algebra $H$ at hand, and consider how to interpret the fusion category symmetry $\eC$ (or more precisely, the fusion ring $\operatorname{Gr}(\eC)$) as a dual weak Hopf symmetry. Before proceeding, let us first recall some basic facts about the representation theory of weak Hopf algebras.

The $H_L$ is the tensor unit with $H$-action given by
\begin{equation}
    x \triangleright a = \varepsilon_L(xh) =x^{\cone} a S(x^{\ctwo}), \forall x \in H, \forall a\in H_L.
\end{equation}
For two representations $V,W$, the action of $H$ on $V\otimes_{\Cbb} W$ is defined as
\begin{equation}
    h \triangleright (v\otimes w):= \sum_{(h)}h^{\cone }v \otimes h^{\ctwo}w,
\end{equation}
the tensor product of $V$ and $W$ is defined as
\begin{equation}
    V\otimes_{H_L} W = \{ x\in V\otimes_{\Cbb} W| \Delta(1)\triangleright x =x  \}.
\end{equation}
Note that $\Delta(1)$ does not act as the identity on $V \otimes_{\Cbb} W$. However, since $\Delta(1)^2 = \Delta(1)$, it serves as the identity on the balanced tensor product $V \otimes_{H_L} W$.
The dual representation of $V$ is $\hat{V}:=\Hom(V,\Cbb)$ equipped with the action $(h\triangleright f)(x):=f(S(h)\triangleright x)$.

Recall that, an element \( l \) (or \( r \)) is called a left (or right) integral of \( H \) if it satisfies \( xl = \varepsilon_L(x) l \) (or \( rx = r\varepsilon_R(x) \)). A left (or right) integral \( l \) (or \( r \)) is considered left (or right) normalized if \( \varepsilon_L(l) = 1_H \) (or \( \varepsilon_R(r) = 1_H \)). When an element \( h \) functions as both a left and right integral, it is referred to as a two-sided integral. A Haar integral in \( H \) is a two-sided integral that is normalized on both sides. For a $C^*$ weak Hopf algebra $H$, there always exist a unique Haar integral.

\begin{lemma}[Schur's lemma for weak Hopf algebra]\label{prop:schur} 
Let $\Gamma: H \to \End(V)$ and $\Phi: H \to \End(W)$ be two irreducible representations of a weak Hopf algebra $H$, and let $f: V \to W$ be a linear map such that $\Phi(x) \circ f = f \circ \Gamma(x)$ for all $x \in H$. If $\Gamma$ and $\Phi$ are not isomorphic, then $f = 0$. When $\Gamma= \Phi$, $f$ is a scalar multiple of the identity map on $V$, i.e., $f \propto \id_V$. 
\end{lemma}

\begin{proof}
The proof follows the same reasoning as in Ref.~\cite{jia2024generalized}.
\end{proof}

For Haar integral $\lambda\in H$ and arbitrary $x\in H$, we have the following \cite{BOHM1998weak,Nikshych2004semisimpleWHA}
\begin{equation}\label{eq:lambdaSpecial}
  \sum_{(\lambda)}  xS(\lambda^{(1)})\otimes \lambda^{(2)}=  \sum_{(\lambda)} 
 S(\lambda^{(1)})\otimes \lambda^{(2)}x.
\end{equation}
Then for a linear map $f: V\to W$, we define
\begin{equation}\label{eq:Fave}
    F=\sum_{(\lambda)} \Phi(S(\lambda^{(1)})) \circ f \circ \Gamma(\lambda^{\ctwo}).
\end{equation}
It's clear from Eq.~\eqref{eq:lambdaSpecial} that
\begin{equation}
   \Phi(x) \circ F =F\circ \Gamma(x),\quad \forall x\in H.
\end{equation}

\begin{lemma}\label{prop:orthgonality}
 (1)  For any $f:V\to W$, and define $F$ as in Eq.~\eqref{eq:Fave}, we have: (i)  If $\Gamma$ and $\Phi$ are not isomorphic, $F=0$; (ii) when $\Gamma=\Phi$, we have
    \begin{equation}\label{eq:Ftrace}
        F=\frac{1}{d_{\Gamma}}\Tr (f) \id_V.
    \end{equation}
    
(2) From the above result, we obtain the orthogonality relation for irreps of $H$:
\begin{equation}
    \sum_{(\lambda)} \Phi_{ij}(S(\lambda^{\cone})) \Gamma_{kl}(\lambda^{\ctwo})=\frac{\delta_{\Phi,\Gamma} \delta_{il}\delta_{jk} }{d_{\Gamma}}=  \sum_{(\lambda)} \Phi_{ij}(S(\lambda^{\ctwo})) \Gamma_{kl}(\lambda^{\cone}).
\end{equation}
where we have used the cocommutativity of $\lambda$.
\end{lemma}

\begin{proof}
   (1) Note that (i) is a direct result of Proposition~\ref{prop:schur}.  For (ii), we have
   \begin{equation}
        \Tr F= \sum_{(\lambda)} \Tr [\Gamma(S(\lambda^{(1)})) \circ f \circ \Gamma(\lambda^{\ctwo})]=\Tr [\Gamma(\varepsilon_L(\lambda)) f]=\Tr f,
    \end{equation}
which implies Eq.~\eqref{eq:Ftrace}, notice that we have used $\varepsilon_L(\lambda)=1_H$.

(2) From (1), the matrix elements $F_{il}$ is of the form
\begin{equation}
    \begin{aligned}
         \sum_{(\lambda)} \sum_{jk} \Phi_{ij}(S(\lambda^{\cone}))f_{jk} \Gamma_{kl}(\lambda^{\ctwo})=  \frac{\delta_{\Phi,\Gamma}}{d_{\Gamma}} \delta_{il} \sum_{jk} f_{jk}\delta_{jk}.
    \end{aligned}
\end{equation}
This implies the required result by taking $f_{jk}=\delta_{jj'}\delta_{kk'}$.
\end{proof}

Suppose we have identified a weak Hopf algebra $H$ that realizes the fusion category symmetry $\eC \simeq \Rep(H)$. In this case, the simple topological defect lines are labeled by the irreducible representations $\Gamma\in \Irr(H)$ of the weak Hopf algebra $H$.
The fusion ring of $\Rep(H)$ can equivalently be described using the irreducible characters $\chi_{\Gamma}$, where $\Gamma \in \Irr(H)$:
\begin{equation}
    \chi_{\Gamma} \cdot \chi_{\Phi} = \chi_{\Gamma \otimes \Phi} = \sum_{\Psi \in \Irr(H)} N_{\Gamma \Phi}^{\Psi} \chi_{\Psi}.
\end{equation}
Here, the product on the left-hand side is defined by $\chi_{\Gamma} \cdot \chi_{\Phi}(x) = \sum_{(x)} \chi_{\Gamma}(x^{\cone}) \chi_{\Phi}(x^{\ctwo})$, which corresponds to the multiplication in the dual weak Hopf algebra $\hat{H}$.
Recall that for a $C^*$ weak Hopf algebra $H$, its dual space $\hat{H} := \Hom(H, \mathbb{C})$ acquires a canonical $C^*$ weak Hopf algebra structure induced by the canonical pairing $\langle \bullet, \bullet \rangle: \hat{H} \times H \to \mathbb{C}$, defined by $\langle \varphi, h \rangle := \varphi(h)$ \footnote{In this work, we will interchangeably use $\hat{H}$ and $H^{\vee}$ to denote the dual space of $H$. 
For larger mathematical expressions, we will use the notation $(\bullet)^{\vee}$ to indicate the dual space.}. The structure morphisms are given by:
\begin{align}
& \langle \hat{\mu}(\varphi \otimes \psi), x \rangle = \langle \varphi \otimes \psi, \Delta(x) \rangle, \label{eq:pair1}\\
& \langle \hat{\eta}(1), x \rangle = \varepsilon(x), \; \text{i.e.,} \; \hat{1} = \varepsilon, \\
& \langle \hat{\Delta}(\varphi), x \otimes y \rangle = \langle \varphi, \mu(x \otimes y) \rangle, \label{eq:pairing-psi}\\
& \hat{\varepsilon}(\varphi) = \langle \varphi, \eta(1) \rangle, \\
& \langle \hat{S}(\varphi), x \rangle = \langle \varphi, S(x) \rangle. \label{eq:pair5}
\end{align}
The star operation on $\hat{H}$ is defined as
\begin{equation}
    \langle \varphi^*, x \rangle = \overline{\langle \varphi, S(x)^* \rangle}.
\end{equation}

It is clear that every character $\chi_{\Gamma}$ of $H$ is a complex-valued function on $H$, hence $\chi_{\Gamma}\in\hat{H}$.  
Since the irreducible characters form a basis of $\operatorname{Gr}(\Rep(H))$, we obtain an embedding 
$\operatorname{Gr}(\Rep(H)) \hookrightarrow \hat{H}$.  
By extending scalars, we may consider 
$R(\Rep(H))=\operatorname{Gr}(\Rep(H))\otimes_{\Zbb}\Cbb$, which becomes a subalgebra of $\hat{H}$.  
Consequently, any fusion category symmetry $\eC$ can be realized as a subalgebra of the dual weak Hopf algebra $\hat{H}$.

\begin{theorem}
    For (1+1)D phase, the fusion category symmetry $\eC$ (more precisely, $\operatorname{Gr} (\eC)$) can be equivalently characterized by some dual weak Hopf algebra $\hat{H}$, where $H$ is the weak Hopf algebra that satisfies $\Rep(H)\simeq \eC$ as fusion categories.
\end{theorem}

The above discussion suggests that weak Hopf symmetry is the most general symmetry for \((1+1)\)-dimensional systems, with fusion category symmetry as a special case. As seen in the previous subsection, both group symmetry and Hopf symmetry are  special cases of weak Hopf symmetry. Furthermore, invertible group symmetry can be viewed as a fusion category symmetry by categorifying the group. For a finite group \( G \), we introduce the discrete category \( \underline{G} \), whose objects are the group elements, and define the fusion as
\begin{equation}
    g \otimes h := gh.
\end{equation}
The fusion category symmetry \( \underline{G} \) corresponds exactly to group symmetry. Notice that \( \underline{G} \) has the same fusion rule as \( \Vect_G \). Recall that \( \Vect_G \simeq \Rep({\Cbb[G]}^{\vee}) \), where \( H \) is chosen as \( \Cbb[G] \). Thus, the fusion category symmetry \( \underline{G} \) is equivalent to dual weak Hopf symmetry \( \hat{H} =\Cbb[G]^{\vee\vee}  \cong \Cbb[G] \), which aligns perfectly with the structure of invertible symmetry.

Moreover, when considering general (not necessarily connected) weak Hopf symmetries, this framework also includes multifusion category symmetries as a special case.

We denote all cocommutative elements of $\hat{H}$ as $\Cocom (\hat{H})$, namely:
\begin{equation}
   \Cocom (\hat{H})=\{f:H\to \Cbb| f(xy)=f(yx)\}.
\end{equation}
This means that all functions in $\Cocom (\hat{H})$ are class functions of $W$.

\begin{proposition} \label{prop:characterBasis}
Let $H$ be a semisimple weak Hopf algebra.  
Then the space of cocommutative elements in the dual, $\Cocom(\hat{H})$, is spanned by the irreducible characters $\chi_{\Gamma}$ of $H$, that is,
\[
\Cocom(\hat{H}) \cong R(\Rep(H)) := \operatorname{Gr}(\Rep(H)) \otimes_{\Zbb} \Cbb.
\]
\end{proposition}

\begin{proof}
We prove the statement in two steps.

1. First, we show that $\Cocom(\hat{H}) \cong (H/[H,H])^{\vee}$.  
Let $[H,H] := \operatorname{Span}\{ xy - yx \mid x,y \in H\}$ and consider the canonical projection $\pi: H \to H/[H,H]$.  
By definition, we have
\[
\Cocom(\hat{H}) = \{ f: H \to \Cbb \mid f([H,H]) = 0 \}.
\]  
Now, define the map
\[
\Psi: (H/[H,H])^{\vee} \to \Cocom(\hat{H}), \quad \Psi(f) = f \circ \pi.
\]  
This map is clearly well-defined, linear, and bijective, hence an isomorphism.

2. Next, we use the fact that the irreducible characters of $H$ form a basis of $(H/[H,H])^{\vee}$ for any finite-dimensional semisimple algebra.  
See, e.g., Ref.~\cite[Theorem 3.6.2]{etingof2011introduction} for a proof of this fact.
\end{proof}

The above result implies that, if we are only concerned with the symmetry encoded by the fusion algebra structure, the fusion category symmetry $\eC = \Rep(H)$ can be identified with the character algebra $\Cocom(\hat{H})$.

\subsection{Reconstructing weak Hopf symmetry from given fusion category}
\label{sec:reconstrucWHA}

The weak Hopf algebra associated with a given fusion category can be reconstructed in the following two ways:

1. \emph{Tannaka-Krein reconstruction from weak fiber functor.} --- Given any fusion category $\eC$ over $\Cbb$, there exists an algebra $A$ over $\mathbb{C}$ such that we can endow $\eC$ with an exact faithful monoidal functor $F: \eC \to {_A}\Mod_A$ (we call it the  weak fiber functor over fusion category \footnote{Notice that in Ref.~\cite{jia2024weakTube}, we refer to this functor as a fiber functor. We emphasize that this is not the same as a fiber functor $F: \eC \to \mathsf{Vect}$, which does not always exist for general fusion categories.}). The algebra $H = \operatorname{End}(F)$, consisting of natural transformations from $F$ to itself. It can be proved that $H$ forms a weak Hopf symmetry, and the representation category $\Rep(H)$ is equivalent to $\eC$ as unitary fusion categories \cite{szlachanyi2000finite,ostrik2003module}.

2. \emph{Boundary tube algebra.} --- An intuitive approach involves the concept of boundary tube algebra \cite{Kitaev2012boundary, bridgeman2023invertible, jia2024weakTube,jia2025weakhopftubealgebra} as applied to the string-net model. For a given bulk fusion category \( \eC \), we consider a gapped boundary described by a \( \eC \)-module category \( {_{\eC}}\eM \) and construct the corresponding boundary tube algebra \( \mathbf{Tube}({_{\eC}}\eM) \), which is a \( C^* \) weak Hopf algebra. Setting \( \eM = \eC \), namely, considering the smooth boundary, we have the boundary tube algebra spanned by the following basis:
\begin{equation} 
\left\{  
\begin{aligned}
    \begin{tikzpicture}
    % 绘制环形区域背景，使红线居中
        \begin{scope}
            \fill[gray!20]
                (0,1.1) arc[start angle=90, end angle=270, radius=1.1] -- 
                (0,-0.5) arc[start angle=270, end angle=90, radius=0.5] -- cycle;
        \end{scope}
           \draw[line width=0.6pt,black,->] (0,0.5)--(0,0.7);
           \draw[line width=0.6pt,black,->] (0,0.7)--(0,1.0);
           \draw[line width=0.6pt,black] (0,0.5)--(0,1.1);
       %lower line
        \draw[line width=.6pt,black] (0,-0.5)--(0,-1.1);
        \draw[line width=0.6pt,black,->] (0,-1.1)--(0,-0.9);
        \draw[line width=0.6pt,black,->] (0,-1.1)--(0,-0.6);
        % 绘制红色半圆
        \draw[red, line width=0.6pt] (0,0.8) arc[start angle=90, end angle=270, radius=0.8];
           \draw[red, line width=0.6pt, ->] (0,-0.8) arc[start angle=270, end angle=180, radius=0.8];
        % 添加标记节点
        \node[line width=0.6pt, dashed, draw opacity=0.5] at (0,1.3) {$g$};
        \node[line width=0.6pt, dashed, draw opacity=0.5] at (0,-1.3) {$c$};
        \node[line width=0.6pt, dashed, draw opacity=0.5] at (-1,0) {$a$};
        \node[line width=0.6pt, dashed, draw opacity=0.5] at (0.3,-0.7) {$\nu$};
        \node[line width=0.6pt, dashed, draw opacity=0.5] at (0,-0.2) {$e$};
        \node[line width=0.6pt, dashed, draw opacity=0.5] at (0,0.2) {$f$};
        \node[line width=0.6pt, dashed, draw opacity=0.5] at (0.3,0.7) {$\mu$};
    \end{tikzpicture}
\end{aligned}
:\quad  \begin{aligned}
    &a,c,e,f,g\in \Irr(\eC),\\
    &\mu \in \Hom(a\otimes f,g),\\
    &\nu\in \Hom(c,a\otimes e) 
\end{aligned}
\right\}.
\label{eq:tubebasis}
\end{equation}
Notice that for a multiplicity-free fusion category, the vertex label can be omitted, since there is only one possible choice.
The fusion category \( \eC = \Fun_{\eC}(\eC, \eC) \) (that is, the functor category of module endofunctors of $_{\eC}\eC$, which mathematically characterizes the boundary phase) can be embedded into the representation category of the tube algebra $\Rep(\mathbf{Tube}({_{\eC}}\eC))$. In particular, there exists a monoidal functor
\begin{equation}
    F: \eC \hookrightarrow \Rep(\mathbf{Tube}({_{\eC}}\eC)).
\end{equation}
See Ref.~\cite{jia2025weakhopftubealgebra} for a detailed proof.  
Therefore, the fusion category \( \eC_{\eM}^{*} \) can be realized in terms of the weak Hopf algebra \( \mathbf{Tube}({_{\eC}}\eC) \). Moreover, it is widely expected that $\eC \simeq \Rep(\mathbf{Tube}({_{\eC}}\eC))$, with a proof based on the internal-hom construction of the boundary tube algebra given in Ref.~\cite{bai2025weakhopf}. However, a transparent proof relying solely on diagrammatic calculations is still absent.

Since the boundary tube algebra plays a crucial role in constructing explicit examples of the cluster state model in Section~\ref{sec:exampleCluster}, we briefly recall its structure here for convenience. For further details, see Refs.~\cite{jia2024weakTube,jia2025weakhopftubealgebra}.
The unit is give by
\begin{equation}
1=\sum_{f,e\in \Irr(\eC)}
\begin{aligned}
    \begin{tikzpicture}
    % 绘制环形区域背景，使红线居中
        \begin{scope}
            \fill[gray!20]
                (0,1.1) arc[start angle=90, end angle=270, radius=1.1] -- 
                (0,-0.5) arc[start angle=270, end angle=90, radius=0.5] -- cycle;
        \end{scope}
           \draw[line width=0.6pt,black,->] (0,0.5)--(0,0.9);
          % \draw[line width=0.6pt,black,->] (0,0.7)--(0,1.0);
           \draw[line width=0.6pt,black] (0,0.5)--(0,1.1);
       %lower line
        \draw[line width=.6pt,black] (0,-0.5)--(0,-1.1);
        %\draw[line width=0.6pt,black,->] (0,-1.1)--(0,-0.9);
        \draw[line width=0.6pt,black,->] (0,-1.1)--(0,-0.7);
        \node[line width=0.6pt, dashed, draw opacity=0.5] at (0.3,-0.7) {$e$};
        \node[line width=0.6pt, dashed, draw opacity=0.5] at (0.3,0.7) {$f$};
    \end{tikzpicture}
\end{aligned}.
\end{equation}
which means we set $a=\mathbb{1}$ and $c=e$, $f=g$ and $\mu=\id$, $\nu=\id$ in Eq.~\eqref{eq:tubebasis}.
The multiplication is defined as gluing two boundary tubes
\begin{equation}
    \begin{aligned}
    \begin{tikzpicture}
    % 绘制环形区域背景，使红线居中
        \begin{scope}
            \fill[gray!20]
                (0,1.1) arc[start angle=90, end angle=270, radius=1.1] -- 
                (0,-0.5) arc[start angle=270, end angle=90, radius=0.5] -- cycle;
        \end{scope}
           \draw[line width=0.6pt,black,->] (0,0.5)--(0,0.7);
           \draw[line width=0.6pt,black,->] (0,0.7)--(0,1.0);
           \draw[line width=0.6pt,black] (0,0.5)--(0,1.1);
       %lower line
        \draw[line width=.6pt,black] (0,-0.5)--(0,-1.1);
        \draw[line width=0.6pt,black,->] (0,-1.1)--(0,-0.9);
        \draw[line width=0.6pt,black,->] (0,-1.1)--(0,-0.6);
        % 绘制红色半圆
        \draw[red, line width=0.6pt] (0,0.8) arc[start angle=90, end angle=270, radius=0.8];
           \draw[red, line width=0.6pt, ->] (0,-0.8) arc[start angle=270, end angle=180, radius=0.8];
        % 添加标记节点
        \node[line width=0.6pt, dashed, draw opacity=0.5] at (0,1.3) {$g$};
        \node[line width=0.6pt, dashed, draw opacity=0.5] at (0,-1.3) {$c$};
        \node[line width=0.6pt, dashed, draw opacity=0.5] at (-1,0) {$a$};
        \node[line width=0.6pt, dashed, draw opacity=0.5] at (0.3,-0.7) {$\nu$};
        \node[line width=0.6pt, dashed, draw opacity=0.5] at (0,-0.2) {$e$};
        \node[line width=0.6pt, dashed, draw opacity=0.5] at (0,0.2) {$f$};
        \node[line width=0.6pt, dashed, draw opacity=0.5] at (0.3,0.7) {$\mu$};
    \end{tikzpicture}
\end{aligned}
\cdot
\begin{aligned}
    \begin{tikzpicture}
    % 绘制环形区域背景，使红线居中
        \begin{scope}
            \fill[gray!20]
                (0,1.1) arc[start angle=90, end angle=270, radius=1.1] -- 
                (0,-0.5) arc[start angle=270, end angle=90, radius=0.5] -- cycle;
        \end{scope}
           \draw[line width=0.6pt,black,->] (0,0.5)--(0,0.7);
           \draw[line width=0.6pt,black,->] (0,0.7)--(0,1.0);
           \draw[line width=0.6pt,black] (0,0.5)--(0,1.1);
       %lower line
        \draw[line width=.6pt,black] (0,-0.5)--(0,-1.1);
        \draw[line width=0.6pt,black,->] (0,-1.1)--(0,-0.9);
        \draw[line width=0.6pt,black,->] (0,-1.1)--(0,-0.6);
        % 绘制红色半圆
        \draw[red, line width=0.6pt] (0,0.8) arc[start angle=90, end angle=270, radius=0.8];
           \draw[red, line width=0.6pt, ->] (0,-0.8) arc[start angle=270, end angle=180, radius=0.8];
        % 添加标记节点
        \node[line width=0.6pt, dashed, draw opacity=0.5] at (0,1.3) {$g'$};
        \node[line width=0.6pt, dashed, draw opacity=0.5] at (0,-1.3) {$c'$};
        \node[line width=0.6pt, dashed, draw opacity=0.5] at (-1,0) {$a'$};
        \node[line width=0.6pt, dashed, draw opacity=0.5] at (0.3,-0.7) {$\nu'$};
        \node[line width=0.6pt, dashed, draw opacity=0.5] at (0,-0.2) {$e'$};
        \node[line width=0.6pt, dashed, draw opacity=0.5] at (0,0.2) {$f'$};
        \node[line width=0.6pt, dashed, draw opacity=0.5] at (0.3,0.7) {$\mu'$};
    \end{tikzpicture}
\end{aligned}=\delta_{e,c'}\delta_{f,g'}
\begin{aligned}
    \begin{tikzpicture}
    % 绘制环形区域背景，使红线居中
        \begin{scope}
            \fill[gray!20]
                (0,1.7) arc[start angle=90, end angle=270, radius=1.7] -- 
                (0,-0.5) arc[start angle=270, end angle=90, radius=0.5] -- cycle;
        \end{scope}
           \draw[line width=0.6pt,black,->] (0,0.5)--(0,0.7);
           \draw[line width=0.6pt,black,->] (0,0.7)--(0,1.0);
                      \draw[line width=0.6pt,black,->] (0,0.5)--(0,1.6);
           \draw[line width=0.6pt,black] (0,0.5)--(0,1.7);
       %lower line
        \draw[line width=.6pt,black] (0,-0.5)--(0,-1.7);
        \draw[line width=0.6pt,black,->] (0,-1.1)--(0,-0.9);
        \draw[line width=0.6pt,black,->] (0,-1.1)--(0,-0.6);
        \draw[line width=0.6pt,black,->] (0,-1.7)--(0,-1.4);
        % 绘制红色半圆
        \draw[red, line width=0.6pt] (0,0.8) arc[start angle=90, end angle=270, radius=0.8];
           \draw[red, line width=0.6pt, ->] (0,-0.8) arc[start angle=270, end angle=180, radius=0.8];
           % 绘制红色半圆
        \draw[red, line width=0.6pt] (0,1.3) arc[start angle=90, end angle=270, radius=1.3];
           \draw[red, line width=0.6pt, ->] (0,-1.3) arc[start angle=270, end angle=180, radius=1.3];
        % 添加标记节点
                \node[line width=0.6pt, dashed, draw opacity=0.5] at (-0.3,1.3) {$\mu$};
                      \node[line width=0.6pt, dashed, draw opacity=0.5] at (-0.3,-1.3) {$\nu$};
                            \node[line width=0.6pt, dashed, draw opacity=0.5] at (0.3,1.6) {$g$};
                                      \node[line width=0.6pt, dashed, draw opacity=0.5] at (0.3,-1.6) {$c$};
        \node[line width=0.6pt, dashed, draw opacity=0.5] at (0.3,1.1) {$g'$};
        \node[line width=0.6pt, dashed, draw opacity=0.5] at (0.3,-1.1) {$c'$};
        \node[line width=0.6pt, dashed, draw opacity=0.5] at (-1,0) {$a'$};
        \node[line width=0.6pt, dashed, draw opacity=0.5] at (0.3,-0.7) {$\nu'$};
        \node[line width=0.6pt, dashed, draw opacity=0.5] at (0,-0.2) {$e'$};
        \node[line width=0.6pt, dashed, draw opacity=0.5] at (0,0.2) {$f'$};
        \node[line width=0.6pt, dashed, draw opacity=0.5] at (0.3,0.7) {$\mu'$};
                \node[line width=0.6pt, dashed, draw opacity=0.5] at (-1.5,0) {$a$};
    \end{tikzpicture}
\end{aligned}.
\end{equation}
After applying the topological local moves—namely the \( F \)-moves, parallel moves, and loop moves as described in~\cite{jia2024weakTube}—to the diagram on the right-hand side, we obtain a linear combination of the basis elements given in Eq.~\eqref{eq:tubebasis}.
The counit is defined by
\begin{equation}
    \varepsilon\left( 
    \begin{aligned}
    \begin{tikzpicture}
    % 绘制环形区域背景，使红线居中
        \begin{scope}
            \fill[gray!20]
                (0,1.1) arc[start angle=90, end angle=270, radius=1.1] -- 
                (0,-0.5) arc[start angle=270, end angle=90, radius=0.5] -- cycle;
        \end{scope}
           \draw[line width=0.6pt,black,->] (0,0.5)--(0,0.7);
           \draw[line width=0.6pt,black,->] (0,0.7)--(0,1.0);
           \draw[line width=0.6pt,black] (0,0.5)--(0,1.1);
       %lower line
        \draw[line width=.6pt,black] (0,-0.5)--(0,-1.1);
        \draw[line width=0.6pt,black,->] (0,-1.1)--(0,-0.9);
        \draw[line width=0.6pt,black,->] (0,-1.1)--(0,-0.6);
        % 绘制红色半圆
        \draw[red, line width=0.6pt] (0,0.8) arc[start angle=90, end angle=270, radius=0.8];
           \draw[red, line width=0.6pt, ->] (0,-0.8) arc[start angle=270, end angle=180, radius=0.8];
        % 添加标记节点
        \node[line width=0.6pt, dashed, draw opacity=0.5] at (0,1.3) {$g$};
        \node[line width=0.6pt, dashed, draw opacity=0.5] at (0,-1.3) {$c$};
        \node[line width=0.6pt, dashed, draw opacity=0.5] at (-1,0) {$a$};
        \node[line width=0.6pt, dashed, draw opacity=0.5] at (0.3,-0.7) {$\nu$};
        \node[line width=0.6pt, dashed, draw opacity=0.5] at (0,-0.2) {$e$};
        \node[line width=0.6pt, dashed, draw opacity=0.5] at (0,0.2) {$f$};
        \node[line width=0.6pt, dashed, draw opacity=0.5] at (0.3,0.7) {$\mu$};
    \end{tikzpicture}
\end{aligned}
    \right) = \frac{
    \delta_{f,e} \delta_{c,g} }{d_c}
    \begin{aligned}
    \begin{tikzpicture}
    % 绘制环形区域背景，使红线居中
        \begin{scope}
            \fill[gray!20]
                (0,1.1) arc[start angle=90, end angle=270, radius=1.1] -- 
                (0,-0.5) arc[start angle=270, end angle=90, radius=0.5] -- cycle;
        \end{scope}
          % \draw[line width=0.6pt,black,->] (0,0.5)--(0,0.7);
           \draw[line width=0.6pt,black,->] (0,0.7)--(0,1.0);
           \draw[line width=0.6pt,black] (0,0.5)--(0,1.1);
       %lower line
        \draw[line width=.6pt,black] (0,-1.1)--(0,1.1);
        \draw[line width=0.6pt,black,->] (0,-1.1)--(0,-0.9);
        \draw[line width=0.6pt,black,->] (0,-1.1)--(0,0.2);
        % 绘制红色半圆
        \draw[red, line width=0.6pt] (0,0.8) arc[start angle=90, end angle=270, radius=0.8];
           \draw[red, line width=0.6pt, ->] (0,-0.8) arc[start angle=270, end angle=180, radius=0.8];
  %%%%outer loop
 \draw[black, line width=0.6pt] (0,-1.1) arc[start angle=-90, end angle=-270, radius=1.1];
                 % 添加标记节点
        \node[line width=0.6pt, dashed, draw opacity=0.5] at (0,1.3) {$g$};
       % \node[line width=0.6pt, dashed, draw opacity=0.5] at (0,-1.3) {$c$};
        \node[line width=0.6pt, dashed, draw opacity=0.5] at (-0.6,0) {$a$};
        \node[line width=0.6pt, dashed, draw opacity=0.5] at (0.3,-0.7) {$\nu$};
        \node[line width=0.6pt, dashed, draw opacity=0.5] at (0.3,-0.2) {$e$};
       % \node[line width=0.6pt, dashed, draw opacity=0.5] at (0,0.2) {$f$};
        \node[line width=0.6pt, dashed, draw opacity=0.5] at (0.3,0.7) {$\mu$};
    \end{tikzpicture}
\end{aligned} = \delta_{f,e} \delta_{c,g} \sqrt{\frac{d_a d_e}{d_g}}
.
\end{equation}
The comultiplication is defined as
\begin{equation}
      \Delta\left( 
    \begin{aligned}
    \begin{tikzpicture}
    % 绘制环形区域背景，使红线居中
        \begin{scope}
            \fill[gray!20]
                (0,1.1) arc[start angle=90, end angle=270, radius=1.1] -- 
                (0,-0.5) arc[start angle=270, end angle=90, radius=0.5] -- cycle;
        \end{scope}
           \draw[line width=0.6pt,black,->] (0,0.5)--(0,0.7);
           \draw[line width=0.6pt,black,->] (0,0.7)--(0,1.0);
           \draw[line width=0.6pt,black] (0,0.5)--(0,1.1);
       %lower line
        \draw[line width=.6pt,black] (0,-0.5)--(0,-1.1);
        \draw[line width=0.6pt,black,->] (0,-1.1)--(0,-0.9);
        \draw[line width=0.6pt,black,->] (0,-1.1)--(0,-0.6);
        % 绘制红色半圆
        \draw[red, line width=0.6pt] (0,0.8) arc[start angle=90, end angle=270, radius=0.8];
           \draw[red, line width=0.6pt, ->] (0,-0.8) arc[start angle=270, end angle=180, radius=0.8];
        % 添加标记节点
        \node[line width=0.6pt, dashed, draw opacity=0.5] at (0,1.3) {$g$};
        \node[line width=0.6pt, dashed, draw opacity=0.5] at (0,-1.3) {$c$};
        \node[line width=0.6pt, dashed, draw opacity=0.5] at (-1,0) {$a$};
        \node[line width=0.6pt, dashed, draw opacity=0.5] at (0.3,-0.7) {$\nu$};
        \node[line width=0.6pt, dashed, draw opacity=0.5] at (0,-0.2) {$e$};
        \node[line width=0.6pt, dashed, draw opacity=0.5] at (0,0.2) {$f$};
        \node[line width=0.6pt, dashed, draw opacity=0.5] at (0.3,0.7) {$\mu$};
    \end{tikzpicture}
\end{aligned}
    \right)
    =
    \sum_{k, l,\zeta} \sqrt{ \frac{d_l}{d_kd_a} }   \begin{aligned}
    \begin{tikzpicture}
    % 绘制环形区域背景，使红线居中
        \begin{scope}
            \fill[gray!20]
                (0,1.1) arc[start angle=90, end angle=270, radius=1.1] -- 
                (0,-0.5) arc[start angle=270, end angle=90, radius=0.5] -- cycle;
        \end{scope}
           \draw[line width=0.6pt,black,->] (0,0.5)--(0,0.7);
           \draw[line width=0.6pt,black,->] (0,0.7)--(0,1.0);
           \draw[line width=0.6pt,black] (0,0.5)--(0,1.1);
       %lower line
        \draw[line width=.6pt,black] (0,-0.5)--(0,-1.1);
        \draw[line width=0.6pt,black,->] (0,-1.1)--(0,-0.9);
        \draw[line width=0.6pt,black,->] (0,-1.1)--(0,-0.6);
        % 绘制红色半圆
        \draw[red, line width=0.6pt] (0,0.8) arc[start angle=90, end angle=270, radius=0.8];
           \draw[red, line width=0.6pt, ->] (0,-0.8) arc[start angle=270, end angle=180, radius=0.8];
        % 添加标记节点
        \node[line width=0.6pt, dashed, draw opacity=0.5] at (0,1.3) {$g$};
        \node[line width=0.6pt, dashed, draw opacity=0.5] at (0,-1.3) {$l$};
        \node[line width=0.6pt, dashed, draw opacity=0.5] at (-1,0) {$a$};
        \node[line width=0.6pt, dashed, draw opacity=0.5] at (0.3,-0.7) {$\zeta$};
        \node[line width=0.6pt, dashed, draw opacity=0.5] at (0,-0.2) {$k$};
        \node[line width=0.6pt, dashed, draw opacity=0.5] at (0,0.2) {$f$};
        \node[line width=0.6pt, dashed, draw opacity=0.5] at (0.3,0.7) {$\mu$};
    \end{tikzpicture}
\end{aligned} \otimes     
\begin{aligned}
    \begin{tikzpicture}
    % 绘制环形区域背景，使红线居中
        \begin{scope}
            \fill[gray!20]
                (0,1.1) arc[start angle=90, end angle=270, radius=1.1] -- 
                (0,-0.5) arc[start angle=270, end angle=90, radius=0.5] -- cycle;
        \end{scope}
           \draw[line width=0.6pt,black,->] (0,0.5)--(0,0.7);
           \draw[line width=0.6pt,black,->] (0,0.7)--(0,1.0);
           \draw[line width=0.6pt,black] (0,0.5)--(0,1.1);
       %lower line
        \draw[line width=.6pt,black] (0,-0.5)--(0,-1.1);
        \draw[line width=0.6pt,black,->] (0,-1.1)--(0,-0.9);
        \draw[line width=0.6pt,black,->] (0,-1.1)--(0,-0.6);
        % 绘制红色半圆
        \draw[red, line width=0.6pt] (0,0.8) arc[start angle=90, end angle=270, radius=0.8];
           \draw[red, line width=0.6pt, ->] (0,-0.8) arc[start angle=270, end angle=180, radius=0.8];
        % 添加标记节点
        \node[line width=0.6pt, dashed, draw opacity=0.5] at (0,1.3) {$l$};
        \node[line width=0.6pt, dashed, draw opacity=0.5] at (0,-1.3) {$c$};
        \node[line width=0.6pt, dashed, draw opacity=0.5] at (-1,0) {$a$};
        \node[line width=0.6pt, dashed, draw opacity=0.5] at (0.3,-0.7) {$\nu$};
        \node[line width=0.6pt, dashed, draw opacity=0.5] at (0,-0.2) {$e$};
        \node[line width=0.6pt, dashed, draw opacity=0.5] at (0,0.2) {$k$};
        \node[line width=0.6pt, dashed, draw opacity=0.5] at (0.3,0.7) {$\zeta$};
    \end{tikzpicture}
\end{aligned}.
\end{equation}
The antipode map is defined as
\begin{equation}
         S\left( 
    \begin{aligned}
    \begin{tikzpicture}
    % 绘制环形区域背景，使红线居中
        \begin{scope}
            \fill[gray!20]
                (0,1.1) arc[start angle=90, end angle=270, radius=1.1] -- 
                (0,-0.5) arc[start angle=270, end angle=90, radius=0.5] -- cycle;
        \end{scope}
           \draw[line width=0.6pt,black,->] (0,0.5)--(0,0.7);
           \draw[line width=0.6pt,black,->] (0,0.7)--(0,1.0);
           \draw[line width=0.6pt,black] (0,0.5)--(0,1.1);
       %lower line
        \draw[line width=.6pt,black] (0,-0.5)--(0,-1.1);
        \draw[line width=0.6pt,black,->] (0,-1.1)--(0,-0.9);
        \draw[line width=0.6pt,black,->] (0,-1.1)--(0,-0.6);
        % 绘制红色半圆
        \draw[red, line width=0.6pt] (0,0.8) arc[start angle=90, end angle=270, radius=0.8];
           \draw[red, line width=0.6pt, ->] (0,-0.8) arc[start angle=270, end angle=180, radius=0.8];
        % 添加标记节点
        \node[line width=0.6pt, dashed, draw opacity=0.5] at (0,1.3) {$g$};
        \node[line width=0.6pt, dashed, draw opacity=0.5] at (0,-1.3) {$c$};
        \node[line width=0.6pt, dashed, draw opacity=0.5] at (-1,0) {$a$};
        \node[line width=0.6pt, dashed, draw opacity=0.5] at (0.3,-0.7) {$\nu$};
        \node[line width=0.6pt, dashed, draw opacity=0.5] at (0,-0.2) {$e$};
        \node[line width=0.6pt, dashed, draw opacity=0.5] at (0,0.2) {$f$};
        \node[line width=0.6pt, dashed, draw opacity=0.5] at (0.3,0.7) {$\mu$};
    \end{tikzpicture}
\end{aligned}
    \right)
    =   \frac{d_f}{d_g} \begin{aligned}
    \begin{tikzpicture}
    % 绘制环形区域背景，使红线居中
        \begin{scope}
            \fill[gray!20]
                (0,1.1) arc[start angle=90, end angle=270, radius=1.1] -- 
                (0,-0.5) arc[start angle=270, end angle=90, radius=0.5] -- cycle;
        \end{scope}
           \draw[line width=0.6pt,black,->] (0,0.5)--(0,0.7);
           \draw[line width=0.6pt,black,->] (0,0.7)--(0,1.0);
           \draw[line width=0.6pt,black] (0,0.5)--(0,1.1);
       %lower line
        \draw[line width=.6pt,black] (0,-0.5)--(0,-1.1);
        \draw[line width=0.6pt,black,->] (0,-1.1)--(0,-0.9);
        \draw[line width=0.6pt,black,->] (0,-1.1)--(0,-0.6);
        % 绘制红色半圆
        \draw[red, line width=0.6pt] (0,0.8) arc[start angle=90, end angle=270, radius=0.8];
           \draw[red, line width=0.6pt, ->] (0,-0.8) arc[start angle=270, end angle=180, radius=0.8];
        % 添加标记节点
        \node[line width=0.6pt, dashed, draw opacity=0.5] at (0,1.3) {$e$};
        \node[line width=0.6pt, dashed, draw opacity=0.5] at (0,-1.3) {$f$};
        \node[line width=0.6pt, dashed, draw opacity=0.5] at (-1,0) {$\bar{a}$};
        \node[line width=0.6pt, dashed, draw opacity=0.5] at (0.3,-0.7) {$\mu$};
        \node[line width=0.6pt, dashed, draw opacity=0.5] at (0,-0.2) {$g$};
        \node[line width=0.6pt, dashed, draw opacity=0.5] at (0,0.2) {$c$};
        \node[line width=0.6pt, dashed, draw opacity=0.5] at (0.3,0.7) {$\nu$};
    \end{tikzpicture}
\end{aligned}.
\end{equation}
It can be proved that this is a \( C^* \) weak Hopf algebra \cite{jia2024weakTube} with $*$-operation given by
\begin{equation}
        \left( 
    \begin{aligned}
    \begin{tikzpicture}
    % 绘制环形区域背景，使红线居中
        \begin{scope}
            \fill[gray!20]
                (0,1.1) arc[start angle=90, end angle=270, radius=1.1] -- 
                (0,-0.5) arc[start angle=270, end angle=90, radius=0.5] -- cycle;
        \end{scope}
           \draw[line width=0.6pt,black,->] (0,0.5)--(0,0.7);
           \draw[line width=0.6pt,black,->] (0,0.7)--(0,1.0);
           \draw[line width=0.6pt,black] (0,0.5)--(0,1.1);
       %lower line
        \draw[line width=.6pt,black] (0,-0.5)--(0,-1.1);
        \draw[line width=0.6pt,black,->] (0,-1.1)--(0,-0.9);
        \draw[line width=0.6pt,black,->] (0,-1.1)--(0,-0.6);
        % 绘制红色半圆
        \draw[red, line width=0.6pt] (0,0.8) arc[start angle=90, end angle=270, radius=0.8];
           \draw[red, line width=0.6pt, ->] (0,-0.8) arc[start angle=270, end angle=180, radius=0.8];
        % 添加标记节点
        \node[line width=0.6pt, dashed, draw opacity=0.5] at (0,1.3) {$g$};
        \node[line width=0.6pt, dashed, draw opacity=0.5] at (0,-1.3) {$c$};
        \node[line width=0.6pt, dashed, draw opacity=0.5] at (-1,0) {$a$};
        \node[line width=0.6pt, dashed, draw opacity=0.5] at (0.3,-0.7) {$\nu$};
        \node[line width=0.6pt, dashed, draw opacity=0.5] at (0,-0.2) {$e$};
        \node[line width=0.6pt, dashed, draw opacity=0.5] at (0,0.2) {$f$};
        \node[line width=0.6pt, dashed, draw opacity=0.5] at (0.3,0.7) {$\mu$};
    \end{tikzpicture}
\end{aligned}
    \right)^*
    =   \frac{d_e}{d_c} \begin{aligned}
    \begin{tikzpicture}
    % 绘制环形区域背景，使红线居中
        \begin{scope}
            \fill[gray!20]
                (0,1.1) arc[start angle=90, end angle=270, radius=1.1] -- 
                (0,-0.5) arc[start angle=270, end angle=90, radius=0.5] -- cycle;
        \end{scope}
           \draw[line width=0.6pt,black,->] (0,0.5)--(0,0.7);
           \draw[line width=0.6pt,black,->] (0,0.7)--(0,1.0);
           \draw[line width=0.6pt,black] (0,0.5)--(0,1.1);
       %lower line
        \draw[line width=.6pt,black] (0,-0.5)--(0,-1.1);
        \draw[line width=0.6pt,black,->] (0,-1.1)--(0,-0.9);
        \draw[line width=0.6pt,black,->] (0,-1.1)--(0,-0.6);
        % 绘制红色半圆
        \draw[red, line width=0.6pt] (0,0.8) arc[start angle=90, end angle=270, radius=0.8];
           \draw[red, line width=0.6pt, ->] (0,-0.8) arc[start angle=270, end angle=180, radius=0.8];
        % 添加标记节点
        \node[line width=0.6pt, dashed, draw opacity=0.5] at (0,1.3) {$f$};
        \node[line width=0.6pt, dashed, draw opacity=0.5] at (0,-1.3) {$e$};
        \node[line width=0.6pt, dashed, draw opacity=0.5] at (-1,0) {$\bar{a}$};
        \node[line width=0.6pt, dashed, draw opacity=0.5] at (0.3,-0.7) {$\mu$};
        \node[line width=0.6pt, dashed, draw opacity=0.5] at (0,-0.2) {$c$};
        \node[line width=0.6pt, dashed, draw opacity=0.5] at (0,0.2) {$g$};
        \node[line width=0.6pt, dashed, draw opacity=0.5] at (0.3,0.7) {$\nu$};
    \end{tikzpicture}
\end{aligned}.\label{eq:starope}
\end{equation}

In constructing the lattice model, we also make use of the Haar integral of the boundary tube algebra. It has been shown in Ref.~\cite{jia2025weakhopftubealgebra} that the Haar integral takes the following form:
\begin{equation}\label{eq:HaarTube}
    \lambda = \frac{1}{\operatorname{rank} \eC} \sum_{a,x,y,\mu} 
    \sqrt{\frac{d_a}{d_x^3 d_y}}
        \begin{aligned}
    \begin{tikzpicture}
    % 绘制环形区域背景，使红线居中
        \begin{scope}
            \fill[gray!20]
                (0,1.1) arc[start angle=90, end angle=270, radius=1.1] -- 
                (0,-0.5) arc[start angle=270, end angle=90, radius=0.5] -- cycle;
        \end{scope}
           \draw[line width=0.6pt,black,->] (0,0.5)--(0,0.7);
           \draw[line width=0.6pt,black,->] (0,0.7)--(0,1.0);
           \draw[line width=0.6pt,black] (0,0.5)--(0,1.1);
       %lower line
        \draw[line width=.6pt,black] (0,-0.5)--(0,-1.1);
        \draw[line width=0.6pt,black,->] (0,-1.1)--(0,-0.9);
        \draw[line width=0.6pt,black,->] (0,-1.1)--(0,-0.6);
        % 绘制红色半圆
        \draw[red, line width=0.6pt] (0,0.8) arc[start angle=90, end angle=270, radius=0.8];
           \draw[red, line width=0.6pt, ->] (0,-0.8) arc[start angle=270, end angle=180, radius=0.8];
        % 添加标记节点
        \node[line width=0.6pt, dashed, draw opacity=0.5] at (0,1.3) {$y$};
        \node[line width=0.6pt, dashed, draw opacity=0.5] at (0,-1.3) {$y$};
        \node[line width=0.6pt, dashed, draw opacity=0.5] at (-1,0) {$a$};
        \node[line width=0.6pt, dashed, draw opacity=0.5] at (0.3,-0.7) {$\mu$};
        \node[line width=0.6pt, dashed, draw opacity=0.5] at (0,-0.3) {$x$};
        \node[line width=0.6pt, dashed, draw opacity=0.5] at (0,0.3) {$x$};
        \node[line width=0.6pt, dashed, draw opacity=0.5] at (0.3,0.7) {$\mu$};
    \end{tikzpicture}
\end{aligned}.
\end{equation}
where $\mathrm{rank} \eC$ denotes the number of objects in $\Irr \eC$.

\subsection{Weak Hopf symmetric single qudit system}

Now let us consider a simplest example of weak Hopf symmetric system, a $C^*$ weak Hopf qudit.
By a weak Hopf qudit we mean a complex weak Hopf algebra $H$ with inner product given by
\begin{equation}
    \langle a,b\rangle=\Lambda(a^*b),
\end{equation}
where $\Lambda$ is the Haar integral of $\hat{H}$ (also called Haar measure of $H$).
We will hereinafter denote Haar integral of $H$ as $\lambda$ and Haar integral of $\hat{H}$ as $\Lambda$.
For the specail case of group algebra $H=\Cbb[G]$, the Haar integral is of the form
\begin{equation}
    \lambda=\frac{1}{|G|}\sum_{g\in G} g,\quad \Lambda=\delta_{1_G}(\bullet).
\end{equation}
The inner product for this case it the usual\footnote{Recall that the involution of $\mathbb{C}[G]$ is defined by $g^* = g^{-1}$ for all $g \in G$, extended antilinearly over the entire space.} $\langle g,h\rangle= \delta_{1_G}(g^*h)=\delta_{1_G, g^{-1}h}=\delta_{g,h}$ for all $g,h\in G$.

\begin{table}
    \centering
{
    \setlength{\tabcolsep}{3mm}
    \renewcommand{\arraystretch}{1.3} % adjust row height
    \begin{tabular}{|l|c|}
    \hhline{|==|}
 Weak Hopf qudit &  $\mathcal{H}=H$   \\
    \hhline{|--|}
%    Standard basis & $|g_0\rangle=|1_{H}\rangle,\cdots, |g_{|H|-1}\rangle$  \\     
 %   \hhline{|--|}
    % Irrep basis & $|\Gamma_{ij}\rangle=\sqrt{\frac{d_{\Gamma}}{|H|}} \sum_{(\lambda)} \Gamma_{ij}(\lambda^{(1)})|\lambda^{(2)}\rangle$  \\
    % \hhline{|--|}
    Regular action  & $\begin{aligned}
    & \XR_g|h\rangle =|gh\rangle,\ \XL_g|h\rangle =|hS^{-1}(g)\rangle,  \\
    & \tilde{\XR}_g|h\rangle =|S^{-1}(g)h\rangle,\ \tilde{\XL}_g|h\rangle =|hg\rangle 
    \end{aligned}$   \\
    \hhline{|--|}
    Dual action  &  $\begin{aligned}
  &\ZR_{\psi}|h\rangle =|\psi \rightharpoonup h\rangle= \sum_{(h)}\psi(h^{\ctwo})|h^{\cone}\rangle\\
   &   \ZL_{\psi}|h\rangle=|h \leftharpoonup \hat{S}(\psi)\rangle = \sum_{(h)}\psi(S(h^{\cone}))|h^{\ctwo}\rangle\\
   &  \tilde{\ZR}_{\psi}|h\rangle =|\hat{S}(\psi) \rightharpoonup h\rangle= \sum_{(h)}\psi(S(h^{\ctwo}))|h^{\cone}\rangle,\\
    &  \tilde{\ZL}_{\psi}|h\rangle=|h \leftharpoonup \psi\rangle = \sum_{(h)}\psi(h^{\cone})|h^{\ctwo}\rangle
    \end{aligned}$ \\
    % \hhline{|--|}
    % Symmetry action  &  $\begin{aligned}
    % & \ZR_{\Gamma} = \Tr' Z_{\Gamma},\ \ZL_{\Gamma} = \Tr' Z^{\ddagger}_{\Gamma}, \\
    % & \tilde{\ZL}_{\Gamma} = \Tr' \tilde{Z}_{\Gamma},\ \tilde{\ZR}_{\Gamma} = \Tr' \tilde{Z}^{\ddagger}_{\Gamma} 
    % \end{aligned}$ \\   
    % \hhline{|--|}
    % Generalized Pauli operators & $\begin{aligned}
    % & X = \XR_{\lambda} = \XL_{\lambda} = \tilde{\XR}_{\lambda} = \tilde{\XL}_{\lambda} = \tilde{X}, \\
    % & Z = \sum_{\Gamma \in \Irr(H)} \frac{d_{\Gamma}}{|H|} J_{\Gamma}, \\
    % & \tilde{Z} = \sum_{\Gamma \in \Irr(H)} \frac{d_{\Gamma}}{|H|} \tilde{J}_{\Gamma}, \\
    % & Z = \tilde{Z}
    % \end{aligned}$ \\
    \hhline{|==|}
    \end{tabular} 
}
    \caption{Summary of quantum operations for weak Hopf qudit system.}
    \label{tab:HopfQudit}
\end{table}

The regular action of a weak Hopf algebra on itself can be viewed as a generalization of Pauli $X$-type operators. For the left action $H \curvearrowright H$, we define:
\begin{equation}
\XR_g |h\rangle = |gh\rangle, \quad \XL_g |h\rangle = |hS^{-1}(g)\rangle.
\end{equation}
For the right action $H \curvearrowleft H$, we define:
\begin{equation}
\tilde{\XR}_g |h\rangle = |S^{-1}(g)h\rangle, \quad \tilde{\XL}_g |h\rangle = |hg\rangle.
\end{equation}
The use of $S^{-1}$ in the above definitions ensures the commutative local stabilizer in weak Hopf lattice gauge theory \cite{Jia2023weak} \footnote{Alternatively, we can choose to include $S^{-1}$ in the weak Hopf Pauli $Z$ operators, in which case there is no need to use $S^{-1}$ in the definition of the weak Hopf Pauli $X$ operators.
}. Additionally, note that for a $C^*$ Hopf algebra and more general weak Kac algebra, $S^{-1} = S$.

There are also canonical actions of the dual weak Hopf algebra $\hat{H}$ on the Hopf qudit $H$, defined using Sweedler’s notation as follows:
\begin{equation}
\varphi \rightharpoonup x:=\sum_{(x)}x^{(1)} \langle \varphi, x^{(2)}\rangle, \quad  x \leftharpoonup \varphi:=\sum_{(x)}   \langle \varphi,    x^{(1)}    \rangle  x^{(2)},
\end{equation}
for all $\varphi \in \hat{H}$ and $x \in H$.
The operators corresponding the $\hat{H}$ action on weak Hopf qudit $H$ can be regarded as generalized Pauli $Z$-type operators.
For left action $\hat{H}\curvearrowright H$, we define:
\begin{align}
      \ZR_{\psi}|h\rangle =|\psi \rightharpoonup h\rangle= \sum_{(h)}\psi(h^{\ctwo})|h^{\cone}\rangle,\\
      \ZL_{\psi}|h\rangle=|h \leftharpoonup \hat{S}(\psi)\rangle = \sum_{(h)}\psi(S(h^{\cone}))|h^{\ctwo}\rangle.
\end{align}
For right action $H \curvearrowleft \hat{H}$, we define
\begin{align}
     \tilde{\ZR}_{\psi}|h\rangle =|\hat{S}(\psi) \rightharpoonup h\rangle= \sum_{(h)}\psi(S(h^{\ctwo}))|h^{\cone}\rangle,\\
      \tilde{\ZL}_{\psi}|h\rangle=|h \leftharpoonup \psi\rangle = \sum_{(h)}\psi(h^{\cone})|h^{\ctwo}\rangle.
\end{align}
Since the character $\chi_{\Gamma}$ of irreducible representations $\Gamma \in \Irr(H)$ lies in $\hat{H}$, we can define $\ZR_{\Gamma}, \ZL_{\Gamma}, \tilde{\ZR}_{\Gamma}, \tilde{\ZL}_{\Gamma}$ by replacing $\psi$ with $\chi_{\Gamma}$ in the expressions given above (In certain cases, to avoid ambiguity, we will also use notations such as $\ZR_{\chi_{\Gamma}}$, and similar variants, for clarity).

It is convenient to introduce representation-matrix-valued $Z$-type operators:
\begin{equation}
    Z_{\Gamma}|h\rangle = \sum_{(h)} \Gamma(h^{\ctwo}) |h^{\cone}\rangle, \quad Z_{\Gamma}^{\ddagger} |h\rangle = \sum_{(h)} \Gamma(S(h^{\cone})) |h^{\ctwo}\rangle.
\end{equation}
Similarly, we define
\begin{equation}
    \tilde{Z}_{\Gamma} |h\rangle= \sum_{(h)} \Gamma(h^{\cone}) |h^{\ctwo}\rangle, \quad \tilde{Z}_{\Gamma}^{\ddagger} |h\rangle= \sum_{(h)} \Gamma(S(h^{\ctwo})) |h^{\cone}\rangle.
\end{equation}
Using $\operatorname{Tr}'$ to denote the trace over the representation space, we obtain
\begin{equation}
    \ZR_{\Gamma} = \Tr' Z_{\Gamma}, \quad \ZL_{\Gamma} = \Tr' Z^{\ddagger}_{\Gamma}, \quad \tilde{\ZL}_{\Gamma} = \Tr' \tilde{Z}_{\Gamma}, \quad \tilde{\ZR}_{\Gamma} = \Tr' \tilde{Z}^{\ddagger}_{\Gamma}.
\end{equation}

Following the notation of Kitaev \cite{Kitaev2003, Buerschaper2013a,jia2023boundary,Jia2023weak}, we have $L_+ = \XR$, $L_- = \XL$, $T_+ = \ZR$, and $T_- = \ZL$. The notations here primarily follows Refs.~\cite{albert2021spin, fechisin2023noninvertible, jia2024generalized}, with slight modifications to accommodate the weak Hopf algebra case.

%\paragraph{Weak Hopf symmetric qudit Hamiltonians}

The simplest example of a Hamiltonian with weak Hopf symmetry can be constructed from the center of $H$. Let $\gamma \in \operatorname{Center}(H)$ (the center of $H$ is defined as $\operatorname{Center}(H)=\{h\in H| hg=gh,\, \forall g\in H\}$), and define $\mathbb{H} = \XR_\gamma$. Then, $[\XR_h, \mathbb{H}] = 0$ for all $h\in H$. Similarly, a model exhibiting comodule algebra symmetry can be constructed.

% \begin{example}
% Consider a single qudit weak Hopf quantum double model that has only one vertex and one face, the weak Hopf qudit $\mathcal{H}=H$ is put on the edge. For vertex and face we define local operators as follows:
% \begin{equation}
% \begin{aligned}
% \begin{tikzpicture}
%     \filldraw (0,0) circle (2pt) node[anchor=east] {}; 
%     \draw[thick, decoration={markings, mark=at position 0.5 with {\arrow{latex}}}, postaction={decorate}] 
%         (0,0) .. controls (2,2) and (-2,2) .. (0,0);
% \end{tikzpicture}
%     \end{aligned}
%     \begin{aligned}
%        & \Av_v^g|x\rangle= \sum_{(g)} \XL_{g^{\cone}} \XR_{g^{\ctwo}}|x\rangle= \sum_{(g)}|g^{\cone}(x)S(g^{\ctwo})\rangle,\\
%        & \Bf_v^{\psi}|f\rangle= \ZL_{\psi} |x\rangle = |x\leftharpoonup \hat{S}(\psi)\rangle.
%     \end{aligned}
% \end{equation}
% The Hamiltonian is defined as
% \begin{equation}
%     \mathbb{H}=-\Av_v^{\lambda}-\Bf_f^{\Lambda},
% \end{equation}
% where $\lambda$ and $\Lambda$ are Haar integral and Haar measure respectively.
% \end{example}

\section{Weak Hopf algebra tensor network states}
\label{sec:tensor-network}

The key feature of Hopf and weak Hopf qudits lies in their comultiplication, which can be used to generate entanglement and quantum correlations. In Ref.~\cite{jia2024generalized}, the Hopf tensor network was employed as a fundamental tool to solve the Hopf cluster state model. For the more general weak Hopf cluster ladder model, one needs to apply a suitable generalization of the Hopf tensor network~\cite{jia2024generalized,Jia2023weak,jia2023boundary,girelli2021semidual,Buerschaper2013a}.

\subsection{Weak Hopf pairing}

The main tool we shall use is the pairing between a weak Hopf algebra $H$ and its dual $\hat{H}$ (Eqs.~\eqref{eq:pair1}--\eqref{eq:pair5}). This construction can be further extended to any pair of weak Hopf algebras, as follows:

\begin{definition}[canonical pairing]
\label{def:pairingDef}
     A canonical pairing (also called (weak) Hopf pairing) is a function $\langle \bullet, \bullet \rangle: M\otimes N\to \mathbb{C}$ between two weak Hopf algebras $M,N$ is a bilinear map satisfying
     \begin{align}
       &  \langle hg,a\rangle = \sum_{(a)} \langle h,a^{(1)}\rangle \langle g,a^{(2)}\rangle,\\
    &     \langle h,ab\rangle =\sum_{(h)} \langle h^{(1)} ,a\rangle \langle h^{(2)},b\rangle,\\
     &    \langle 1_M,a\rangle =\varepsilon_N(a),\\
       &  \langle h,1_N\rangle =\varepsilon_M(h),\\
        & \langle S_M(h),a\rangle = \langle h, S_N(a)\rangle.
     \end{align}
\end{definition}

The pairing between $\hat{H}$ and $H$ is a particular instance of the canonical pairing defined above. Based on this, we introduce a generalized weak Hopf tensor network formulated via pairings, endowed with advantageous properties that make it suitable for solving the lattice model discussed below. 
While our primary emphasis is on the pairing between $\hat{H}$ and $H$, the extension to other pairings is straightforward.

We will draw elements $h\in H$ and $\psi \in \hat{H}$ and their pairing $\langle \psi, h\rangle=\psi(h)\in \Cbb$ as
\begin{equation}
\begin{aligned}
\begin{tikzpicture}
    \draw[black, line width=1.0pt] (0,0) -- (0,0.7);    
   \filldraw[green!30, draw=black] (0,0) circle (0.2);
    \node at (0, 0) {$\scriptstyle h$};
\end{tikzpicture}
\end{aligned}, \quad 
\begin{aligned}
\begin{tikzpicture}
    \draw[black, line width=1.0pt] (0,0) -- (0,-0.7);    
    \filldraw[red!30, draw=black] 
       (0,0) circle (0.2);
    \node at (0, 0) {$\scriptstyle \psi$};
\end{tikzpicture}
\end{aligned},\quad 
\begin{aligned}
\begin{tikzpicture}
    \draw[black, line width=1.0pt] (0,0) -- (0,0.7);    
    \filldraw[red!30, draw=black] 
       (0,0.7) circle (0.2);
    \node at (0, 0.7) {$\scriptstyle \psi$};
      \filldraw[green!30, draw=black] 
       (0,0) circle (0.2);
    \node at (0, 0) {$\scriptstyle h$};
\end{tikzpicture}
\end{aligned}.
\end{equation}
The diagram for $H$ is read upwards, while the diagram for $\hat{H}$ is read downwards.
The tensor product $h\otimes g$ (or $\psi\otimes \phi$) is represented by concatenating two diagrams and swapping of elements is represented as 
\begin{equation}
\tau (h\otimes g)=
\begin{aligned}
     \begin{tikzpicture}
    \draw[thick] (0, 2) .. controls (0, 1.5) and (1, 1.5) .. (1, 1);
    \draw[thick] (1, 2) .. controls (1, 1.5) and (0, 1.5) .. (0, 1);
    \filldraw[green!30, draw=black] 
       (0,0.8) circle (0.2);
    \node at (0, 0.8) {$\scriptstyle h$};
    \filldraw[green!30, draw=black] 
       (1,0.8) circle (0.2);
    \node at (1, 0.8) {$\scriptstyle g$};
\end{tikzpicture}  
\end{aligned},
\quad 
\tau (\psi\otimes \phi)=
\begin{aligned}
\begin{tikzpicture}
    \draw[thick] (0, 2) .. controls (0, 1.5) and (1, 1.5) .. (1, 1);
    \draw[thick] (1, 2) .. controls (1, 1.5) and (0, 1.5) .. (0, 1);
    \filldraw[red!30, draw=black] 
       (0,2.2) circle (0.2);
    \node at (0, 2.2) {$\scriptstyle \psi$}; 
    \filldraw[red!30, draw=black] 
        (1,2.2) circle (0.2);
    \node at (1, 2.2) {$\scriptstyle \phi$};
\end{tikzpicture}
\end{aligned},
\end{equation}
where $\tau = \tau^{-1}$ denotes the swapping map in $\mathsf{Vect}_{\Cbb}$, allowing us to avoid specifying upper or lower crossings.
The multiplication in $H$ and $\hat{H}$ are represented as 
\begin{equation}
  h\cdot g=  \begin{aligned}
			\begin{tikzpicture}
				 \draw[black, line width=1.0pt]  (-0.5, 0) .. controls (-0.4, 1) and (0.4, 1) .. (0.5, 0);
				 \draw[black, line width=1.0pt]  (0,0.75)--(0,1.15);
			 \filldraw[green!30, draw=black] 
       (-0.5,-0.2) circle (0.2);
    \node at (-0.5, -0.2) {$\scriptstyle h$};
 \filldraw[green!30, draw=black] 
       (0.5,-0.2) circle (0.2);
    \node at (0.5, -0.2) {$\scriptstyle g$};
				\end{tikzpicture}
			\end{aligned},\quad 
   \psi\cdot \phi=
   \begin{aligned}
				\begin{tikzpicture}
					\draw[black, line width=1.0pt]  (-0.5, 1) .. controls (-0.4, 0) and (0.4, 0) .. (0.5, 1);
					\draw[black, line width=1.0pt]  (0,0.23)--(0,-0.23);
					\filldraw[red!30, draw=black] 
        (-0.5,1.2) circle (0.2);
         \node at (-0.5, 1.2) {$\scriptstyle \psi$};
    \filldraw[red!30, draw=black] 
       (0.5,1.2) circle (0.2);
           \node at (0.5, 1.2) {$\scriptstyle \phi$};
				\end{tikzpicture}
			\end{aligned}.
\end{equation}
The comultiplication can be represented as
\begin{equation}
  \Delta(h)=   \begin{aligned}
				\begin{tikzpicture}
					\draw[black, line width=1.0pt]  (-0.5, 1) .. controls (-0.4, 0) and (0.4, 0) .. (0.5, 1);
					\draw[black, line width=1.0pt]  (0,0.23)--(0,-0.23);
				 \filldraw[green!30, draw=black] 
        (0,-0.2) circle (0.2);
         \node at (0, -.2) {$\scriptstyle h$};
                   \node at (-0.5, 1.3) {$\scriptstyle h^{\cone}$};
	                   \node at (0.5, 1.3) {$\scriptstyle h^{\ctwo}$};			
    \end{tikzpicture}
			\end{aligned},\quad 
   \Delta(\psi)=\begin{aligned}
			\begin{tikzpicture}
				 \draw[black, line width=1.0pt]  (-0.5, 0) .. controls (-0.4, 1) and (0.4, 1) .. (0.5, 0);
				 \draw[black, line width=1.0pt]  (0,0.75)--(0,1.15);
			 \filldraw[red!30, draw=black] 
        (0,1.2) circle (0.2);
         \node at (0, 1.2) {$\scriptstyle \psi$};
          \node at (-0.5, -.3) {$\scriptstyle \psi^{\cone}$};
              \node at (0.5, -.3) {$\scriptstyle \psi^{\ctwo}$};
				\end{tikzpicture}
			\end{aligned}.
\end{equation}
The counit $\varepsilon$ of $H$ is an element of $\hat{H}$ and it is the unit element: $1_{\hat{H}}=\varepsilon$. The unit $1_H\in H$ is the counit of $\hat{H}$ in the sense that $\hat{\varepsilon}(\psi)=\psi(1_{H})$. Both of them are removable under the respective multiplication map.
The antipodes of $H$ and $\hat{H}$ are represented as 
\begin{equation}
S(h)=
    \begin{aligned}
\begin{tikzpicture}
    \draw[black, line width=1.0pt] (0,0) -- (0,0.7);    
    \filldraw[green!30, draw=black] 
       (0,-0.2) circle (0.2);
    \node at (0, -0.2) {$\scriptstyle h$};
        \filldraw[black] (0, 0.35) circle (2pt);  % Shaded dot on the line
\end{tikzpicture}
\end{aligned}, \quad  \hat{S}(\psi)=
\begin{aligned}
\begin{tikzpicture}
    \draw[black, line width=1.0pt] (0,0) -- (0,0.7); 
      \filldraw[black] (0, 0.35) circle (2pt);  % Shaded dot on the line
    \filldraw[red!30, draw=black] 
      (0,0.9) circle (0.2);
    \node at (0, 0.9) {$\scriptstyle \psi$};
\end{tikzpicture}
\end{aligned}.
\end{equation}
For weak Hopf algebra, $S^{-1}$ is usually not the same as $S$, we can similarly denote $S^{-1}$ as a 
\begin{equation}
S^{-1}(h)=
    \begin{aligned}
\begin{tikzpicture}
    \draw[black, line width=1.0pt] (0,0) -- (0,0.7);    
    \filldraw[green!30, draw=black] 
       (0,-0.2) circle (0.2);
    \node at (0, -0.2) {$\scriptstyle h$};
        \filldraw[white, draw=black] (0, 0.35) circle (2pt);  % Shaded dot on the line
\end{tikzpicture}
\end{aligned}\,\,, \quad  \hat{S}^{-1}(\psi)=
\begin{aligned}
\begin{tikzpicture}
    \draw[black, line width=1.0pt] (0,0) -- (0,0.7); 
      \filldraw[white, draw=black] (0, 0.35) circle (2pt);  % Shaded dot on the line
    \filldraw[red!30, draw=black] 
      (0,0.9) circle (0.2);
    \node at (0, 0.9) {$\scriptstyle \psi$};
\end{tikzpicture}
\end{aligned},\quad 
\begin{aligned}
\begin{tikzpicture}
    \draw[black, line width=1.0pt] (0,0) -- (0,1); 
    \filldraw[white, draw=black] (0, 0.35) circle (2pt);  % Shaded dot on the line
    \filldraw[black] (0, 0.65) circle (2pt);  % Shaded dot on the line
\end{tikzpicture}
\end{aligned}=
\begin{aligned}
\begin{tikzpicture}
    \draw[black, line width=1.0pt] (0,0) -- (0,1); 
   % \filldraw[white, draw=black] (0, 0.35) circle (2pt);  % Shaded dot on the line
   % \filldraw[black] (0, 0.65) circle (2pt);  % Shaded dot on the line
\end{tikzpicture}
\end{aligned}=
\begin{aligned}
\begin{tikzpicture}
    \draw[black, line width=1.0pt] (0,0) -- (0,1); 
    \filldraw[black] (0, 0.35) circle (2pt);  % Shaded dot on the line
    \filldraw[white, draw=black] (0, 0.65) circle (2pt);  % Shaded dot on the line
\end{tikzpicture}
\end{aligned}\,\, .
\end{equation}
For a weak Kac algebra (e.g., a semisimple Hopf algebra), one has $S^{-1} = S$, so it is unnecessary to distinguish between them.

The pairing diagram introduced above has the convenient property of satisfying the pairing conditions (Eqs.~\eqref{eq:pair1}–\eqref{eq:pair5}). Eq.~\eqref{eq:pair1} can be represented as 
\begin{equation}
     \begin{aligned}
				\begin{tikzpicture}
					\draw[black, line width=1.0pt]  (-0.5, 1) .. controls (-0.4, 0) and (0.4, 0) .. (0.5, 1);
					\draw[black, line width=1.0pt]  (0,0.23)--(0,-0.23);
				 \filldraw[green!30, draw=black] 
         (0,-.2) circle (0.2);
         \node at (0, -.2) {$\scriptstyle x$};
        \filldraw[red!30, draw=black] 
         (-0.5,1.2) circle (0.2);
         \node at (-0.5, 1.2) {$\scriptstyle \varphi$};
    \filldraw[red!30, draw=black] 
        (0.5,1.2) circle (0.2);
           \node at (0.5, 1.2) {$\scriptstyle \psi$};	
    \end{tikzpicture}
	\end{aligned}.
\end{equation}
When read from top to bottom, this represents the left-hand side of Eq.~\eqref{eq:pair1}, and when read from bottom to top, it represents the right-hand side of Eq.~\eqref{eq:pair1}. Similarly, Eqs.~\eqref{eq:pairing-psi} and \eqref{eq:pair5} can be represented as
\begin{equation}
    \begin{aligned}
			\begin{tikzpicture}
				 \draw[black, line width=1.0pt]  (-0.5, 0) .. controls (-0.4, 1) and (0.4, 1) .. (0.5, 0);
				 \draw[black, line width=1.0pt]  (0,0.75)--(0,1.15);
			 \filldraw[red!30, draw=black] 
                (0,1.2) circle (0.2);
         \node at (0, 1.2) {$\scriptstyle \varphi$};
          \filldraw[green!30, draw=black] 
                (-0.5,-.2) circle (0.2);
    \node at (-0.5, -0.2) {$\scriptstyle x$};
 \filldraw[green!30, draw=black] 
                (0.5,-.2) circle (0.2);
    \node at (0.5, -0.2) {$\scriptstyle y$};
				\end{tikzpicture}
			\end{aligned}, \quad \begin{aligned}
\begin{tikzpicture}
    \draw[black, line width=1.0pt] (0,0) -- (0,0.7);    
    \filldraw[red!30, draw=black] 
                      (0,.9) circle (0.2);
    \node at (0, 0.9) {$\scriptstyle \varphi$};
      \filldraw[green!30, draw=black] 
                       (0,-.2) circle (0.2);
    \node at (0, -0.2) {$\scriptstyle x$};
       \filldraw[black] (0, 0.35) circle (2pt);
\end{tikzpicture}
\end{aligned},
\quad \begin{aligned}
\begin{tikzpicture}
    \draw[black, line width=1.0pt] (0,0) -- (0,0.7);    
    \filldraw[red!30, draw=black] 
                      (0,.9) circle (0.2);
    \node at (0, 0.9) {$\scriptstyle \varphi$};
      \filldraw[green!30, draw=black] 
                       (0,-.2) circle (0.2);
    \node at (0, -0.2) {$\scriptstyle x$};
       \filldraw[white, draw= black] (0, 0.35) circle (2pt);
\end{tikzpicture}
\end{aligned}.
\end{equation}
Notice that, from Eq.~\eqref{eq:pair5}, we can derive $\langle \hat{S}^{-1}(\varphi), x \rangle = \langle \varphi, S^{-1}(x) \rangle$.

Using the pairing, one can construct a general weak Hopf tensor network by first applying comultiplication to elements, then multiplying certain components, and finally implementing the pairing.

The tensor network representation presented here is more compact and closely resembles a string diagram compared to that in Ref.~\cite{jia2024generalized}, since we do not express every structure explicitly in terms of tensors. While this is convenient for calculations, the entanglement pattern cannot be read out directly as in the usual tensor network formalism. To obtain a representation that encodes entanglement features, one needs to employ a basis expansion and express all structures in terms of that basis. A detailed discussion of this point will be provided shortly. This representation can be effectively employed to solve the cluster ladder model (see Sections~\ref{sec:latticeI} and~\ref{sec:latticeII} for details).  
It is also worth noting that our tensor network representation differs from that in Ref.~\cite{molnar2022matrix}, where a matrix-product operator representation of the weak Hopf algebra is introduced.

\subsection{Weak Hopf tensor network states and operators}

We now provide further details on weak Hopf tensor networks and their one-dimensional counterparts, namely weak Hopf matrix-product states and matrix-product operators constructed from weak Hopf structures.

Consider weak Hopf algebra $H$ with basis $\{v_i\}$ (for dual weak Hopf algebra, we will fix the dual basis $\{\widehat{v}_j$\} with $\widehat{v}_j(v_i)=\delta_{ij}$). We can represent a element $x =\sum_i x_i v_i\in H$ as a tensor:
\begin{equation}
  x= \sum_i x_iv_i 
  =   \begin{aligned}
\begin{tikzpicture}
    \draw[black, line width=1.0pt] (0,0) -- (0,0.7);    
   \filldraw[green!30, draw=black] (0,0) circle (0.2);
    \node at (0, 0) {$\scriptstyle x$};
\end{tikzpicture}
\end{aligned}.
\end{equation}
This means that we assume the legs represent the degrees of freedom in the basis $\{v_i\}$.
In the given basis, we can represent comultiplication as a structure matrix
\begin{equation}
    \Delta(v_i)= \sum_{j,k} C^{jk}_i v_j \otimes v_k,
\end{equation}
thus $\Delta(x)=\sum_{i}x_iC_{i}^{jk}v_j \otimes v_k$.
Similarly, coassociativity guarantees that we can introduce
\begin{equation}
    \Delta_{n-1}(v_i)=\sum_{j,k,\cdots, l}C_i^{jk\cdots l} v_j \otimes v_k \otimes \cdots \otimes v_{l},
\end{equation}
the corresponding structure tensor will have $n$ legs
\begin{equation}
C_i^{jk\cdots l}
=  \begin{aligned}
        \begin{tikzpicture}
    % 定义张量的样式
\tikzset{
  tensor/.style={
    draw=black,
    fill=gray!30,
    regular polygon,
    regular polygon sides=3,
    shape border rotate=180,
    inner sep=0pt,
    minimum size=0.65cm
  }
}
    \tikzstyle{leg} = [line width=1.0pt]
    % 绘制张量
    \node[tensor] (T) at (0,0) {$C$};
    % 绘制上下两个 leg，颜色为深蓝色
   \draw[leg, color=black] (T.north) ++(-0.3,0) -- ++(0,0.4);
      \draw[leg, color=black] (T.north) ++(0.3,0) -- ++(0,0.4);
        \draw[leg, color=black] (T.south)   -- ++(0,-0.4);
    \node at ($(T.north) + (0,0.2)$) {\dots};
\end{tikzpicture}
    \end{aligned}.
\end{equation}
We adopt the convention that the tensor network is oriented upwards.
When applying the comultiplication $n-1$ times, $\Delta_{n-1}(x) = \sum_{(x)} \, x^{\cone} \otimes \cdots \otimes x^{(n)}$ can be represented as
\begin{equation}\label{eq:nx}
\Delta_{n-1}(x)
=  \begin{aligned}
        \begin{tikzpicture}
\draw[black, line width=1.0pt]  (-0.5, 1) .. controls (-0.4, 0) and (0.4, 0) .. (0.5, 1);
					\draw[black, line width=1.0pt]  (0,0.23)--(0,-0.23);
                   	\draw[black, line width=1.0pt]  (0,0.23)--(0,1);
				 \filldraw[green!30, draw=black] 
        (0,-0.2) circle (0.2);
         \node at (0, -.2) {$\scriptstyle x$};
                   \node at (-0.5, 1.3) {$\scriptstyle x^{\cone}$};
                    \node at (0, 1.3) {$\scriptstyle \cdots$};	
                    \node at (-0.2, .6) {$\scriptstyle \cdots$};	
                                        \node at (.2, .6) {$\scriptstyle \cdots$};	
	                   \node at (0.5, 1.3) {$\scriptstyle x^{(n)}$};		
\end{tikzpicture}
    \end{aligned}= \begin{aligned}
        \begin{tikzpicture}
    % 定义张量的样式
\tikzset{
  tensor/.style={
    draw=black,
    fill=gray!30,
    regular polygon,
    regular polygon sides=3,
    shape border rotate=180,
    inner sep=0pt,
    minimum size=0.65cm
  }
}
    \tikzstyle{leg} = [line width=1.0pt]
    % 绘制张量
    \node[tensor] (T) at (0,0) {$C$};
    % 绘制上下两个 leg，颜色为深蓝色
   \draw[leg, color=black] (T.north) ++(-0.3,0) -- ++(0,0.4);
      \draw[leg, color=black] (T.north) ++(0.3,0) -- ++(0,0.4);
        \draw[leg, color=black] (T.south)   -- ++(0,-0.4);
    \node at ($(T.north) + (0,0.2)$) {\dots};
     \filldraw[green!30, draw=black] 
        (0,-1) circle (0.2);
    \node at (0, -1) {$\scriptstyle x$};
\end{tikzpicture}
    \end{aligned}.
\end{equation}
The legs represent components $x^{(1)}, \cdots, x^{(n)}$ in the comultiplication, thus the order of these legs are crucial. 

The comultiplication generates entanglement. For instance, consider the group algebra $\Cbb[G]$ and the element 
$|x\rangle = \frac{1}{\sqrt{|G|}} \sum_{g \in G} |g\rangle$, which can be regarded as a state. After applying the comultiplication, we obtain a generalized Greenberger–Horne–Zeilinger (GHZ) state:
\begin{equation}
    \Delta_{n-1}(|x\rangle) = \frac{1}{\sqrt{|G|}} \sum_{g \in G} |g\rangle \otimes \cdots \otimes |g\rangle,
\end{equation}
which is known to be a highly entangled state. The comultiplication for a group-valued qudit acts as a quantum cloning operation on the basis consisting of group elements.

We can also introduce the multiplication structure matrix
\begin{equation}
    v_i\cdot v_j=\sum_k A_{ij}^k v_k,
\end{equation}
with which we have $x\cdot y=\sum_{i,j,k}x_i y_j A_{ij}^k v_k$.
The associativity of multiplication allows us to introduce $A_{ij\cdots k}^l$ by $v_i v_j \cdots v_k=\sum_l A_{ij\cdots k}^lv_l$ without ambiguity, we draw it as
\begin{equation}
A^l_{ij\cdots k}
=  \begin{aligned}
        \begin{tikzpicture}
    % 定义张量的样式
\tikzset{
  tensor/.style={
    draw=black,
    fill=gray!30,
    regular polygon,
    regular polygon sides=3,
    shape border rotate=0, % 顶点在上，底边在下
    minimum size=0.6cm,   % circumscribed circle diameter
    inner sep=0pt          % remove internal padding
  }
}

    \tikzstyle{leg} = [line width=1.0pt]
    % 绘制张量
    \node[tensor] (T) at (0,0) {$A$};
    % 绘制上下两个 leg，颜色为深蓝色
   \draw[leg, color=black] (T.south) ++(0.3,0) -- ++(0,-0.4);
      \draw[leg, color=black] (T.north)  -- ++(0,0.4);
        \draw[leg, color=black] (T.south) ++(-0.3,0) -- ++(0,-0.4);
    \node at ($(T.south) + (0,-0.2)$) {\dots};
\end{tikzpicture}
    \end{aligned}.
\end{equation}
The antipode matrix is defined by $S(v_i)=\sum_j S_{i}^j v_j$, which can be represented diagrammatically as
\begin{equation}
S^j_{i}
= \begin{aligned}
\begin{tikzpicture}
    \draw[black, line width=1.0pt] (0,0) -- (0,1); 
    \filldraw[black] (0, 0.5) circle (2pt);  % Shaded dot on the line
   % \filldraw[white, draw=black] (0, 0.65) circle (2pt);  % Shaded dot on the line
\end{tikzpicture}
\end{aligned}
\end{equation}
The unit is given by 
\(
1_H = \sum_i \iota_i v_i,
\) 
and the counit by 
\(
\varepsilon(v_j) = \epsilon_j.
\) 
Accordingly, they can both be represented as vectors:
\begin{equation}
  1_H=  \begin{aligned}
\begin{tikzpicture}
    \draw[black, line width=1.0pt] (0,0) -- (0,0.7);    
   \filldraw[green!30, draw=black] (0,0) circle (0.2);
    \node at (0, 0) {$\scriptstyle  1_H$};
\end{tikzpicture}
\end{aligned},\quad
\varepsilon =
\begin{aligned}
\begin{tikzpicture}
    \draw[black, line width=1.0pt] (0,0) -- (0,-0.7);    
    \filldraw[red!30, draw=black] 
       (0,0) circle (0.2);
    \node at (0, 0) {$\scriptstyle \varepsilon$};
\end{tikzpicture}
\end{aligned}.
\end{equation}

The axioms of a weak Hopf algebra can be reformulated in terms of these structure tensors as follows (where we adopt the Einstein summation convention):
\begin{itemize}
\item \emph{Algebra:} $1x = x1 = x$ as $\iota_i A_{ij}^{k} x_j = x_i A_{ij}^{k} \iota_j = x_k$; associativity as $A_{ij}^{k} A_{kl}^{p} = A_{jl}^{m} A_{im}^{p}$.

\item \emph{Coalgebra:} $\varepsilon(x^{\langle 1\rangle})x^{\langle 2\rangle} = x = x^{\langle 1\rangle} \varepsilon(x^{\langle 2\rangle})$ as $x_i C_i^{jk} \epsilon_j = x_k = x_i C_i^{kj} \epsilon_j$; coassociativity as $C_i^{jk} C_j^{st} = C_i^{jk} C_k^{st}$.

% \item \emph{Comultiplication of identity:} 
% \[
% \Delta(1) = (\iota_p C_p^{ab}) e_a \otimes e_b.
% \]

\item \emph{Compatibility of multiplication and comultiplication:} 
\(
C_k^{mn} A_{ij}^{k} = C_i^{ab} C_j^{cd} A_{ac}^{m} A_{bd}^{n}.
\)

\item \emph{Weak counit:} 
\(
\epsilon_z A_{ijk}^z 
= C_j^{ab} \; (\epsilon_s A_{i a}^{s}) \; (\epsilon_t A_{b k}^{t})
= C_j^{ab} \; (\epsilon_s A_{i b}^{s}) \; (\epsilon_t A_{a k}^{t}).
\)

\item \emph{Weak identity:} 
\(
\iota_p C_p^{mnq}
= (\iota_i C_i^{mb}) \, (\iota_j C_j^{cq}) \, A_{bc}^{n} 
= (\iota_k C_k^{sq}) \, (\iota_l C_l^{mp}) \, A_{sp}^{n}.
\)

\item \emph{Antipode axioms:}
\(
x_iC_i^{ab} S_a^r A_{rb}^{k} = \epsilon_s x_i A_{ia}^{s} (\iota_p C_p^{ka})\), 
\(x_iC_i^{ab} A_{ar}^{k} S_b^r = (\iota_p C_p^{bk}) \epsilon_s A_{bi}^{s}\),
\(
x_i S_i^j = C_i^{abc} S_a^s S_c^r A_{sbr}^{j}.
\)
\end{itemize}
These axioms are particularly useful when performing calculations with weak Hopf tensor networks.

The generalized Pauli-$X$ operators can be represented as
\begin{equation}
    \XR_{g}=\begin{aligned}
        \begin{tikzpicture}
    % 定义张量的样式
    \tikzstyle{tensor} = [draw=black, fill=yellow!30, rounded corners, minimum size=0.8cm]
    \tikzstyle{leg} = [line width=1.5pt]
    % 绘制张量
    \node[tensor] (T) at (0,0) {$\XR_g$};
    % 绘制上下两个 leg，颜色为深蓝色
   \draw[leg, color=blue!70!black] (T.north) -- ++(0,0.4);
    \draw[leg, color=blue!70!black] (T.south) -- ++(0,-0.4);
    % 绘制左右两个 leg，颜色为浅棕色
   % \draw[leg, color=lightgray] (T.west) -- ++(-0.4,0);
   % \draw[leg, color=lightgray] (T.east) -- ++(0.4,0);
\end{tikzpicture}
    \end{aligned}:=
    \begin{aligned}
        \begin{tikzpicture}
           \draw[fill=gray!30]  (0,0) -- (1,0) -- (0.5,0.5) -- cycle;
           \draw[line width=.6pt,black] (0.2,0)--(0.2,-0.6);
           \draw[line width=.6pt,black] (0.8,0)--(0.8,-1);
           \draw[line width=.6pt,black] (0.5,0.5)--(0.5,0.9);
           \draw[black,fill=green!30] (0.2,-.6) circle (0.2);   
            \node[ line width=0.6pt, draw opacity=0.5] (a) at (0.2,-0.6){$\scriptstyle g$};
                               \node at (0.5, 0.2) {$\scriptstyle A$};
        \end{tikzpicture}
    \end{aligned}  ,\quad 
   \XL_g= 
    \begin{aligned}
     \begin{tikzpicture}
    % 定义张量的样式
    \tikzstyle{tensor} = [draw=black, fill=yellow!30, rounded corners, minimum size=0.8cm]
    \tikzstyle{leg} = [line width=1.5pt]
    % 绘制张量
    \node[tensor] (T) at (0,0) {$\XL_g$};
    % 绘制上下两个 leg，颜色为深蓝色
   \draw[leg, color=blue!70!black] (T.north) -- ++(0,0.4);
    \draw[leg, color=blue!70!black] (T.south) -- ++(0,-0.4);
    % 绘制左右两个 leg，颜色为浅棕色
   % \draw[leg, color=lightgray] (T.west) -- ++(-0.4,0);
   % \draw[leg, color=lightgray] (T.east) -- ++(0.4,0);
\end{tikzpicture}
    \end{aligned}:=
    \begin{aligned}
        \begin{tikzpicture}
           \draw[fill=gray!30]  (0,0) -- (1,0) -- (0.5,0.5) -- cycle;
           \draw[line width=.6pt,black] (0.2,0)--(0.2,-1);
           \draw[line width=.6pt,black] (0.8,0)--(0.8,-.6);
           \draw[line width=.6pt,black] (0.5,0.5)--(0.5,.9);
            \filldraw[white, draw= black] (0.8, -0.2) circle (2pt);
                   \draw[black,fill=green!30] (0.8,-.6) circle (0.2);  
            \node[ line width=0.6pt, draw opacity=0.5] (a) at (0.8,-.6){$\scriptstyle g$};
               \node at (0.5, 0.2) {$\scriptstyle A$};
        \end{tikzpicture}
    \end{aligned}
\end{equation}
where we use blue thick legs to emphasize the weak Hopf input and output qudits.
Using the comultiplication, we can construct symmetry MPOs (which, as we will see later, serve as the symmetry operators for the weak Hopf cluster state model):
\begin{align}
 W_g
=
   \begin{aligned}
\begin{tikzpicture}
    % 定义张量的样式
    \tikzstyle{tensor} = [draw=black, fill=yellow!30, rounded corners, minimum size=0.8cm]
    \tikzstyle{leg} = [line width=1.5pt]
    % 绘制第一个张量
    \node[tensor] (T1) at (0,0) {$\XR_{g^{\cone},i_1}$};
    \draw[leg, color=blue!70!black] (T1.north) -- ++(0,0.4);
    \draw[leg, color=blue!70!black] (T1.south) -- ++(0,-0.4);
    \draw[leg, color=cyan!30!white] (T1.west) -- ++(-0.45,0);
    % \draw[leg, color=lightgray] (T1.east) -- ++(0.2,0);
    % 绘制省略号
    \node at ($(T1.east) + (2,0.2)$) {$\cdots$};
    % 绘制第二个张量
    \node[tensor] (T2) at ($(T1.east) + (0.9,0)$) {$\XR_{g^{\ctwo},i_2}$};
    \draw[leg, color=blue!70!black] (T2.north) -- ++(0,0.4);
    \draw[leg, color=blue!70!black] (T2.south) -- ++(0,-0.4);
    \draw[leg, color=cyan!30!white] (T2.west) --(T1.east);
    \draw[leg, color=cyan!30!white] (T2.east) -- ++(0.2,0) coordinate (conn3);
    % 绘制第三个张量
    \node[tensor] (T3) at ($(T2.east) + (1.5,0)$) {$\XR_{g^{(n)},i_n}$};
    \draw[leg, color=blue!70!black] (T3.north) -- ++(0,0.4);
    \draw[leg, color=blue!70!black] (T3.south) -- ++(0,-0.4);
    \draw[leg, color=cyan!30!white] (T3.west) -- ++(-0.2,0);
    \draw[leg, color=cyan!30!white] (T3.east) -- ++(0.4,0);
    % 连接第一个张量和第二个张量
  % \draw[leg, color=lightgray] (conn1) -- (T2.west); 
    % 连接左右两个 leg，横平竖直的回路
    \draw[leg, color=cyan!30!white] 
        (T1.west) ++(-0.43,0) -- ++(0,-0.6) -- ++(6.15,0) -- ++(0,0.6) ;
\end{tikzpicture}
    \end{aligned}:= 
\begin{aligned}
\begin{tikzpicture}
  % first
  \begin{scope}[xshift=0cm]
    \draw[fill=gray!30]  (0,0) -- (1,0) -- (0.5,0.5) -- cycle;
    %\draw[line width=.6pt,black] (0.2,0)--(0.2,-0.6);
    \draw[line width=.6pt,black] (0.8,0)--(0.8,-1);
    \draw[line width=.6pt,black] (0.5,0.5)--(0.5,0.9);
    % \draw[black,fill=green!30] (0.2,-.6) circle (0.2);   
    % \node at (0.2,-0.6){$\scriptstyle g$};
    \node at (0.5, 0.2) {$\scriptstyle A$};
  \end{scope}
  % second (shifted by 1.5cm)
  \begin{scope}[xshift=1.5cm]
    \draw[fill=gray!30]  (0,0) -- (1,0) -- (0.5,0.5) -- cycle;
    %\draw[line width=.6pt,black] (0.2,0)--(0.2,-0.6);
    \draw[line width=.6pt,black] (0.8,0)--(0.8,-1);
    \draw[line width=.6pt,black] (0.5,0.5)--(0.5,0.9);
    % \draw[black,fill=green!30] (0.2,-.6) circle (0.2);   
    % \node at (0.2,-0.6){$\scriptstyle g$};
    \node at (0.5, 0.2) {$\scriptstyle A$};
  \end{scope}
  % third (shifted by 2.5cm)
  \begin{scope}[xshift=3.5cm]
    \draw[fill=gray!30]  (0,0) -- (1,0) -- (0.5,0.5) -- cycle;
    %\draw[line width=.6pt,black] (0.2,0)--(0.2,-0.6);
    \draw[line width=.6pt,black] (0.8,0)--(0.8,-1);
    \draw[line width=.6pt,black] (0.5,0.5)--(0.5,0.9);
    % \draw[black,fill=green!30] (0.2,-.6) circle (0.2);   
    % \node at (0.2,-0.6){$\scriptstyle g$};
    \node at (0.5, 0.2) {$\scriptstyle A$};
  \end{scope}
    \begin{scope}[yshift=-1cm,xshift=3.4cm]
    % 定义张量的样式
\tikzset{
  tensor/.style={
    draw=black,
    fill=gray!30,
    regular polygon,
    regular polygon sides=3,
    shape border rotate=180,
    inner sep=0pt,
    minimum size=0.85cm,
    % xscale=1.5 % make it wider horizontally
  }
}
    \tikzstyle{leg} = [line width=.6pt]
    % 绘制张量
    \node[tensor] (T) at (0,0) {$C$};
    % 绘制上下两个 leg，颜色为深蓝色
   \draw[leg, color=black] (T.north) ++(-0.3,0) -- ++(0,0.2) -- ++ (-3,0) -- ++ (0,.59);
   \draw[leg, color=black] (T.north) ++(-0.1,0) -- ++(0,0.4) -- ++ (-1.5,0) -- ++ (0,.39);
      \draw[leg, color=black] (T.north) ++(0.3,0) -- ++(0,0.77);
        \draw[leg, color=black] (T.south)   -- ++(0,-0.4);
    \node at ($(T.north) + (0.13,0.2)$) {$\scriptstyle  \cdots$};
     \filldraw[green!30, draw=black] 
        (0,-0.8) circle (0.2);
    \node at (0, -.8) {$\scriptstyle g$};
\end{scope}
  \node at (3, .2) {$\cdots$};
\end{tikzpicture}
\end{aligned}
\end{align}
where light cyan legs represents comultiplication $\Delta_{n-1}(g)=\sum_{k_1\cdots k_n} C_j^{k_1\cdots k_n}g_j v_{k_1}\otimes \cdots \otimes v_{k_n}$.  It is clear that $W_g W_h = W_{gh}$. In the MPO representation, this multiplication corresponds to the vertical gluing of MPOs for $W_g$ and $W_h$. Note also that $W_{1_H} \neq I$, since for a general weak Hopf algebra we have $\Delta(1_H) \neq 1_H \otimes 1_H$\footnote{In many situations, this serves as an indicator of anomaly, since in the anomaly-free case the corresponding symmetry is a Hopf symmetry. Indeed, for a fusion category $\eC$ admitting a fiber functor $F:\eC \to \Vect$, there always exists a Hopf algebra $H$ such that $\eC \simeq \Rep(H)$, in which case $W_{1_H} = I$. A weak Hopf symmetry is a Hopf symmetry if and only if $\Delta(1_H) = 1_H \otimes 1_H$.}
In the group case, the operator $W_g$ becomes on-site, and the cyan legs disappear.

For given representation $\Gamma: H\to \End(\mathcal{V})$, the generalized Pauli-$Z$ operators can be represented as
\begin{equation}
    Z_{\Gamma}=\begin{aligned}
        \begin{tikzpicture}
    % 定义张量的样式
    \tikzstyle{tensor} = [draw=black, fill=yellow!30, rounded corners, minimum size=0.8cm]
    \tikzstyle{leg} = [line width=1.5pt]

    % 绘制张量
    \node[tensor] (T) at (0,0) {$Z_{\Gamma}$};

    % 绘制上下两个 leg，颜色为深蓝色
    \draw[leg, color=blue!70!black] (T.north) -- ++(0,0.4);
    \draw[leg, color=blue!70!black] (T.south) -- ++(0,-0.4);

    % 绘制左右两个 leg，颜色为浅棕色
    \draw[leg, color=lightgray] (T.west) -- ++(-0.4,0);
    \draw[leg, color=lightgray] (T.east) -- ++(0.4,0);
\end{tikzpicture}
    \end{aligned}
    =\begin{aligned}
        \begin{tikzpicture}
           \draw[fill=gray!30]  (0,0) -- (1,0) -- (0.5,-0.5) -- cycle;
           \draw[line width=.6pt,black] (0.2,0)--(0.2,1);
           \draw[line width=.6pt,black] (0.8,0)--(0.8,0.6);
           \draw[line width=.6pt,black] (0.5,-0.5)--(0.5,-1.1);
          \draw[black,fill=gray!30] (0.6,.4) rectangle ++(0.4,0.4);   
           \draw[line width=.6pt,black,gray] (1,0.6)--(1.4,0.6);
            \draw[line width=.6pt,black,gray] (0.6,0.6)--(0.23,0.6);
            \draw[line width=.6pt,black,gray] (0,0.6)--(0.17,0.6);
            \node[ line width=0.6pt, dashed, draw opacity=0.5] (a) at (0.8,0.6){$\Gamma$};
             \node at (0.5, -0.2) {$ C$};
        \end{tikzpicture}
    \end{aligned},\quad 
    Z_{\Gamma}^{\ddagger}= 
    \begin{aligned}
              \begin{tikzpicture}
    % 定义张量的样式
    \tikzstyle{tensor} = [draw=black, fill=yellow!30, rounded corners, minimum size=0.8cm]
    \tikzstyle{leg} = [line width=1.5pt]

    % 绘制张量
    \node[tensor] (T) at (0,0) {$Z_{\Gamma}^{\ddagger}$};

    % 绘制上下两个 leg，颜色为深蓝色
    \draw[leg, color=blue!70!black] (T.north) -- ++(0,0.4);
    \draw[leg, color=blue!70!black] (T.south) -- ++(0,-0.4);
    % 绘制左右两个 leg，颜色为浅棕色
    \draw[leg, color=lightgray] (T.west) -- ++(-0.4,0);
    \draw[leg, color=lightgray] (T.east) -- ++(0.4,0);
\end{tikzpicture}
    \end{aligned}
    =
    \begin{aligned}
        \begin{tikzpicture}
           \draw[fill=gray!30]  (0,0) -- (1,0) -- (0.5,-0.5) -- cycle;
           \draw[line width=.6pt,black] (0.2,0)--(0.2,0.6);
           \draw[line width=.6pt,black] (0.8,0)--(0.8,1);
           \draw[line width=.6pt,black] (0.5,-0.5)--(0.5,-1.1);
            \draw[black,fill=gray!30] (0,.4) rectangle ++(0.4,0.4);   
           \draw[line width=.6pt,black,gray] (0.4,.6)--(.77,.6);
                      \draw[line width=.6pt,black,gray] (0.92,.6)--(.83,.6);
            \draw[line width=.6pt,black,gray] (0,.6)--(-.4,.6);
            \filldraw[black] (0.2, 0.22) circle (2pt);  % Shaded dot on the line
            \node[ line width=0.6pt, dashed, draw opacity=0.5] (a) at (0.2,0.6){$\Gamma$};
           % \node[ line width=0.6pt, dashed, draw opacity=0.5] (a) at (0.2,0.6){$S$};
            \node at (0.5, -0.2) {$ C$};
        \end{tikzpicture}
    \end{aligned}
\end{equation}
where blue legs represent weak Hopf qudits, and gray legs represent the virtual  representation space degrees of freedom.
The operator $\ZR_{\chi_{\Gamma}}=\Tr' Z_{\Gamma}$ and $\ZL_{\chi_{\Gamma}}=\Tr' Z_{\Gamma}^{\ddagger}$ can be represented as\footnote{In this work, we use $\Tr'$ to represent trace over representation space.}
\begin{equation}
    \ZR_{\chi_{\Gamma}}=\begin{aligned}
      \begin{tikzpicture}
    % 定义张量的样式
    \tikzstyle{tensor} = [draw=black, fill=yellow!30, rounded corners, minimum size=0.8cm]
    \tikzstyle{leg} = [line width=1.5pt]

    % 绘制张量
    \node[tensor] (T) at (0,0) {$Z_{\Gamma}$};

    % 绘制上下两个 leg，颜色为深蓝色
    \draw[leg, color=blue!70!black] (T.north) -- ++(0,0.4);
    \draw[leg, color=blue!70!black] (T.south) -- ++(0,-0.4);

    % 绘制左右两个 leg，颜色为浅棕色
    \draw[leg, color=lightgray] (T.west) -- ++(-0.455,0);
    \draw[leg, color=lightgray] (T.east) -- ++(0.4,0);

    % 连接左右两个 leg，横平竖直的回路
    \draw[leg, color=lightgray] 
        (T.west) ++(-0.43,0) -- ++(0,-0.6) -- ++(1.63,0) -- ++(0,0.6) ;
\end{tikzpicture}
    \end{aligned},
    \quad 
    \ZL_{\chi_{\Gamma}}= 
    \begin{aligned}
        \begin{tikzpicture}
    % 定义张量的样式
    \tikzstyle{tensor} = [draw=black, fill=yellow!30, rounded corners, minimum size=0.8cm]
    \tikzstyle{leg} = [line width=1.5pt]

    % 绘制张量
    \node[tensor] (T) at (0,0) {$Z_{\Gamma}^{\ddagger}$};

    % 绘制上下两个 leg，颜色为深蓝色
    \draw[leg, color=blue!70!black] (T.north) -- ++(0,0.4);
    \draw[leg, color=blue!70!black] (T.south) -- ++(0,-0.4);

    % 绘制左右两个 leg，颜色为浅棕色
    \draw[leg, color=lightgray] (T.west) -- ++(-0.455,0);
    \draw[leg, color=lightgray] (T.east) -- ++(0.4,0);

    % 连接左右两个 leg，横平竖直的回路
    \draw[leg, color=lightgray] 
        (T.west) ++(-0.43,0) -- ++(0,-0.6) -- ++(1.63,0) -- ++(0,0.6) ;
\end{tikzpicture}
    \end{aligned}
\end{equation}
If we treat $\Gamma_{\alpha\beta}$ as a function from weak Hopf algebra $H$ to $\mathbb{C}$, then $\ZR_{\Gamma_{\alpha\beta}}$ and $\ZL_{\Gamma_{\alpha\beta}}$ are well-defined and can be represented as
\begin{equation}
    \ZR_{\Gamma_{\alpha\beta}} =
    \begin{aligned}
\begin{tikzpicture}
    % 定义张量的样式
    \tikzstyle{tensor} = [draw=black, fill=yellow!30, rounded corners, minimum size=0.8cm]
    \tikzstyle{leg} = [line width=1.5pt]
    \tikzstyle{endpoint} = [circle, fill=lightgray, minimum size=0.2cm, inner sep=0pt]

    % 绘制张量
    \node[tensor] (T) at (0,0) {$Z_{\Gamma}$};

    % 绘制上下两个 leg，颜色为深蓝色
    \draw[leg, color=blue!70!black] (T.north) -- ++(0,0.4);
    \draw[leg, color=blue!70!black] (T.south) -- ++(0,-0.4);

    % 绘制左右两个 leg，颜色为浅棕色
    \draw[leg, color=lightgray] (T.west) -- ++(-0.4,0);
    \draw[leg, color=lightgray] (T.east) -- ++(0.4,0);

    % 添加左右 leg 的端点
    \node[endpoint] (dotL) at ($(T.west) + (-0.4,0)$) {};
    \node[endpoint] (dotR) at ($(T.east) + (0.4,0)$) {};

    % 在端点下方添加小字体 \alpha 和 \beta
    \node[below=0.1cm of dotL] {\scriptsize $\alpha$};
    \node[below=0.1cm of dotR] {\scriptsize $\beta$};
\end{tikzpicture}
    \end{aligned},
    \quad 
    \ZL_{\Gamma_{\alpha\beta}} =
    \begin{aligned}
\begin{tikzpicture}
    % 定义张量的样式
    \tikzstyle{tensor} = [draw=black, fill=yellow!30, rounded corners, minimum size=0.8cm]
    \tikzstyle{leg} = [line width=1.5pt]
    \tikzstyle{endpoint} = [circle, fill=lightgray, minimum size=0.2cm, inner sep=0pt]

    % 绘制张量
    \node[tensor] (T) at (0,0) {$Z_{\Gamma}^{\ddagger}$};

    % 绘制上下两个 leg，颜色为深蓝色
    \draw[leg, color=blue!70!black] (T.north) -- ++(0,0.4);
    \draw[leg, color=blue!70!black] (T.south) -- ++(0,-0.4);

    % 绘制左右两个 leg，颜色为浅棕色
    \draw[leg, color=lightgray] (T.west) -- ++(-0.4,0);
    \draw[leg, color=lightgray] (T.east) -- ++(0.4,0);

    % 添加左右 leg 的端点
    \node[endpoint] (dotL) at ($(T.west) + (-0.4,0)$) {};
    \node[endpoint] (dotR) at ($(T.east) + (0.4,0)$) {};

    % 在端点下方添加小字体 \alpha 和 \beta
    \node[below=0.1cm of dotL] {\scriptsize $\alpha$};
    \node[below=0.1cm of dotR] {\scriptsize $\beta$};
\end{tikzpicture}
    \end{aligned},
\end{equation}
where gray dots represent the fixed leg labels.

An open $Z_{\Gamma}$-string MPO that acts on qudits labeled by $i_1,\cdots,i_n$ can be represented as
\begin{align}
W_{\Gamma} = Z_{\Gamma,i_1} \ast  Z_{\Gamma,i_2} \ast \cdots \ast  Z_{\Gamma,i_n} 
    =
 \begin{aligned}
\begin{tikzpicture}
    % 定义张量的样式
    \tikzstyle{tensor} = [draw=black, fill=yellow!30, rounded corners, minimum size=0.8cm]
    \tikzstyle{leg} = [line width=1.5pt]
    % 绘制第一个张量
    \node[tensor] (T1) at (0,0) {$Z_{\Gamma,i_1}$};
    \draw[leg, color=blue!70!black] (T1.north) -- ++(0,0.4);
    \draw[leg, color=blue!70!black] (T1.south) -- ++(0,-0.4);
    \draw[leg, color=lightgray] (T1.west) -- ++(-0.4,0);
    \draw[leg, color=lightgray] (T1.east) -- ++(0.4,0) coordinate (conn1);
    % 绘制省略号
    \node at ($(T1.east) + (1.95,0)$) {\dots};
    % 绘制第二个张量
    \node[tensor] (T2) at ($(T1.east) + (0.9,0)$) {$Z_{\Gamma,i_2}$};
    \draw[leg, color=blue!70!black] (T2.north) -- ++(0,0.4);
    \draw[leg, color=blue!70!black] (T2.south) -- ++(0,-0.4);
    \draw[leg, color=lightgray] (T2.west) -- ++(-0.4,0) coordinate (conn2);
    \draw[leg, color=lightgray] (T2.east) -- ++(0.2,0) coordinate (conn3);
    % 绘制第三个张量
    \node[tensor] (T3) at ($(T2.east) + (1.5,0)$) {$Z_{\Gamma,i_n}$};
    \draw[leg, color=blue!70!black] (T3.north) -- ++(0,0.4);
    \draw[leg, color=blue!70!black] (T3.south) -- ++(0,-0.4);
    \draw[leg, color=lightgray] (T3.west) -- ++(-0.2,0);
    \draw[leg, color=lightgray] (T3.east) -- ++(0.4,0);
      % 连接第一个张量和第二个张量
 %  \draw[leg, color=lightgray] (conn1)++(0.3,0) -- (T2.west)++(0.3,0); 
\end{tikzpicture}
    \end{aligned}.
\end{align}
A closed  $Z_{\Gamma}$-string MPO  can be obtained by taking trace over the representation space:
\begin{align}
W_{\Gamma} =  \Tr'  Z_{\Gamma,i_1} \ast  Z_{\Gamma,i_2} \ast \cdots \ast  Z_{\Gamma,i_n} 
=
   \begin{aligned}
\begin{tikzpicture}
    % 定义张量的样式
    \tikzstyle{tensor} = [draw=black, fill=yellow!30, rounded corners, minimum size=0.8cm]
    \tikzstyle{leg} = [line width=1.5pt]
    % 绘制第一个张量
    \node[tensor] (T1) at (0,0) {$Z_{\Gamma,i_1}$};
    \draw[leg, color=blue!70!black] (T1.north) -- ++(0,0.4);
    \draw[leg, color=blue!70!black] (T1.south) -- ++(0,-0.4);
    \draw[leg, color=lightgray] (T1.west) -- ++(-0.45,0);
    \draw[leg, color=lightgray] (T1.east) -- ++(0.4,0) coordinate (conn1);
    % 绘制省略号
    \node at ($(T1.east) + (1.92,0)$) {\dots};
    % 绘制第二个张量
    \node[tensor] (T2) at ($(T1.east) + (0.9,0)$) {$Z_{\Gamma,i_2}$};
    \draw[leg, color=blue!70!black] (T2.north) -- ++(0,0.4);
    \draw[leg, color=blue!70!black] (T2.south) -- ++(0,-0.4);
    \draw[leg, color=lightgray] (T2.west) -- ++(-0.4,0) coordinate (conn2);
    \draw[leg, color=lightgray] (T2.east) -- ++(0.2,0) coordinate (conn3);
    % 绘制第三个张量
    \node[tensor] (T3) at ($(T2.east) + (1.5,0)$) {$Z_{\Gamma,i_n}$};
    \draw[leg, color=blue!70!black] (T3.north) -- ++(0,0.4);
    \draw[leg, color=blue!70!black] (T3.south) -- ++(0,-0.4);
    \draw[leg, color=lightgray] (T3.west) -- ++(-0.2,0);
    \draw[leg, color=lightgray] (T3.east) -- ++(0.4,0);
    % 连接第一个张量和第二个张量
  % \draw[leg, color=lightgray] (conn1) -- (T2.west); 
    % 连接左右两个 leg，横平竖直的回路
    \draw[leg, color=lightgray] 
        (T1.west) ++(-0.42,0) -- ++(0,-0.6) -- ++(5.33,0) -- ++(0,0.6) ;
\end{tikzpicture}
    \end{aligned}.
\end{align}
A crucial property that we will use later is that $\chi_{\Gamma}$ belongs to the dual space $\hat{H}$. Applying comultiplication, we have
\begin{equation}
    \hat{\Delta}_{n-1}(\chi_{\Gamma}) = \sum_{(\chi_{\Gamma})} \chi_{\Gamma}^{(1)} \otimes \cdots \otimes \chi_{\Gamma}^{(n)},
\end{equation}
The closed $Z_{\Gamma}$-string operator is identical to the string operator derived from the comultiplication of the character:
\begin{equation}
    \Tr' Z_{\Gamma, i_1} \ast Z_{\Gamma, i_2} \ast \cdots \ast Z_{\Gamma, i_n} = \sum_{(\chi_{\Gamma})} \ZR_{\chi_{\Gamma}^{(1)},i_1} \otimes \ZR_{\chi_{\Gamma}^{(2)},i_2} 
\otimes \cdots \otimes \ZR_{\chi_{\Gamma}^{(n)},i_n}.
\end{equation}
This equality follows from the fact that $\hat{\Delta}_n(\chi_{\Gamma})$ acting on $g_1 \otimes \cdots \otimes g_n$ yields $\chi_{\Gamma}(g_1 \cdots g_n)$.
This property is fundamental for constructing the symmetry operators of the cluster ladder model and analyzing its properties. As we will see, the dual weak Hopf symmetries of the weak Hopf cluster ladder model necessarily rely on the comultiplication structure.

The $W_{\Gamma}$ operator satisfy the fusion rule of $\Rep(H)$ in the sense that 
\begin{equation}\label{eq:Wgamma}
    W_{\Gamma}W_{\Phi} =\sum_{\Psi\in \Irr(\Rep(H))} N_{\Gamma \Phi}^{\Psi} W_{\Psi}.
\end{equation}
In fact, for $\Gamma,\Phi\in \Rep(H)$, $ Z_{\Gamma}Z_{\Phi}$ can be represented as 
\begin{equation}
   \begin{aligned}
        \begin{tikzpicture}
           \draw[fill=gray!30]  (0,0) -- (1,0) -- (0.5,-0.5) -- cycle;
           \draw[line width=.6pt,black] (0.2,0)--(0.2,1);
           \draw[line width=.6pt,black] (0.8,0)--(0.8,0.6);
           \draw[line width=.6pt,black] (0.5,-0.5)--(0.5,-1.1);
          \draw[black,fill=gray!30] (0.6,.4) rectangle ++(0.4,0.4);   
           \draw[line width=.6pt,black,gray] (1,0.6)--(1.4,0.6);
            \draw[line width=.6pt,black,gray] (0.6,0.6)--(0.23,0.6);
            \draw[line width=.6pt,black,gray] (0,0.6)--(0.17,0.6);
         \draw[fill=gray!30]  (-0.3,1.5) -- (0.7,1.5) -- (0.2,1) -- cycle;
    \draw[line width=.6pt,black] (-0.1,1.5)--(-0.1,2.3);
           \draw[line width=.6pt,black] (0.5,1.5)--(0.5,1.9);
           \draw[black,fill=gray!30] (0.3,1.7) rectangle ++(0.4,0.4); 
          \draw[line width=.6pt,black,gray] (0.7,1.9)--(1.1,1.9);
          \draw[line width=.6pt,black,gray] (0.3,1.9)--(0,1.9);
            \node[ line width=0.6pt, dashed, draw opacity=0.5] (a) at (0.8,0.6){$\Phi$};
            \node[ line width=0.6pt, dashed, draw opacity=0.5] (a) at (0.5,1.9){$\Gamma$};
            \node at (0.5, -.18) {$ C$}; 
                        \node at (0.2, 1.28) {$ C$}; 
      %   \node[ line width=0.6pt, dashed, draw opacity=0.5] (a) at (0.7,-1){$k$};
        \end{tikzpicture}
    \end{aligned}
    = \begin{aligned}
        \begin{tikzpicture}
           \draw[fill=gray!30]  (0,0) -- (1,0) -- (0.5,-0.5) -- cycle;
           \draw[line width=.6pt,black] (0.2,0)--(0.2,0.4);
           \draw[line width=.6pt,black] (0.8,0)--(0.8,0.6);
           \draw[line width=.6pt,black] (0.5,-0.5)--(0.5,-1.1);
          \draw[black,fill=gray!30] (0.6,.4) rectangle ++(0.4,0.4);   
           \draw[line width=.6pt,black,gray] (1,0.6)--(1.4,0.6);
            \draw[line width=.6pt,black,gray] (0.6,0.6)--(0.5,0.6);
           % \draw[line width=.6pt,black,cyan] (0.5,0.6)--(0.5,0.8);
         \draw[fill=gray!30]  (-0.3,-1.1) -- (0.7,-1.1) -- (0.2,-1.6) -- cycle;
         \draw[line width=.6pt,black] (-0.1,-1.1)--(-0.1,1);
           \draw[black,fill=gray!30] (0,.2) rectangle ++(0.4,0.4); 
          \draw[line width=.6pt,black,gray] (0.4,0.45)--(0.5,0.45);
          %\draw[line width=.6pt,black,cyan] (0.5,.45)--(0.5,0.3);
        \draw[line width=.6pt,black,gray] (0,0.45)--(-0.05,0.45);
                \draw[line width=.6pt,black,gray] (-0.15,0.45)--(-0.35,0.45);
             \draw[line width=.6pt,black] (0.2,-1.6)--(0.2,-2.1);
            \node[ line width=0.6pt, dashed, draw opacity=0.5] (a) at (0.8,0.6){$\Phi$};
            \node[ line width=0.6pt, dashed, draw opacity=0.5] (a) at (0.2,0.4){$\Gamma$};
          \node at (0.5, -0.2) {$ C$};
                   \node at (0.2, -1.3) {$ C$}; 
       %  \node[ line width=0.6pt, dashed, draw opacity=0.5] (a) at (0.4,-1.8){$k$};
        \end{tikzpicture}
    \end{aligned}
    =\begin{aligned}
        \begin{tikzpicture}
          % \draw[fill=lightgray]  (0,0) -- (1,0) -- (0.5,-0.5) -- cycle;
         %  \draw[line width=.6pt,black] (0.2,0)--(0.2,0.4);
        %   \draw[line width=.6pt,black] (0.8,0)--(0.8,0.6);
           \draw[line width=.6pt,black] (0.5,-0.5)--(0.5,-1.1); 
           \draw[line width=.6pt,black,gray] (0.2,-0.25)--(0,-0.25);
            \draw[line width=.6pt,black,gray] (2,-0.25)--(2.4,-0.25);
           % \draw[line width=.6pt,black,gray] (0.5,0.6)--(0.5,0.8);
         \draw[fill=gray!30]  (-0.3,-1.1) -- (0.7,-1.1) -- (0.2,-1.6) -- cycle;
         \draw[line width=.6pt,black] (-0.1,-1.1)--(-0.1,0.5);
           \draw[black,fill=gray!30] (0.2,-.5) rectangle ++(1.8,0.6); 
        %  \draw[line width=.6pt,black,cyan] (0.4,0.45)--(0.5,0.45);
          %\draw[line width=.6pt,black,cyan] (0.5,.45)--(0.5,0.3);
        %\draw[line width=.6pt,black,cyan] (0,0.45)--(-0.05,0.45);
            %    \draw[line width=.6pt,black,cyan] (-0.15,0.45)--(-0.35,0.45);
             \draw[line width=.6pt,black] (0.2,-1.6)--(0.2,-2.1);
        % \node[ line width=0.6pt, dashed, draw opacity=0.5] (a) at (0.4,-1.8){$k$};
             \node at (0.23, -1.32) {$ C$}; 
              \node[ line width=0.6pt, dashed, draw opacity=0.5] (a) at (1,-0.2){$\oplus_{\Psi}N_{\Gamma\Phi}^{\Psi}\Psi$};
        \end{tikzpicture}
    \end{aligned},
\end{equation}
where we used $\sum_{(g)}\Gamma(g^{\cone})\otimes \Phi(g^{\ctwo})=\sum_{\Psi\in \Irr(\Rep(H))} N_{\Gamma \Phi}^{\Psi} \Psi(g)$.
This implies that 
\begin{equation}
    Z_{\Gamma}Z_{\Phi}=Z_{\Gamma \otimes \Phi}=\sum_{\Psi\in \Irr(\Rep(H))} N_{\Gamma \Phi}^{\Psi} Z_{\Psi},
\end{equation}
which in turn yields Eq.~\eqref{eq:Wgamma}. For closed string, Eq.~\eqref{eq:Wgamma} can also be derived using dual Hopf algebra $\hat{H}$ structure.

\section{Symmetry TFT for weak Hopf lattice gauge theory}
\label{sec:WHAsymTFT}

The Symmetry TFT, also known as topological holography, provides a general framework for understanding non-invertible or categorical symmetries in both gapped and gapless phases \cite{Kong2020algebraic, SchaferNameki2024ICTP, huang2023topologicalholo, bhardwaj2024lattice, freed2024topSymTFT, gaiotto2021orbifold, bhardwaj2023generalizedcharge, apruzzi2023symmetry, bhardwaj2024gappedphases, Zhang2024anomaly, Ji2020categoricalsym}. 
For an \( n \)-dimensional quantum field theory \(\mathsf{QFT}_n\) with fusion category symmetry \(\EuScript{S}\), the corresponding Symmetry TFT \(\mathcal{Z}(\EuScript{S})\) has the advantage of separating the topological data of the fusion category symmetry from the non-topological dynamics of the theory.
The symmetry \(\EuScript{S}\) is encoded in the symmetry boundary \(\EuScript{B}_{\text{sym}}^{\EuScript{S}}\), while the non-topological data resides on the physical boundary \(\EuScript{B}_{\text{phys}}\). Upon compactification of the \((n+1)\)-dimensional Symmetry TFT, we obtain an \(n\)-dimensional \(\mathsf{QFT}_n\) with non-invertible symmetry \(\EuScript{S}\).
Constructing a lattice model for SymTFT is a fundamental and intriguing problem. In the case of $(1+1)$D phases, the quantum double model provides a natural framework for addressing this task.

\subsection{Weak Hopf lattice gauge theory}

The weak Hopf lattice gauge theory, also referred to as the weak Hopf quantum double model, is comprehensively examined in Refs.~\cite{chang2014kitaev, Jia2023weak} as a lattice-based realization of the weak Hopf generalization of BF-theory or Dijkgraaf-Witten theory \cite{dijkgraaf1990topological}. For a weak Hopf algebra \( H \), the representation category \( \Rep(H) \) is typically a multifusion category, establishing an equivalence between the weak Hopf lattice gauge theory and the multifusion string-net model \cite{jia2024weakTube}, a lattice representation of the Turaev-Viro-Barrett-Westbury topological quantum field theory (TQFT).
Although the input data \( \Rep(H) \) forms a multifusion category, the associated topological excitations are described by a modular tensor category, \( \mathcal{Z}(\Rep(H)) \simeq \Rep(D(H)) \) as braided fusion category. Here, \( D(H) \) represents the quantum double of \( H \), while \( \mathcal{Z}(\Rep(H)) \) denotes the Drinfeld center of \( \Rep(H) \) \cite{jia2024weakTube}.

Since quantum double plays a crucial role in our construction, let us briefly review its definition here.
Consider a weak Hopf algebra \(H\). The tensor product \(\hat{H}^{\text{cop}} \otimes H\) is equipped with an algebra structure where multiplication is given by \cite{drinfel1988quantum,majid1990physics,majid1994some}
\begin{equation}
    (\varphi \otimes h) (\psi \otimes g) = \sum_{(h)} \sum_{(\psi)} \varphi \psi^{(2)} \otimes h^{(2)} g \langle \psi^{(1)}, S^{-1}(h^{(3)}) \rangle \langle \psi^{(3)}, h^{(1)} \rangle,
\end{equation}
where we use the pairing between $\hat{H}$ and $H$ and the comultiplication of \(\psi\) is considered within \(\hat{H}\).
It is useful to rewrite the multiplication using the so-called straightening relation:
\begin{equation}
   h \cdot \psi  = \psi^{\ctwo} \otimes  g^{\ctwo}  \langle \psi^{(1)}, S^{-1}(h^{(3)}) \rangle \langle \psi^{(3)}, h^{(1)} \rangle.
\end{equation}
The unit of this algebra is given by \(\varepsilon \otimes 1_H\). 
The linear span \(M\) of the elements
\begin{align}
    \varphi \otimes xh - \varphi (x\rightharpoonup \varepsilon) \otimes h, \quad x \in H_L, \\
    \varphi \otimes yh - \varphi (\varepsilon \leftharpoonup y) \otimes h, \quad y \in H_R,
\end{align}
constitutes a two-sided ideal of \(\hat{H}^{\text{cop}} \otimes H\). We denote the quotient algebra \((\hat{H}^{\text{cop}} \otimes H)/M\) as \(D(H)\), with equivalence classes in \(D(H)\) represented by \([\varphi \otimes h]\) for \(\varphi \otimes h \in \hat{H}^{\text{cop}} \otimes H\).

\begin{definition}[Quantum double] \label{def:quantum-double}
The quantum double of a weak Hopf algebra \(H\), denoted by \(D(H)\), is a weak Hopf algebra equipped with the following structure:

\begin{enumerate}
    \item[(1)] The multiplication operation is defined by
    \[
    [\varphi \otimes h] [\psi \otimes g] = \sum_{(\psi), (h)} [\varphi \psi^{(2)} \otimes h^{(2)} g] \langle \psi^{(1)}, S^{-1}(h^{(3)}) \rangle \langle \psi^{(3)}, h^{(1)} \rangle.
    \]
    
    \item[(2)] The unit element is
    \[
    [\varepsilon \otimes 1_H].
    \]
    
    \item[(3)] The comultiplication is given by
    \[
    \Delta([\varphi \otimes h]) = \sum_{(\varphi), (h)} [\varphi^{(2)} \otimes h^{(1)}] \otimes [\varphi^{(1)} \otimes h^{(2)}].
    \]
    
    \item[(4)] The counit is defined as
    \[
    \varepsilon ([\varphi \otimes h]) = \langle \varphi, \varepsilon_R(S^{-1}(h)) \rangle.
    \]
    
    \item[(5)] The antipode is given by
    \[
    S([\varphi \otimes h]) = \sum_{(\varphi), (h)} [\hat{S}^{-1}(\varphi^{(2)}) \otimes S(h^{(2)})] \langle \varphi^{(1)}, h^{(3)} \rangle \langle \varphi^{(3)}, S^{-1}(h^{(1)}) \rangle.
    \]
\end{enumerate}
\end{definition}

Notice that the category ${}_{H}\mathsf{YD}^H$ of left-right Yetter--Drinfeld modules is equivalent to the category $\Rep(D(H))$. Moreover, taking the Drinfeld center of $\Rep(H)$, any object therein naturally carries the structure of a Yetter--Drinfeld module. Hence we have 
\[
\mathcal{Z}(\Rep(H))  \simeq {}_{H}\mathsf{YD}^H \simeq \Rep(D(H)),
\]
they all characterize the (2+1)D topological excitations of the weak Hopf lattice gauge theory.

\paragraph{Weak Hopf quantum double lattice model}
Consider a 2d manifold $\Sigma$, a lattice on it is a  cellulation $C(\Sigma)=C^0(\Sigma)\cup C^1(\Sigma)\cup C^2(\Sigma)$ with $C^0(\Sigma)$ denoting the set of vertices, $C^1(\Sigma)$ denoting the set of edges, and $C^2(\Sigma)$ denoting the set of faces.
The bulk input data is a weak Hopf algebra $H$. For a directed lattice on $2d$ closed surface, we assign each edge a weak Hopf qudit $\mathcal{H}_e=H$, the total Hilbert space is $\mathcal{H}_{\rm tot}=\otimes_{e\in C^1(\Sigma)} \mathcal{H}_e$. All edges of the lattice are oriented, to construct the local stabilizers, we adopt the following convention:
\begin{equation}
    \begin{aligned}
        \begin{tikzpicture}   
           \draw[line width=1pt,-stealth,black] (0,0)--(1,0);
         \draw[line width=1pt,black] (0.8,0)--(1.8,0);
        \draw[black,fill=black] (1.8,0) circle (0.08); 
           \draw[black,fill=black] (0,0) circle (0.08);
           \node[ line width=0.6pt, dashed, draw opacity=0.5] (a) at (-0.4,0){$\XL$};
           \node[line width=0.6pt, dashed, draw opacity=0.5] (a) at (2.2,0){$\XR$};
           \node[line width=0.6pt, dashed, draw opacity=0.5] (a) at (0.9,-0.4){$\ZR$};
           \node[line width=0.6pt, dashed, draw opacity=0.5] (a) at (0.9,0.4){$\ZL$};
        \end{tikzpicture}
    \end{aligned}.
\end{equation}
With this convention, consider the following configuration:
\begin{equation}
\begin{aligned}
\begin{tikzpicture}
    \draw[-stealth,line width=1pt,black] (-1,0) -- (0,0); 
    \draw[-stealth,line width=1pt,black] (0,0) -- (1,0); 
    \draw[-stealth,line width=1pt,black] (0,0) -- (0,1); 
    \draw[-stealth,line width=1pt,black] (0,-1) -- (0,0); 
    \draw[-stealth,line width=1pt,black] (0,1) -- (1,1);
    \draw[-stealth,line width=1pt,black] (1,0) -- (1,1);  
    \draw[line width=1pt, dotted, red] (0,0) -- (0.5,0.5);
    \draw [fill = black] (0,0) circle (1.2pt);
    \draw [fill = black] (0.5,0.5) circle (1.2pt);
    \draw[red,-stealth] (0.2,0.2) arc (45:330:0.3);
    \draw[red,-stealth] (0.3,0.3) arc (-135:160:0.3);
    \node[ line width=0.2pt, dashed, draw opacity=0.5] (a) at (-0.2,-0.2){$v$};
    \node[ line width=0.2pt, dashed, draw opacity=0.5] (a) at (0.7,0.7){$f$};
    \node[ line width=0.2pt, dashed, draw opacity=0.5] (a) at (-1.2,0){$j_5$};
    \node[ line width=0.2pt, dashed, draw opacity=0.5] (a) at (0,-1.2){$j_6$};
    \node[ line width=0.2pt, dashed, draw opacity=0.5] (a) at (0.8,-0.3){$j_1$};
    \node[ line width=0.2pt, dashed, draw opacity=0.5] (a) at (1.3,0.5){$j_2$};
    \node[ line width=0.2pt, dashed, draw opacity=0.5] (a) at (0.5,1.2){$j_3$};
    \node[ line width=0.2pt, dashed, draw opacity=0.5] (a) at (-0.2,0.5){$j_4$};
\end{tikzpicture}
\end{aligned}.
\end{equation}
Here, $s=(v,f)$ is an edge-face link, and we assume a counterclockwise order around each face and vertex, with the starting point at link $s$.

\begin{figure}[t]
    \centering
    \includegraphics[width=10cm]{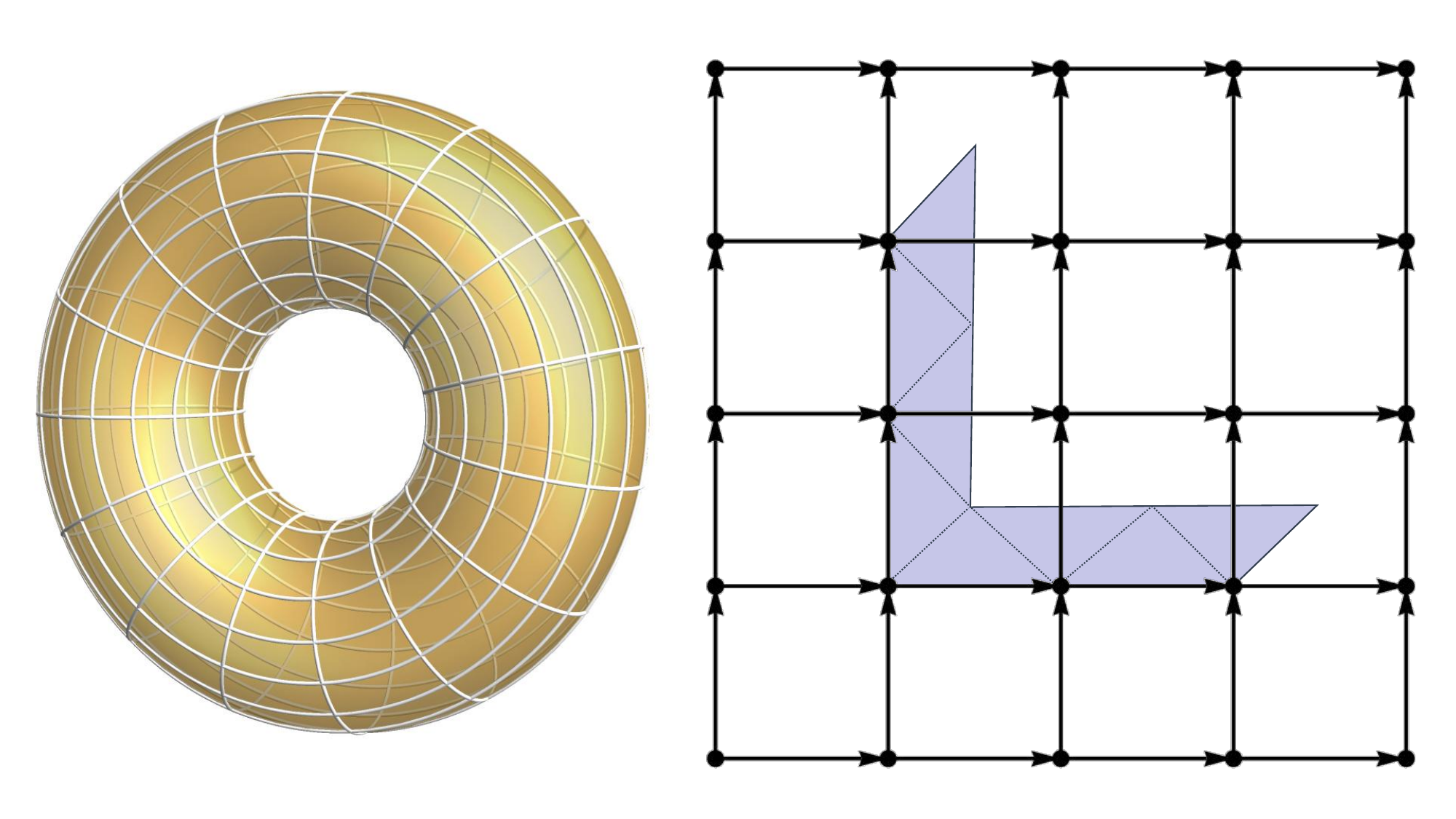}
\caption{An illustration of a quantum double lattice, viewed as a cellulation of a surface. 
The purple region indicates a ribbon on the lattice.}
    \label{fig:cluster-graph}
\end{figure}

The corresponding vertex and face operators are defined based on comultiplication as follows:
\begin{equation}
\begin{aligned}
&\Av^h(s) = \sum_{(h)} \XL_{h^{(1)}}(j_4) \otimes \XR_{h^{(2)}}(j_5) \otimes \XR_{h^{(3)}}(j_6) \otimes \XL_{h^{(4)}}(j_1), \\
&\Bf^{\varphi}(s) = \sum_{(\varphi)} \ZL_{\varphi^{(1)}}(j_1) \otimes \ZL_{\varphi^{(2)}}(j_2) \otimes \ZR_{\varphi^{(3)}}(j_3) \otimes \ZR_{\varphi^{(4)}}(j_4), \label{eq:ABdef}
\end{aligned}
\end{equation}
where $h \in H$ and $\varphi \in \hat{H}$.
At the link $s$, we have a representation of quantum double $D(H)$:
\begin{equation}
    [\varphi \otimes h]\mapsto \Bf^{\varphi}(s)\Av^h(s),
\end{equation}
where $[\varphi\otimes h]$ is an element in quantum double, see \cite{Jia2023weak} for the proof.

For any site \(s\), the operators \(\Av^h(s)\) and \(\Bf^\varphi(s)\) obey the following commutation relations \cite{Jia2023weak}:
\begin{align}
    &\Av^h(s)\Bf^\varphi(s) = \sum_{(h),(\varphi)} \Bf^{\varphi^{(2)}}(s) \Av^{h^{(2)}}(s) \langle \varphi^{(1)}, S^{-1}(h^{(3)}) \rangle \langle \varphi^{(3)}, h^{(1)} \rangle, \label{eq:rep-comm-rel-1} \\
    &\Bf^\varphi(s) \Av^{xh}(s) = \Bf^{\varphi(x\rightharpoonup \varepsilon)}(s) \Av^h(s), 
    \Bf^\varphi(s) \Av^{yh}(s) = \Bf^{\varphi(\varepsilon \leftharpoonup y)}(s) \Av^h(s), \label{eq:rep-comm-rel-2}
\end{align}
for any \(x \in H_L=\varepsilon_L(H)\) and \(y \in H_R=\varepsilon_R(H)\).

It can be proven that [Ref. \cite[Lemma 2]{chang2014kitaev}
\begin{equation}
    [\varphi\otimes 1][\varepsilon \otimes h] =[\varphi \otimes h] =[\varepsilon \otimes h]  [\varphi \otimes 1]
\end{equation}
From this and Eqs. \eqref{eq:rep-comm-rel-1} and \eqref{eq:rep-comm-rel-2}, we obtain that if $h$ is taken as the Haar integral $\lambda$ and $\varphi$ is taken as the Haar measure $\Lambda$, then the corresponding local operators are projectors that commute with each other. Additionally, the definition of the operator does not depend on the initial link $s$ but only on the vertex and faces. Therefore, we define the operators as $\Av_v = \Av^{\lambda}(s)$ and $\Bf_f = \Bf^{\Lambda}(s)$. The Hamiltonian is defined as
\begin{equation}
    \mathbb{H}=-\sum_v \Av_v -\sum_f \Bf_f.
\end{equation}
The topological excitations of the model are characterized by the representation category of quantum double $D(H)$.

\paragraph{Topological boundary condition of weak Hopf lattice gauge theory}

The topological boundary condition of weak Hopf lattice gauge theory can be characterized in several different ways.
From anyon condensation perspective, we have the following (see Refs.~\cite{Kong2014,Cong2017,jia2022electricmagnetic,burnell2018anyon,eliens2010anyon,jia2023boundary,Jia2023weak}):
\begin{itemize}
    \item Lagrangian algebra $\cA$. The topological boundary condition is characterized by a Lagrangian algebra in the bulk phase $\Rep(D(H))$, the boundary phase is given by anyon condensation $\eB=\Rep(D(H))_{\cA}$, the category of $\cA$-modules in $\Rep(D(H))$.
    \item Frobenius algebra $\cF$. For this case, the topological boundary condition is characterized by a  Frobenius algebra  in the bulk phase $\Rep(D(H))$. The anyon condensation is realized in two steps, we first take the quotient category $\Rep(D(H))/\cF$, then we take idempotent complete  to obtain the boundary phase $\eB=\operatorname{IC}(\Rep(D(H))/\cF)$.
\end{itemize}

Since $\Rep(D(H))$ is non-chiral topological phase, the boundary can also be characterized at the level a weak Hopf gauge symmetry $W$ (or equivalently at the level of (multi)fusion category $\Rep(H)$).

\begin{itemize}
    \item In the framework of quantum double lattice model, the bulk gauge symmetry is the weak Hopf algebra $W$, the boundary gauge symmetry is a comodule algebra $K$ over $H$.
    \item In the framework of string-net model, the bulk input (multi)fusion category is $\Rep(H)$, the boundary is determined by a $\Rep(H)$-module category $\eM$. For a given comodule algebra $H$, the category of $K$-modules $\Mod_K$ is of module  category over $\Rep(H)$.
\end{itemize}

\subsection{Algebraic theory of SymTFT for weak Hopf lattice gauge theory}

\begin{figure}[t]
    \centering
    \includegraphics[width=8cm]{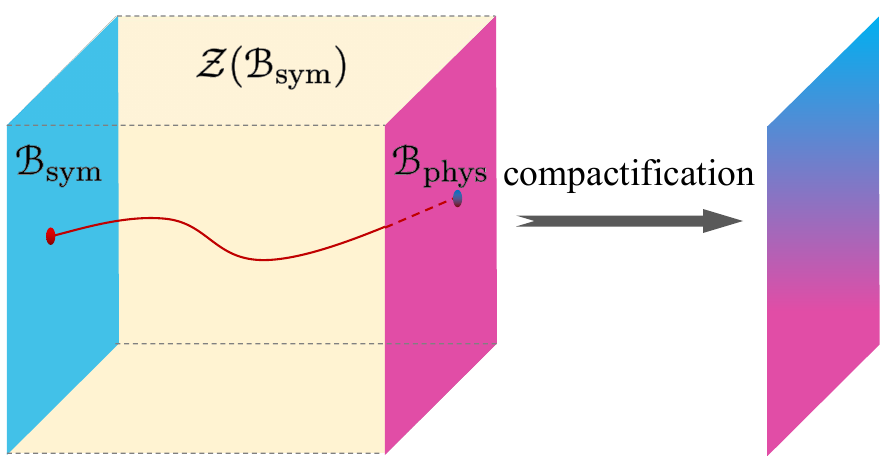}
    \caption{Depiction of the SymTFT sandwich: the symmetry boundary is drawn in cyan, while the physical boundary is drawn in magenta.}
    \label{fig:SymTFT}
\end{figure}

The SymTFT is a sandwich structure as shown in Figure~\ref{fig:SymTFT}, the sandwich manifold defined as $\mathbb{M}^{1,1}\times [0,1]$ where $\mathbb{M}^{1,1}$ represents the $(1+1)$ D manifold on which our system lies on. 
The sandwich manifold $\mathbb{M}^{1,1}\times [0,1]$ is a (2+1)D manifold that has two boundaries.
The basic idea is we put the non-invertible symmetry on one of the boundary, and put our physical system on the other boundary, then after doing compactification over the interval $[0,1]$, we obtain a (1+1)D system with non-invertible symmetry.

Consider a \((1+1)\)D phase \(\mathcal{T}\) with fusion category symmetry given by the fusion category \(\eB_{\rm sym}\), whose physical information is captured by \(\eB_{\rm phys}\). The SymTFT for this phase is a triple \((\eB_{\rm sym}, \eB_{\rm phys}, \mathcal{Z}(\eB_{\rm sym}))\), where \(\mathcal{Z}(\eB_{\rm sym})\) is the Drinfeld center of \(\eB_{\rm sym}\), a unitary modular tensor category characterizing a \((2+1)\)D topological phase. The symmetry boundary \(\eB_{\rm sym}\) must be a gapped topological boundary condition of \(\mathcal{Z}(\eB_{\rm sym})\). The physical boundary \(\eB_{\rm phys}\) is also a boundary condition for the bulk TQFT, but it is not necessarily topological or even gapped.
If \(\mathcal{T}\) is gapless, \(\eB_{\rm phys}\) corresponds to a gapless boundary condition. Conversely, if \(\mathcal{T}\) is gapped, \(\eB_{\rm phys}\) represents a gapped topological boundary condition.

% \begin{figure}[t]
%     \centering
%     \includegraphics[width=8cm]{FigRibbon.pdf}
%     \caption{Depiction of local operator and half-braiding for SymTFT sandwich of weak Hopf lattice gauge theory.}
%     \label{fig:ladder}
% \end{figure}

For this section, we will assume that \(\mathcal{T}\) is gapped. Therefore, both \(\eB_{\rm sym}\) and \(\eB_{\rm phys}\) are chosen as gapped boundary conditions of the TQFT \(\mathcal{Z}(\eB_{\rm sym})\). Now, we set the \((2+1)\)D bulk theory as the weak Hopf lattice gauge theory \(\mathcal{Z}(\eB_{\rm sym}) \simeq_{\otimes, \rm br} \Rep(D(H)) \simeq_{\otimes, \rm br} \mathcal{Z}(\Rep(H))\), where the weak Hopf gauge symmetry is given by \(H\).
Note that this does not imply that \(\eB_{\rm sym}\) is monoidally equivalent to \(\Rep(H)\); it only means that they are weak Morita equivalent (i.e., their Drinfeld centers are equivalent as unitary modular tensor categories). All topological boundary conditions of a weak Hopf lattice gauge theory are weak Morita equivalent, allowing us to choose \(\eB_{\rm sym}\) as an arbitrary topological gapped boundary condition.
For weak Hopf gauge symmetry $H$, the corresponding SymTFT is a sandwich that consists of the following data:
\begin{itemize}
    \item For the (2+1)D bulk, we put a weak Hopf lattice gauge theory whose topological excitation is given by the representation category $\Rep(D(H))$ of the quantum double $D(H)$.
    \item The symmetry boundary is a topological boundary condition of bulk lattice gauge phase $\Rep(D(H))$, this can be characterized by a comodule algebra $K$ over weak Hopf algebra $H$. We denote the boundary vacuum charge as $\cA_K$, which is a Lagrangian algebra in $\Rep(D(H))$.
    \item The physical boundary condition $\eB_{\rm phys}$ is also chosen as a topological boundary condition, thus it is also characterized by a comodule algebra $J$ over $H$. We denote the boundary vacuum charge as $\cA_J$, which is also a Lagrangian algebra in $\Rep(D(H))$.
    \item The local operators of are generated by Wilson lines that can end at both boundaries. In weak Hopf lattice gauge theory, the Wilson lines are ribbon operators that connect two boundaries~\cite{meusburger2021hopf,meusburger2017kitaev,chen2021ribbon,jia2023boundary,Jia2023weak}, representing the transport of vacuum charge from one boundary to the other.
    \item The symmetry action on local operators is realized by half-braiding. In weak Hopf lattice gauge theory, this is characterized by boundary closed ribbons \cite{meusburger2021hopf,meusburger2017kitaev,chen2021ribbon,jia2023boundary,Jia2023weak}, which are ribbons with both ends anchored on the symmetry boundary. Since the symmetry boundary closed ribbon intersects with the local operator ribbon, this intersection can be used to realize the half-braiding.
\end{itemize}

%%%%%%%%%%%%%%%%%%%%%%%%%%%%%%%%%%%%
\section{Weak Hopf cluster ladder model I: smooth-rough construction}
\label{sec:latticeI}

In this section, we focus on a special case of the cluster ladder model, 
where one boundary is chosen to be smooth and the other to be rough. 
The general discussion of topological boundary conditions will be deferred to the next section. 
The main results of this section, together with those of the next, can be summarized as follows:

\begin{theorem}[Informal summary]
Any SymTFT with (multi)fusion category symmetry \(\EuScript{S}\) can be realized via a cluster state model based on the smooth-rough construction. We first use Tannaka-Krein reconstruction (or weak Hopf tube algebra approach \cite{Kitaev2012boundary,jia2024weakTube,jia2025weakhopftubealgebra}) to obtain a weak Hopf gauge symmetry \(H_{\EuScript{S}}\) such that \(\EuScript{S} \simeq \Rep(H_{\EuScript{S}})\) as (multi)fusion categories. Then, we construct the cluster state model based on this weak Hopf symmetry \(\Hbb_{\rm{cluster}}[H_{\EuScript{S}}]\), which realizes the \(\EuScript{S}\)-SymTFT.
\end{theorem}

\begin{figure}
    \centering
    \includegraphics[width=0.65\linewidth]{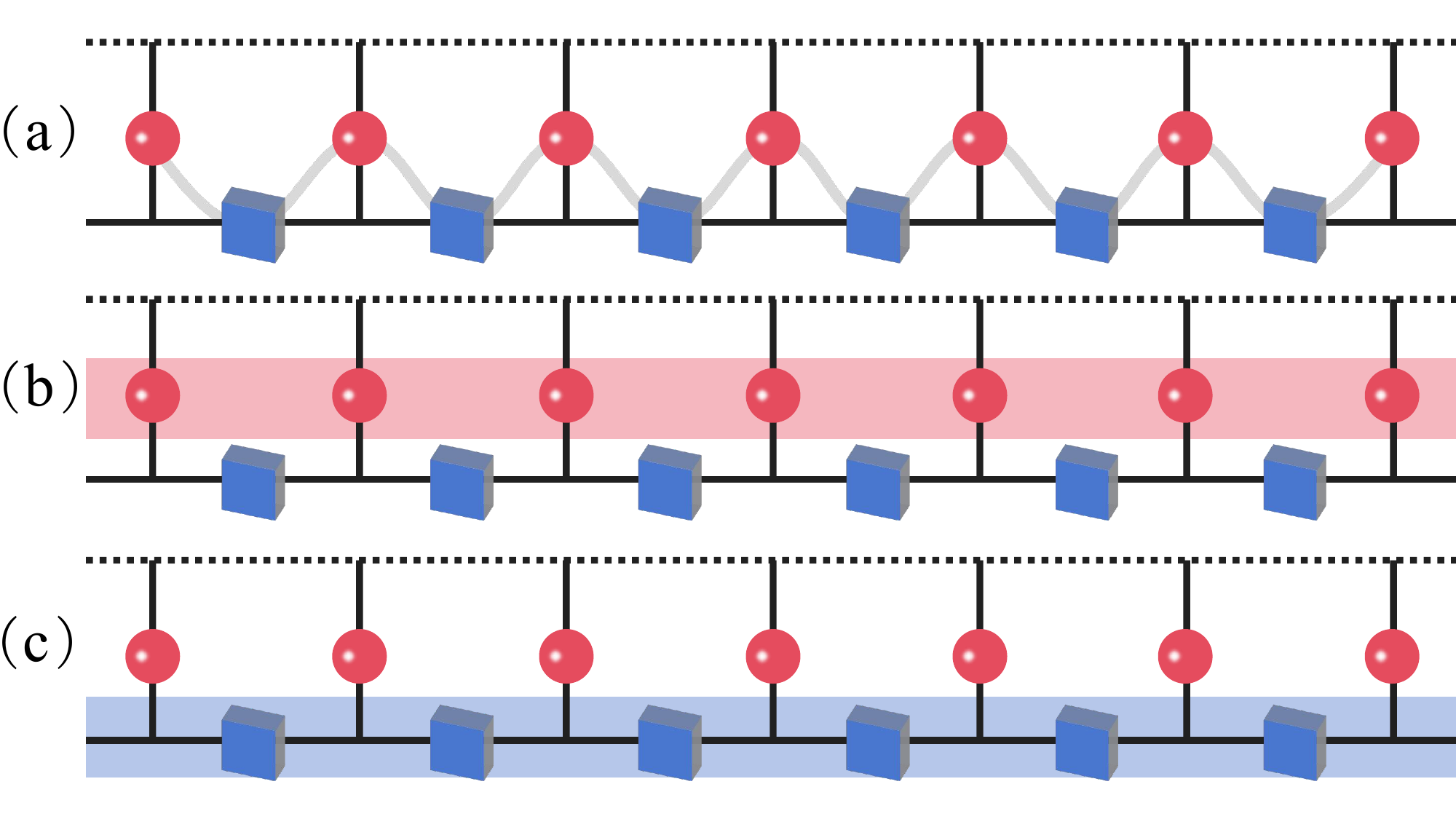}
\caption{Correspondence between the cluster state model and the quantum double model with one smooth boundary and one rough boundary.  
(a) The black lattice denotes the quantum double lattice, where physical degrees of freedom reside on solid edges. In the corresponding cluster state model, qubits are placed on vertices, with odd vertices shown in red, even vertices in blue, and links between them in gray.  
(b) The magnetic $X$-string operator on the rough boundary corresponds to the $\Cocom(H)$ symmetry operator of the cluster state model.  
(c) The electric $Z$-string operator on the smooth boundary corresponds to the $\Cocom(\hat{H})$ symmetry operator of the cluster state model.}
\label{fig:Z2ladder} 
\end{figure}

\subsection{Weak Hopf cluster state model}

As shown in Ref.~\cite{jia2024generalized} for the Hopf algebra case (and in the finite group case~\cite{fechisin2023noninvertible}), 
the CSS-type cluster state model can be regarded as a quantum double model with one smooth boundary and one rough boundary. 
Inspired by this construction, we define the weak Hopf cluster state model as a weak Hopf quantum lattice gauge theory 
with one smooth boundary and one rough boundary; see Figure~\ref{fig:Z2ladder}. 
In the study of cluster states, the qudits are commonly placed on the vertices, which are divided into two sets, referred to as \emph{odd} and \emph{even} vertices. 
In the quantum double model setting, the qudits are instead placed on the edges, 
where an odd vertex corresponds to a bulk edge and an even vertex corresponds to a boundary edge.

\begin{definition}[Weak Hopf cluster state model] For weak Hopf algebra $H$,  the weak Hopf cluster state model is a weak Hopf lattice gauge theory with one rough boundary and one smooth boundary.
\end{definition}

The smooth boundary here is analogous to the Neumann boundary in quantum field theory. 
In the lattice model, this corresponds to the boundary degrees of freedom being identical to those in the bulk. 
In the weak Hopf setting, the boundary input weak Hopf algebra is the same as the bulk weak Hopf algebra; 
in the fusion category setting, the boundary module category is taken to be the same as that of the bulk fusion category.
The subtlety arises for rough boundaries. The notion of a rough boundary refers to the removal of degrees of freedom at the boundary, 
which has been commonly employed in lattice models since the original work of Ref.~\cite{bravyi1998quantum}. 
This usage differs from another common convention in the literature in the general weak Hopf or fusion category setting, where a rough boundary corresponds to the module category $\Vect$ 
over the input fusion category in the string-net model framework, or equivalently, to the comodule algebra $\mathbb{C}$ 
over the input weak Hopf algebra in the quantum double model framework.
For some fusion categories $\eC$ that are anomalous \cite{thorngren2019fusion}, $\Vect_{\Cbb}$ is not a module category over $\eC$, 
and $\mathbb{C}$ is not a comodule algebra over $H$. 
However, for anomaly-free symmetries, the two notions coincide: 
the rough boundary is equivalent to the Dirichlet boundary in field theory.

We will set symmetry boundary as smooth boundary.
For bulk weak Hopf gauge symmetry $H$, the smooth boundary corresponds to comodule algebra $K_s=H$, namely, choose boundary gauge symmetry as $H$ itself. The corresponding boundary excitation is $\eB_s\cong \Rep(H)$, which is the non-invertible symmetry in the framework of SymTFT.

Starting from a weak Hopf quantum double model with one smooth boundary and one rough boundary, and applying a coarse-graining procedure to the bulk, we ultimately arrive at the following ladder lattice structure (with horizontal periodic boundary conditions):
\begin{equation}\label{eq:ClusterLattice}
\begin{aligned}
\begin{tikzpicture}
    % Define the number of rungs in the ladder
    \def\n{5}
    % Define the size of each square
    \def\s{1}
        % Draw the shaded background lattice
    \fill[green!20] (0, 0) rectangle (\n*\s+\s, \s); % Rectangle covering the whole background
    % Draw the ladder with arrows in the middle of each edge
    \foreach \i in {0,...,\n} {
        % Draw solid bottom edges with arrows in the middle pointing right
        \draw[-stealth, line width=1.0pt,blue, midway] (\i*\s, 0) -- (\i*\s+\s, 0);
        % Draw dotted top edges with arrows in the middle pointing right
        \draw[dotted, line width=1.0pt,red, midway] (\i*\s, \s) -- (\i*\s+\s, \s);
        % Draw upward ladder edges with arrows in the middle pointing up
        \draw[-stealth,line width=1.0pt, midway] (\i*\s, 0) -- (\i*\s, \s);
    }
    % Draw the right-most vertical ladder edge with an arrow in the middle pointing up
    \draw[-stealth, midway,line width=1.0pt] (\n*\s+\s, 0) -- (\n*\s+\s, \s);
\end{tikzpicture}
\end{aligned}
\end{equation}
where bulk edges are intentionally depicted in black, while smooth boundary edges are colored blue for clarity.
We assign to each edge a Hilbert space $\mathcal{H}_e = H$, so that the total Hilbert space is given by $\mathcal{H}_{\rm tot} = \bigotimes_{e} \mathcal{H}_e$.

There are two types of stabilizers, one is the smooth boundary vertex operator \(\Av_{v_s}\):
\begin{align}
 A_v^h=\sum_{(h)}\XR_{h^{\cone}} \otimes \XL_{h^{\ctwo}} \otimes \XL_{h^{\cthree}}   
 = \sum_{(h)}\begin{aligned}
           \begin{tikzpicture}
    \begin{scope}[rotate=90] % Rotate the entire scope by 90 degrees
        \fill[green!20] (0, 0) rectangle (1,2); % Rectangle 
        \draw[-stealth,blue,line width = 1.6pt] (0,1) -- (0,0);
      %   \draw[cyan,dotted, line width = 1pt] (0,1) -- (0.5,1.5); 
        \draw[-stealth,blue,line width = 1.6pt] (0,2) -- (0,1); 
        \draw[-stealth,black] (0,1) -- (1,1); 
       	\node[ line width=0.2pt, dashed, draw opacity=0.5] (a) at (-0.4,0.5){$\XL_{h^{\ctwo}} $};
    		\node[ line width=0.2pt, dashed, draw opacity=0.5] (a) at (-0.4,1.5){$\XR_{h^{\cone}}$};
            \node[ line width=0.2pt, dashed, draw opacity=0.5] (a) at (0.6,0.5){$\XL_{h^{\cthree}}$};
             %\node[ line width=0.2pt, dashed, draw opacity=0.5] (a) at (0.7,1.4){$s$};
    \end{scope}
\end{tikzpicture}
 \end{aligned},
\end{align}
where we have assumed a starting site $s=(v_s,f)$ and then apply the comultiplication the obtain the local operator. Since generalized X-operators form a regular representation of weak Hopf algebra $H$ and $\Delta(gh)=\Delta(g)\Delta(h)$, we have $A_{v_s}^gA_{v_s}^h=A^{gh}_{v_s}$ at any vertex $v_s$ (notice that these local operators are assumed to have the same starting site, see \cite{Jia2023weak} for detail).

Another type of local operators are the face operators $\Bf_f^{\psi}=\sum_{(\psi)}\ZR_{\psi^{\cone}}\otimes \ZL_{\psi^{\ctwo}} \otimes \ZL_{\psi^{\cthree}}$:
\begin{equation}
\begin{aligned}
   &   \Bf_f^{\psi} \big{|}\begin{aligned}
    \begin{tikzpicture}
    \begin{scope}[rotate=-90] % Rotate the entire scope by -90 degrees
        \fill[green!20] (0, 0) rectangle (1,1); % Rectangle 
       % \draw[-stealth,red,line width = 1.6pt] (0,0) -- (0,1);
     %  \draw[cyan,dotted, line width = 1pt] (0,0) -- (0.5,-0.5); 
        %\draw[-stealth,red, line width = 1.6pt] (0,1) -- (0,2); 
        \draw[-stealth,blue,line width = 1.6pt] (1,0) -- (1,1);
        \draw[-stealth,black] (1,1) -- (0,1); 
                \draw[-stealth,black] (1,0) -- (0,0); 
        %\node[ line width=0.2pt, dashed, draw opacity=0.5] (a) at (-0.3,0.5){$$};
        \node[ line width=0.2pt, dashed, draw opacity=0.5] (a) at (1.4,0.5){$s$};
       % \node[ line width=0.2pt, dashed, draw opacity=0.5] (a) at (0.7,1.4){$s_b$};
        \node[ line width=0.2pt, dashed, draw opacity=0.5] (a) at (0.5,1.25){$t$};
            \node[ line width=0.2pt, dashed, draw opacity=0.5] (a) at (0.5,-.25){$r$};
    \end{scope}
\end{tikzpicture}   
\end{aligned}\big{\rangle}
= \sum_{(r),(s),(t)} \psi(r^{\ctwo} S(s^{\cone}) S(t^{(1)})) 
\big{|}\begin{aligned}
    \begin{tikzpicture}
    \begin{scope}[rotate=-90] % Rotate the entire scope by -90 degrees
        \fill[green!20] (0, 0) rectangle (1,1); % Rectangle 
       % \draw[-stealth,red,line width = 1.6pt] (0,0) -- (0,1);
     %  \draw[cyan,dotted, line width = 1pt] (0,0) -- (0.5,-0.5); 
        %\draw[-stealth,red, line width = 1.6pt] (0,1) -- (0,2); 
        \draw[-stealth,blue,line width = 1.6pt] (1,0) -- (1,1);
        \draw[-stealth,black] (1,1) -- (0,1); 
                \draw[-stealth,black] (1,0) -- (0,0); 
        %\node[ line width=0.2pt, dashed, draw opacity=0.5] (a) at (-0.3,0.5){$$};
        \node[ line width=0.2pt, dashed, draw opacity=0.5] (a) at (1.4,0.5){$s^{\ctwo}$};
       % \node[ line width=0.2pt, dashed, draw opacity=0.5] (a) at (0.7,1.4){$s_b$};
        \node[ line width=0.2pt, dashed, draw opacity=0.5] (a) at (0.5,1.3){$t^{\ctwo}$};
            \node[ line width=0.2pt, dashed, draw opacity=0.5] (a) at (0.5,-.3){$r^{\cone}$};
    \end{scope}
\end{tikzpicture}   
\end{aligned}\big{\rangle}.
\end{aligned}
\end{equation}
For any $h \in \Cocom(H)$ and $\psi \in \Cocom(\hat{H})$, the vertex and face operators commute, i.e., $[\Av_v^{h}, \Bf_f^{\psi}] = 0$ for all vertices $v$ and faces $f$. In addition, we have $[\Av_v^g, \Av_{v'}^h] = 0$ and $[\Bf_f^{\psi}, \Bf_{f'}^{\phi}] = 0$ for all $v \neq v'$ and $f \neq f'$, and for all $g, h \in H$ and $\psi, \phi \in \hat{H}$. This holds either because the operators act on disjoint sets of edges, or—when they act on overlapping edges—their actions commute due to acting from opposite sides, namely, $[\XR_g, \XL_h] = 0$ and $[\ZR_\psi, \ZL_\phi] = 0$.
The local vertex operator is defined as $\Av_v = \Av_v^{\lambda}$, where $\lambda$ is the Haar integral of $H$, and the face operator is given by $\Bf_f = \Bf_f^{\Lambda}$, where $\Lambda$ is the Haar integral of $\hat{H}$.
The Hamiltonian reads
\begin{equation}\label{eq:ClusterStateModel}
    \Hbb_{\rm cluster} = -\sum_{v_s: \text{smooth}} \Av_{v_s} - \sum_f \Bf_f.
\end{equation}
Since $\lambda \in \Cocom(H)$ and $\Lambda \in \Cocom(\hat{H})$, all local stabilizers commute with each other and are projectors.

Since for any (multi)fusion category symmetry $\eS$, there exists a corresponding weak Hopf algebra $H_{\eS}$ that realizes it, the weak Hopf cluster state model provides a lattice realization of all $\eS$-symmetric theories (the symmetry will be discussed later in Section~\ref{sec:symmetryCluster}). 
When there is no anyonic tunneling channels between the smooth and rough boundaries, the model realizes an $\eS$-protected SPT phase.

For a general weak Hopf symmetry $H$, the rough boundary (recall that by rough boundary we simply mean the removal of all degrees of freedom at the boundary) need not correspond to the $\Vect_{\Cbb}$-module category of $\Rep(H)$. The reason is that, unlike in the Hopf algebra case, a general weak Hopf algebra does not admit a fiber functor from $\Rep(H)$ to $\Vect_{\Cbb}$ (see \cite{thorngren2019fusion} for this characterization of anomalies in non-invertible symmetries). Equivalently, while in the Hopf case the rough boundary can be characterized by the $H$-comodule algebra $\Cbb$, in the weak Hopf case $\Cbb$ fails to carry such an $H$-comodule structure. In this sense, the weak Hopf symmetry $H$ is anomalous, and the associated cluster state model realizes only a partial SSB phase.
From this perspective, the anomaly indicator of a weak Hopf symmetry can be formulated equivalently in any of the following ways:
\begin{equation}
\begin{aligned}
    \Delta(1_H) = 1_H \otimes 1_H;
\quad \varepsilon(xy) = \varepsilon(x)\varepsilon(y);\\
\sum_{(x)}S(x_{(1)})x_{(2)} = 1_H\varepsilon(x),\, \forall x;
\quad \sum_{(x)} x_{(1)}S(x_{(2)}) = 1_H\varepsilon(x), \,\forall x.
\end{aligned}
\end{equation}
Among these, the first condition is particularly useful: it implies that the symmetry operator $W_{1_H}$ fails to act as the identity in the anomalous weak Hopf case—a point we will return to in later discussion.

As we treat $\mathbb{H}_{\rm cluster}$ as a $(1+1)$D model, there is a natural notion of locality on the lattice.  
The local operators $\Av^g_{v_s}$ and $\Bf_f^{\psi}$ generate local algebras at sites $s=(v_s,f)$,  
\begin{equation}
    \mathcal{A}_{(v_s,f)}^{\rm loc} = \langle \Av^g_{v_s}, \Bf_f^{\psi} \mid g \in H, \, \psi \in \hat{H} \rangle,
\end{equation}
which is isomorphic to the quantum double algebra via the identification $\Bf_f^{\psi}\Av_{v_s}^g \leftrightarrow [\psi\otimes g]$.  
For non-invertible weak symmetries, the locality structure must still be preserved, a point that will be addressed in Section~\ref{sec:symmetryCluster}.

\begin{remark}
In Ref.~\cite{jia2024generalized}, we adopt a convention different from the one used here, where the local ordering around each vertex and face is clockwise. Consequently, the vertex operator and face operator must be modified accordingly. 
Different conventions lead to different lattice models, but they result in the same quantum phase. Additionally, when constructing the lattice model, it is necessary to set an initial configuration of orientations for all edges. Different choices of this configuration also yield distinct lattice models; however, their quantum phases remain identical.
In the case of $\mathbb{C}[\mathbb{Z}_2]$, all these different models are the same due to the fact that the inverse of $1$ is itself, and they correspond to the CSS-type cluster state model \cite{brell2015generalized}. For a general Abelian group, varying the local ordering around vertices and faces does not alter the lattice model, but different initial configurations still result in distinct models. In contrast, for a non-Abelian group, both the local ordering and the initial configurations lead to different models \cite{fechisin2023noninvertible,jia2024generalized}.
\end{remark}

\subsection{Weak Hopf cluster state}

The cluster state model can be solved using the weak Hopf tensor network introduced in Section~\ref{sec:tensor-network}. Assuming periodic boundary conditions for the cluster model, we only need to consider two types of edges: the bottom smooth boundary edge (blue in Eq.~\ref{eq:ClusterLattice}) and the bulk edge (black in Eq.~\ref{eq:ClusterLattice}).

For boundary edge (even vertex of cluster lattice), the corresponding local tensor is the comultiplication of Haar integral $\lambda$:
\begin{equation}
     \Delta(\lambda)= \begin{aligned}			
\begin{tikzpicture}
    % Draw the central blue box with \lambda
    \node[draw, fill=blue!20, minimum width=0.6cm, minimum height=0.6cm] (center) at (0,0) {$\scriptstyle\lambda$};
    % Draw the edges with labels a and b
    \draw[line width=1.0pt] (center.north) -- ++(0,.5) node[above] {$\scriptstyle\lambda^{\ctwo}$};
    \draw[line width=1.0pt ] (center.south) -- ++(0,-.5) node[below] {$\scriptstyle \lambda^{\cone}$};
\end{tikzpicture}
\end{aligned} =  \begin{aligned}
				\begin{tikzpicture}
					\draw[black, line width=1.0pt]  (-0.5, 1) .. controls (-0.4, 0) and (0.4, 0) .. (0.5, 1);
					\draw[black, line width=1.0pt]  (0,0.23)--(0,-0.23);
	 % \filldraw[green!30, draw=black] 
  %       (-0.3, 0) -- (0.3, 0) -- (0, -0.4) -- cycle;
         \filldraw[green!30, draw=black] 
       (0,-0.2) circle (0.2);
    % Place the label at the centroid of the triangle
    \node at (0, -0.2) {\textcolor{black}{$\scriptstyle \lambda$}};
                   \node at (-0.5, 1.3) {$\scriptstyle \lambda^{\cone}$};
	                   \node at (0.5, 1.3) {$\scriptstyle \lambda^{\ctwo}$};			
    \end{tikzpicture}
\end{aligned}.
\end{equation}
Similarly, for bulk edge (odd vertex of cluster lattice), the local tensor corresponds to $\Delta_2(\lambda)$:
\begin{equation}
     \Delta_2(\lambda)= \begin{aligned}			
\begin{tikzpicture}
    \draw[line width=1.0pt] (center.north) (0,0) -- ++(0.8,0) node[right] {$\scriptstyle\lambda^{\ctwo}$};
    \draw[line width=1.0pt] (center.north) (0,0) -- ++(-0.8,0) node[left] {$\scriptstyle\lambda^{\cthree}$};
    \draw[line width=1.0pt ] (center.south) -- ++(0,-.5) node[below] {$\scriptstyle \lambda^{\cone}$};
    \node[draw, fill=gray!20, minimum width=0.6cm, minimum height=0.6cm] (center) at (0,0) {$\scriptstyle\lambda$};
\end{tikzpicture}
\end{aligned} =  \begin{aligned}
				\begin{tikzpicture}
					\draw[black, line width=1.0pt]  (-0.5, 1) .. controls (-0.4, 0) and (0.4, 0) .. (0.5, 1);
					\draw[black, line width=1.0pt]  (0,1)--(0,-0.23);
				\filldraw[green!30, draw=black] 
       (0,-0.2) circle (0.2);
    % Place the label at the centroid of the triangle
    \node at (0, -0.2) {\textcolor{black}{$\scriptstyle \lambda$}};
                   \node at (-0.7, 1.3) {$\scriptstyle \lambda^{\cone}$};
	                   \node at (0.7, 1.3) {$\scriptstyle \lambda^{\cthree}$};	
                     \node at (0, 1.3) {$\scriptstyle \lambda^{\ctwo}$};	
    \end{tikzpicture}
\end{aligned}.
\end{equation}
Since Haar integral \(\lambda\) is cocommutative, we use a simplified notation in which the comultiplication components \(\cone, \ctwo, \dots\) can be cyclically permuted. For example, \(\lambda^{\cone} \otimes \lambda^{\ctwo} \otimes \lambda^{\cthree} = \lambda^{\cthree} \otimes \lambda^{\cone} \otimes \lambda^{\ctwo}\).

For face, a local tensor corresponds to the Haar measure $\Lambda\in \hat{H}$ is needed to glue the edge tensors together:
\begin{equation}
    ( \id \otimes \id\otimes \hat{S}) \circ \hat{\Delta}(\Lambda)=
    \begin{aligned}			
\begin{tikzpicture}
\draw[line width=1.0pt] (center.north) (0,0) -- ++(0.8,0) node[right] {$\scriptstyle\Lambda^{\ctwo}$};
\draw[line width=1.0pt] (center.north) (0,0) -- ++(-0.8,0) node[left] {$\scriptstyle \hat{S}(\Lambda^{\cthree})$};
\draw[line width=1.0pt ] (center.south) -- ++(0,-.5) node[below] {$\scriptstyle \Lambda^{\cone}$};
\node[draw, fill=yellow!20, minimum size=0.6cm, shape=circle] (center) at (0,0) {$\scriptstyle\Lambda$};
 \filldraw[black] (-0.5,0) circle (2pt);  % Shaded dot on the line
\end{tikzpicture}
\end{aligned} 
    =
    \begin{aligned}
			\begin{tikzpicture}
				 \draw[black, line width=1.0pt]  (-0.5, 0) .. controls (-0.4, 1) and (0.4, 1) .. (0.5, 0);
				 \draw[black, line width=1.0pt]  (0,-0)--(0,1.15);
			% \filldraw[red!30, draw=black] 
        (-0.3, 1) -- (0.3, 1) -- (0, 1.4) -- cycle;
         \filldraw[red!30, draw=black] 
        (0,1.2) circle (0.2);
         \node at (0, 1.2) {$\scriptstyle \Lambda$};
          \node at (-0.7, -.3) {$\scriptstyle \Lambda^{\cone}$};
           \node at (0, -.3) {$\scriptstyle \Lambda^{\ctwo}$};
              \node at (0.9, -.3) {$\scriptstyle \hat{S}( \Lambda^{\cthree})$};
               \filldraw[black] (0.4,0.375) circle (2pt);  % Shaded dot on the line
				\end{tikzpicture}
			\end{aligned}.
\end{equation}
The antipode \(\hat{S}\) is introduced because the direction of the lattice is chosen as in Eq.~\eqref{eq:ClusterLattice}. If a different lattice configuration is chosen, the antipode may be placed on different legs of the tensor.

The cluster state is defined as the following tensor network state:
\begin{equation}
    \begin{aligned}			
\begin{tikzpicture}
    \draw[line width=1.0pt] (-0.5,0) -- ++(11.3,0) ;
    \draw[line width=1.0pt, red] (0,0) -- ++(0,-.8);
        \draw[line width=1.0pt, red] (2,0) -- ++(0,-.8);
    \draw[line width=1.0pt, red] (4,0) -- ++(0,-.8);
    \draw[line width=1.0pt, red] (6,0) -- ++(0,-.8);
    \draw[line width=1.0pt, red] (8,0) -- ++(0,-.8);
    \draw[line width=1.0pt, red] (10,0) -- ++(0,-.8);
    %%%%%
     \filldraw[black] (0.5,0) circle (2pt);  % Shaded dot on the line
      \filldraw[black] (2.5,0) circle (2pt);
       \filldraw[black] (4.5,0) circle (2pt);
        \filldraw[black] (6.5,0) circle (2pt);
         \filldraw[black] (8.5,0) circle (2pt);
          \filldraw[black] (10.5,0) circle (2pt);
    %%%%%%
    \draw[line width=1.0pt] (1,-1) -- ++(0,1) ;
    \draw[line width=1.0pt, red] (1,-1) -- ++(0,-.8);
    \node[draw, fill=blue!20, minimum width=0.6cm, minimum height=0.6cm] (center) at (1,-1) {$\scriptstyle\lambda$};
     %%%%
     \draw[line width=1.0pt] (3,-1) -- ++(0,1) ;
    \draw[line width=1.0pt, red] (3,-1) -- ++(0,-.8);
    \node[draw, fill=blue!20, minimum width=0.6cm, minimum height=0.6cm] (center) at (3,-1) {$\scriptstyle\lambda$};
    %%%%%%
         \draw[line width=1.0pt] (5,-1) -- ++(0,1) ;
    \draw[line width=1.0pt, red] (5,-1) -- ++(0,-.8);
    \node[draw, fill=blue!20, minimum width=0.6cm, minimum height=0.6cm] (center) at (5,-1) {$\scriptstyle\lambda$};
    %%%%%%
         \draw[line width=1.0pt] (7,-1) -- ++(0,1) ;
    \draw[line width=1.0pt, red] (7,-1) -- ++(0,-.8);
    \node[draw, fill=blue!20, minimum width=0.6cm, minimum height=0.6cm] (center) at (7,-1) {$\scriptstyle\lambda$};
    %%%%%%
         \draw[line width=1.0pt] (9,-1) -- ++(0,1) ;
    \draw[line width=1.0pt, red] (9,-1) -- ++(0,-.8);
    \node[draw, fill=blue!20, minimum width=0.6cm, minimum height=0.6cm] (center) at (9,-1) {$\scriptstyle\lambda$};
    %%%%%%
    \node[draw, fill=gray!20, minimum width=0.6cm, minimum height=0.6cm] (center) at (0,0) {$\scriptstyle\lambda$};
    \node[draw, fill=yellow!20, minimum size=0.6cm, shape=circle] (center) at (1,0) {$\scriptstyle\Lambda$};
     \node[draw, fill=gray!20, minimum width=0.6cm, minimum height=0.6cm] (center) at (2,0) {$\scriptstyle\lambda$};
      \node[draw, fill=yellow!20, minimum size=0.6cm, shape=circle] (center) at (3,0) {$\scriptstyle\Lambda$};
        \node[draw, fill=gray!20, minimum width=0.6cm, minimum height=0.6cm] (center) at (4,0) {$\scriptstyle\lambda$};
      \node[draw, fill=yellow!20, minimum size=0.6cm, shape=circle] (center) at (5,0) {$\scriptstyle\Lambda$};
        \node[draw, fill=gray!20, minimum width=0.6cm, minimum height=0.6cm] (center) at (6,0) {$\scriptstyle\lambda$};
      \node[draw, fill=yellow!20, minimum size=0.6cm, shape=circle] (center) at (7,0) {$\scriptstyle\Lambda$};
        \node[draw, fill=gray!20, minimum width=0.6cm, minimum height=0.6cm] (center) at (8,0) {$\scriptstyle\lambda$};
      \node[draw, fill=yellow!20, minimum size=0.6cm, shape=circle] (center) at (9,0) {$\scriptstyle\Lambda$};
        \node[draw, fill=gray!20, minimum width=0.6cm, minimum height=0.6cm] (center) at (10,0) {$\scriptstyle\lambda$};
\end{tikzpicture}
\end{aligned} 
\end{equation}
where we have drawn the free leg (corresponding to physical degrees of freedom) in red.
The cluster state is the ground state of the cluster state model defined in Eq.~\eqref{eq:ClusterStateModel}.

The tensor network representation shown above is a simplified version. 
To directly extract the entanglement structure, as well as the multiplication and comultiplication, 
one needs to use the tensor network constructed from the structure constant tensors 
defined in Section~\ref{sec:tensor-network}, as shown in Figure~\ref{fig:WHAclusterTNrep}. 
The local tensors corresponding to the odd vertices (bulk edges) and even vertices (boundary edges) 
of the cluster lattice can be packed together (drawn in light red and light teal in Figure~\ref{fig:WHAclusterTNrep}), rendering the network equivalent to a standard MPS form.

\begin{figure}
    \centering
    \includegraphics[width=0.5\linewidth]{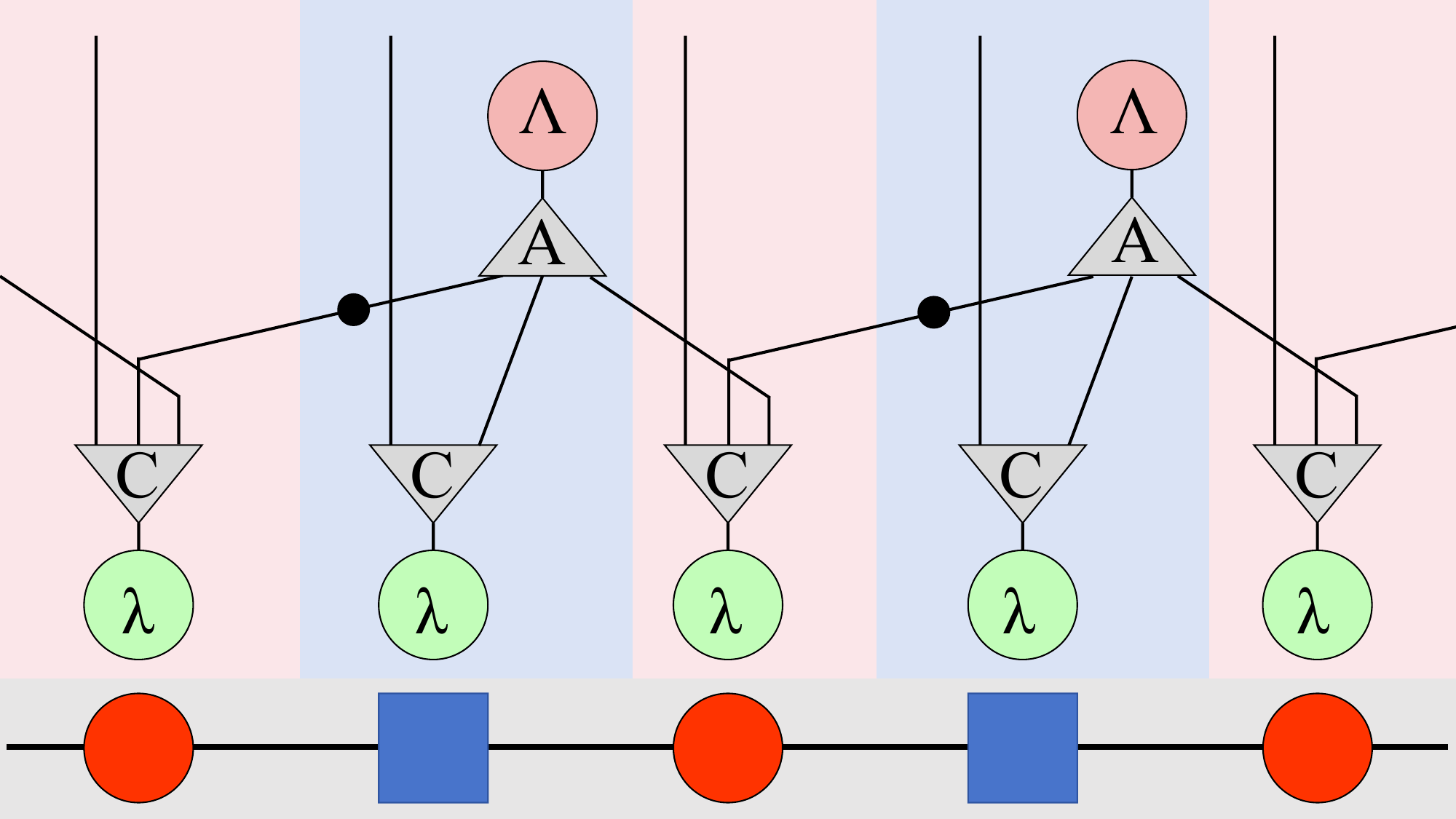}
\caption{Tensor network representation of weak Hopf cluster states. 
The bottom layer represents the cluster lattice, where red circles denote odd vertices 
and blue rectangles denote even vertices. 
The tensor network above the lattice is expressed in terms of the structure constants 
of the weak Hopf algebra.}
    \label{fig:WHAclusterTNrep}
\end{figure}

\subsection{Symmetry of the model}
\label{sec:symmetryCluster}

The correspondence between the cluster state model and the quantum double model can also be extended to their symmetry operators. 
For example, in the $\Zbb_2$ CSS-type cluster state, the symmetry operators are given by the $X$-string and $Z$-string operators. 
This correspondence can be further generalized to group-valued CSS-type cluster states and Hopf-qudit CSS-type cluster states~\cite{fechisin2023noninvertible,jia2024generalized}. 
We propose that the symmetries of the cluster state model can be characterized by ribbon operators, of which the $X$-string and $Z$-string operators are special cases. 
In particular, on an open manifold, the boundary modes are created by ribbon operators.

Ribbon operators are fundamental tools for studying topological excitations. 
In this subsection, following the works of~\cite{Kitaev2003,Bombin2008family,Cong2017,chen2021ribbon,jia2023boundary,Jia2023weak}, 
we first recall the construction of ribbon operators (see also~\cite{jia2023boundary,Jia2023weak} for more details). 
We then discuss their role in characterizing symmetry.

A direct triangle $\tau = (s_0, s_1, e)$ consists of two adjacent face-vertex links, $s_0$ and $s_1$, connected by a directed edge $e$. Similarly, a dual triangle $\tilde{\tau} = (s_0, s_1, \tilde{e})$ consists of two adjacent links $s_0$ and $s_1$ connected by a dual edge $\tilde{e}$ (by ``dual edge'' we mean an edge in dual lattice, we will draw dual edge as dashed line, the orientation of dual edge is obtained by rotating the direct edge $\pi/2$ clockwise). The triangle’s direction is defined as the direction from $s_0$ to $s_1$, which may or may not align with the edge direction.
A triangle is left-handed if its edge lies on the left when traversing in its positive direction, and right-handed otherwise. Left-handed and right-handed direct triangles are denoted by $\tau_L$ and $\tau_R$, respectively, with dual counterparts $\tilde{\tau}_L$ and $\tilde{\tau}_R$.
A ribbon $\rho$ is an ordered sequence of triangles $\tau_1, \dots, \tau_n$ with a consistent direction, such that $\partial_1 \tau_j = \partial_0 \tau_{j+1}$ for all $j$, and with no self-intersections. A ribbon is closed if it forms a loop, i.e., $\partial_1 \tau_n = \partial_0 \tau_1$. Direct and dual triangles in a directed ribbon must have opposite chiralities.
Ribbons are classified as type-A or type-B \cite{jia2023boundary,Jia2023weak}, depending on the chirality of their triangles. Type-A ribbons, denoted by $\rho_A$, consist of left-handed direct triangles and right-handed dual triangles. Conversely, type-B ribbons, denoted by $\rho_B$, consist of right-handed direct triangles and left-handed dual triangles.

The ribbon operator is constructed recursively, starting with the definition of the triangle operator and extending it through a recursive relation.
For triangles, we have (for $h,x\in H$ and $\varphi \in \hat{H}$)
\begin{align}
	&	\begin{aligned}
		\begin{tikzpicture}
			\draw[-latex,black] (1,0) -- (-1,0); 
			\node[ line width=0.2pt, dashed, draw opacity=0.5] (a) at (0,0.2){$x$};
			\draw[dotted, cyan, line width=0.5pt] (-1,0) -- (0,-1);
			\node[ line width=0.2pt, dashed, draw opacity=0.5] (a) at (-0.6,-0.7){$s_1$};
			\draw[dotted, cyan, line width=0.5pt] (1,0) -- (0,-1);
			\node[ line width=0.2pt, dashed, draw opacity=0.5] (a) at (0.6,-0.7){$s_0$};
			\draw[-stealth,gray, line width=3pt] (0.5,-0.4) -- (-0.5,-0.4); 
		\end{tikzpicture}
	\end{aligned}
	\quad
	\begin{aligned}
		F^{h,\varphi}( \tau_R)|x\rangle =  
		\varepsilon(h)  \ZL_{\varphi} |x\rangle
		=  \varepsilon(h) | x \leftharpoonup \hat{S}( \varphi) \rangle,
	\end{aligned} \label{eq:tri1} \\
	&\begin{aligned}
		\begin{tikzpicture}
			\draw[-latex,black] (-1,0) -- (1,0); 
			\node[ line width=0.2pt, dashed, draw opacity=0.5] (a) at (0,0.2){$x$};
			\draw[dotted, cyan, line width=0.5pt] (-1,0) -- (0,-1);
			\node[ line width=0.2pt, dashed, draw opacity=0.5] (a) at (-0.6,-0.7){$s_1$};
			\draw[dotted, cyan, line width=0.5pt] (1,0) -- (0,-1);
			\node[ line width=0.2pt, dashed, draw opacity=0.5] (a) at (0.6,-0.7){$s_0$};
			\draw[-stealth,gray, line width=3pt] (0.5,-0.4) -- (-0.5,-0.4); 
		\end{tikzpicture}
	\end{aligned}
	\quad
	\begin{aligned}
		F^{h,\varphi}( \tau_R)|x\rangle = 	\varepsilon(h)  \ZR_{\varphi} |x\rangle
		= \varepsilon(h) | \varphi \rightharpoonup x \rangle.
	\end{aligned}
\end{align}
% For right-handed dual triangles, we have
	\begin{align}
		&	\begin{aligned}
			\begin{tikzpicture}
				\draw[-latex,dashed,black] (1,0) -- (-1,0); 
				\node[ line width=0.2pt, dashed, draw opacity=0.5] (a) at (0,0.2){$x$};
				\draw[dotted, cyan, line width=0.5pt] (-1,0) -- (0,-1);
				\node[ line width=0.2pt, dashed, draw opacity=0.5] (a) at (-0.6,-0.7){$s_1$};
				\draw[dotted, cyan, line width=0.5pt] (1,0) -- (0,-1);
				\node[ line width=0.2pt, dashed, draw opacity=0.5] (a) at (0.6,-0.7){$s_0$};
				\draw[-stealth,gray, line width=3pt] (0.5,-0.4) -- (-0.5,-0.4); 
			\end{tikzpicture}
		\end{aligned}
		\quad
		\begin{aligned}
			F^{h,\varphi}( \tilde{\tau}_R)|x\rangle =   \hat{\varepsilon}(\varphi) \XL_h |x\rangle
			= \hat{\varepsilon}(\varphi) |  x \triangleleft  S(h) \rangle  ,
		\end{aligned}\\
		&\begin{aligned}
			\begin{tikzpicture}
				\draw[-latex,dashed,black] (-1,0) -- (1,0); 
				\node[ line width=0.2pt, dashed, draw opacity=0.5] (a) at (0,0.2){$x$};
				\draw[dotted, cyan, line width=0.5pt] (-1,0) -- (0,-1);
				\node[ line width=0.2pt, dashed, draw opacity=0.5] (a) at (-0.6,-0.7){$s_1$};
				\draw[dotted, cyan, line width=0.5pt] (1,0) -- (0,-1);
				\node[ line width=0.2pt, dashed, draw opacity=0.5] (a) at (0.6,-0.7){$s_0$};
				\draw[-stealth,gray, line width=3pt] (0.5,-0.4) -- (-0.5,-0.4); 
			\end{tikzpicture}
		\end{aligned}
		\quad
		\begin{aligned}
			F^{h,\varphi}( \tilde{\tau}_R)|x\rangle =  
			\hat{\varepsilon}(\varphi) \XR_h |x\rangle= \hat{\varepsilon}(\varphi) | h\triangleright x \rangle.
		\end{aligned} \label{eq:tri8}
	\end{align}
	% For left-handed dual triangles, we have
\begin{align}
		&	\begin{aligned}
			\begin{tikzpicture}
				\draw[-latex,dashed,black] (1,0) -- (-1,0); 
				\node[ line width=0.2pt, dashed, draw opacity=0.5] (a) at (0,0.2){$x$};
				\draw[dotted, cyan, line width=0.5pt] (-1,0) -- (0,-1);
				\node[ line width=0.2pt, dashed, draw opacity=0.5] (a) at (-0.6,-0.7){$s_0$};
				\draw[dotted, cyan, line width=0.5pt] (1,0) -- (0,-1);
				\node[ line width=0.2pt, dashed, draw opacity=0.5] (a) at (0.6,-0.7){$s_1$};
				\draw[-stealth,gray, line width=3pt] (-0.5,-0.4) -- (0.5,-0.4); 
			\end{tikzpicture}
		\end{aligned}
		\quad
		\begin{aligned}
			F^{h,\varphi}( \tilde{\tau}_L)|x\rangle =  \hat{\varepsilon}(\varphi)   \tilde{\XL}_h |x\rangle = \hat{\varepsilon}(\varphi)  
			| x\triangleleft  h \rangle.
		\end{aligned}\\
		&\begin{aligned}
			\begin{tikzpicture}
				\draw[-latex,dashed,black] (-1,0) -- (1,0); 
				\node[ line width=0.2pt, dashed, draw opacity=0.5] (a) at (0,0.2){$x$};
				\draw[dotted, cyan, line width=0.5pt] (-1,0) -- (0,-1);
				\node[ line width=0.2pt, dashed, draw opacity=0.5] (a) at (-0.6,-0.7){$s_0$};
				\draw[dotted, cyan, line width=0.5pt] (1,0) -- (0,-1);
				\node[ line width=0.2pt, dashed, draw opacity=0.5] (a) at (0.6,-0.7){$s_1$};
				\draw[-stealth,gray, line width=3pt] (-0.5,-0.4) -- (0.5,-0.4); 
			\end{tikzpicture}
		\end{aligned}
		\quad
		\begin{aligned}
			F^{h,\varphi}( \tilde{\tau}_L)|x\rangle =   \hat{\varepsilon}(\varphi) \tilde{\XR}_h |x\rangle
			= \hat{\varepsilon}(\varphi)    	
			|S(h)\triangleright  x\rangle.
		\end{aligned}
	\end{align}
	% For left-handed direct triangles, we have 
	\begin{align}
		&	\begin{aligned}
			\begin{tikzpicture}
				\draw[-latex,black] (1,0) -- (-1,0); 
				\node[ line width=0.2pt, dashed, draw opacity=0.5] (a) at (0,0.2){$x$};
				\draw[dotted, cyan, line width=0.5pt] (-1,0) -- (0,-1);
				\node[ line width=0.2pt, dashed, draw opacity=0.5] (a) at (-0.6,-0.7){$s_0$};
				\draw[dotted, cyan, line width=0.5pt] (1,0) -- (0,-1);
				\node[ line width=0.2pt, dashed, draw opacity=0.5] (a) at (0.6,-0.7){$s_1$};
				\draw[-stealth,gray, line width=3pt] (-0.5,-0.4) -- (0.5,-0.4); 
			\end{tikzpicture}
		\end{aligned}
		\quad
		\begin{aligned}
			F^{h,\varphi}( \tau_L)|x\rangle =  
			\varepsilon(h)  \tilde{\ZL}_{\varphi} |x\rangle
			=  \varepsilon(h) 
			| x \leftharpoonup \varphi \rangle ,
		\end{aligned}\\
		&\begin{aligned}
			\begin{tikzpicture}
				\draw[-latex,black] (-1,0) -- (1,0); 
				\node[ line width=0.2pt, dashed, draw opacity=0.5] (a) at (0,0.2){$x$};
				\draw[dotted, cyan, line width=0.5pt] (-1,0) -- (0,-1);
				\node[ line width=0.2pt, dashed, draw opacity=0.5] (a) at (-0.6,-0.7){$s_0$};
				\draw[dotted, cyan, line width=0.5pt] (1,0) -- (0,-1);
				\node[ line width=0.2pt, dashed, draw opacity=0.5] (a) at (0.6,-0.7){$s_1$};
				\draw[-stealth,gray, line width=3pt] (-0.5,-0.4) -- (0.5,-0.4); 
			\end{tikzpicture}
		\end{aligned}
		\quad
		\begin{aligned}
			F^{h,\varphi}( \tau_L)|x\rangle = 
				\varepsilon(h)  \tilde{\ZR}_{\varphi} |x\rangle
			= \varepsilon(h) 	|  \hat{S}( \varphi) \rightharpoonup  x \rangle.
		\end{aligned}  \label{eq:tri16}
	\end{align}

For a type-B ribbon $\rho = \rho_B$ and an element $h \otimes \varphi \in H \otimes \hat{H}$, the ribbon operator is defined as $F^{h,\varphi}(\rho) = F^{h \otimes \varphi}(\rho_B)$. The operator acts non-trivially only on the edges of $\rho_B$. Consider the decomposition $\rho = \rho_B = \rho_1 \cup \rho_2$, where both $\rho_1$ and $\rho_2$ are oriented in the same direction as $\rho$, and $\partial_1 \rho_1 = \partial_0 \rho_2$. For $h \otimes \varphi \in H^{\rm op} \otimes \hat{H}$, the ribbon operator on the composite ribbon is defined as
\begin{equation}
\begin{aligned}
    F^{h,\varphi}(\rho) = &\sum_{(h \otimes \varphi)} F^{(h \otimes \varphi)^{(1)}}(\rho_1) F^{(h \otimes \varphi)^{(2)}}(\rho_2) \\
= & \sum_k \sum_{(k), (h)} F^{h^{(1)}, \hat{k}}(\rho_1) F^{S(k^{(3)}) h^{(2)} k^{(1)}, \varphi(k^{(2)} \bullet)}(\rho_2).
\end{aligned}
\end{equation}
 where $\{k\}$ is an orthogonal basis of $H$ with $\{\hat{k}\}$ its dual basis and the comultiplication is taken in the dual weak Hopf algebra $D(H)^{\vee}$.
This definition is independent of the decomposition $\rho = \rho_1 \cup \rho_2$. A similar construction applies for type-A ribbons.

\paragraph{Symmetry for closed 1d ladder}
For closed 1d ladder on closed manifold $\mathbb{S}^1$, we claim that there are two symmetries comes from rough and smooth boundaries respectively:
\begin{equation}
    \Cocom(H) \times \Cocom(\hat{H}).
\end{equation}
For any $h \in \Cocom(H)$, the associated symmetry operator is given by $X$-string operator
\begin{equation}
    S_h = \sum_{(h)}\XR_{h^{\cone}} \otimes \XR_{h^{\ctwo}} \otimes \cdots \otimes \XR_{h^{(n)}} \\
    = \sum_{(h)}
    \begin{aligned}
\begin{tikzpicture}
    % Define the number of rungs in the ladder
    \def\n{5}
    % Define the size of each square
    \def\s{1}
        % Draw the shaded background lattice
    \fill[green!20] (0, 0) rectangle (\n*\s+\s, \s); % Rectangle covering the whole background
    % Draw the ladder with arrows in the middle of each edge
    \foreach \i in {0,...,\n} {
        % Draw solid bottom edges with arrows in the middle pointing right
        \draw[-stealth, line width=1.0pt,blue, midway] (\i*\s, 0) -- (\i*\s+\s, 0);
        % Draw dotted top edges with arrows in the middle pointing right
        \draw[dotted, line width=1.0pt,red, midway] (\i*\s, \s) -- (\i*\s+\s, \s);
        % Draw upward ladder edges with arrows in the middle pointing up
        \draw[-stealth,line width=1.0pt, midway] (\i*\s, 0) -- (\i*\s, \s);
    }
    % Draw the right-most vertical ladder edge with an arrow in the middle pointing up
    \draw[-stealth, midway,line width=1.0pt] (\n*\s+\s, 0) -- (\n*\s+\s, \s);
    \node[ line width=0.2pt, dashed, draw opacity=0.5] (a) at (.45,0.5){$\XR_{h^{\cone}}$};
    \node[ line width=0.2pt, dashed, draw opacity=0.5] (a) at (1.45,0.5){$\XR_{h^{\ctwo}}$};
    \node[ line width=0.2pt, dashed, draw opacity=0.5] (a) at (2.45,0.5){$\XR_{h^{\cthree}}$};
        \node[ line width=0.2pt, dashed, draw opacity=0.5] (a) at (4.45,0.5){$\cdots$};
     \node[ line width=0.2pt, dashed, draw opacity=0.5] (a) at (5.45,0.5){$\XR_{h^{(n)}}$};
      \draw[-stealth, white, line width=2pt, midway] (6, 0) -- (6, 1.02);
\end{tikzpicture}
\end{aligned}
\end{equation}
which acts on all vertex edges of the lattice (notice that we have assume the periodic boundary condition in the above lattice, i.e. $(n+1=1 )\operatorname{mod} n$ ).
It can be verified that this operator commutes with both the vertex and face terms of the Hamiltonian, i.e.,
\begin{equation}
    [S_h, \Av_{v_s}^{\lambda}] = [S_h, \Bf_f^{\Lambda}] = 0 \quad \forall v_s, f.
\end{equation}
For $\psi\in \Cocom(\hat{H})$, the symmetry operator is given by 
\begin{equation}
    T_{\psi}=\sum_{(\psi)} \ZR_{\psi^{\cone}} \otimes \ZR_{\psi^{\ctwo}} \otimes \cdots \otimes \ZR_{\psi^{(n)}}
    = \sum_{(\psi)}
    \begin{aligned}
\begin{tikzpicture}
    % Define the number of rungs in the ladder
    \def\n{5}
    % Define the size of each square
    \def\s{1}
        % Draw the shaded background lattice
    \fill[green!20] (0, 0) rectangle (\n*\s+\s, \s); % Rectangle covering the whole background
    % Draw the ladder with arrows in the middle of each edge
    \foreach \i in {0,...,\n} {
        % Draw solid bottom edges with arrows in the middle pointing right
        \draw[-stealth, line width=1.0pt,blue, midway] (\i*\s, 0) -- (\i*\s+\s, 0);
        % Draw dotted top edges with arrows in the middle pointing right
        \draw[dotted, line width=1.0pt,red, midway] (\i*\s, \s) -- (\i*\s+\s, \s);
        % Draw upward ladder edges with arrows in the middle pointing up
        \draw[-stealth,line width=1.0pt, midway] (\i*\s, 0) -- (\i*\s, \s);
    }
    % Draw the right-most vertical ladder edge with an arrow in the middle pointing up
    \draw[-stealth, midway,line width=1.0pt] (\n*\s+\s, 0) -- (\n*\s+\s, \s);
    \node[ line width=0.2pt, dashed, draw opacity=0.5] (a) at (.5,-0.5){\color{blue}$\ZR_{\psi^{\cone}}$};
    \node[ line width=0.2pt, dashed, draw opacity=0.5] (a) at (1.5,-0.5){\color{blue}$\ZR_{\psi^{\ctwo}}$};
    \node[ line width=0.2pt, dashed, draw opacity=0.5] (a) at (2.5,-0.5){\color{blue}$\ZR_{\psi^{\cthree}}$};
    \node[ line width=0.2pt, dashed, draw opacity=0.5] (a) at (4.5,-0.5){\color{blue}$\cdots$};
     \node[ line width=0.2pt, dashed, draw opacity=0.5] (a) at (5.5,-0.5){\color{blue}$\ZR_{\psi^{(n)}}$};
      \draw[-stealth, white, line width=2pt, midway] (6, 0) -- (6, 1.02);
\end{tikzpicture}
\end{aligned}
\end{equation}
which acts on all horizontal edges.
It can be proved that this operator commutes with both the vertex and face terms of the Hamiltonian, i.e.,
\begin{equation}
    [T_{\psi}, \Av_{v_s}^{\lambda}] = [T_{\psi}, \Bf_f^{\Lambda}] = 0 \quad \forall v, f.
\end{equation}
For $\Gamma \in \Rep(H)$, its character $\chi_{\Gamma}$ lies in $\Cocom(\hat{H})$. 
Therefore, the $H$ cluster state model also possesses $\Rep(H)$ as a sub-symmetry. 
While the fusion ring of $\Rep(H)$ has non-negative integer coefficients and $\Cocom(\hat{H})$ may involve complex coefficients, 
they share the same basis consisting of the characters of all irreducible representations (see Proposition~\ref{prop:characterBasis}). 
With a slight abuse of terminology, one may also regard $\Cocom(\hat{H})$ as being equivalent to $\Rep(H)$ as symmetries of the model.

There is another way to derive the symmetry operators and to understand the symmetries of the model.
An 1d closed cluster lattice can also be regarded as a bicycle-wheel quantum double lattice, which is a specific cellulation of a disk. We assign a weak Hopf qudit to each edge, corresponding to the sites of the cluster lattice. The boundary edges correspond to even sites, while the internal bulk edges (drawn in red here) correspond to odd sites. For example, we have
\begin{equation}
\begin{aligned}
		\begin{tikzpicture}
               % \filldraw[color=blue!60, fill=red!5, very thick](0,0) circle (1.5);
               \filldraw[color=blue!60, fill=green!20, very thick](0,0) circle (1.5);
                \draw[dotted,gray,very thick] (0.95,0.95) -- (0,0.7) -- (-0.95,0.95)-- (-0.7,0) --(-0.95,-0.95)--(0,-0.7)--(0.95,-0.95)--(0.7,0)--(0.95,0.95);
			\draw[latex-,red!60,very thick] (0,0) -- (0,-1.5); % Reversed red arrow
			\draw[latex-,red!60,very thick] (0,0) -- (0,1.5);  % Reversed red arrow
   			\draw[latex-,red!60,very thick] (0,0) -- (1.5,0);  % Reversed red arrow
         	\draw[latex-,red!60,very thick] (0,0) -- (-1.5,0); % Reversed red arrow
                \draw[blue!60, very thick, -latex] (1.5,0) arc (0:20:1.5);
                \draw[blue!60, very thick, -latex] (0,1.5) arc (90:110:1.5);
                \draw[blue!60, very thick, -latex] (-1.5,0) arc (180:200:1.5);
                \draw[blue!60, very thick, -latex] (0,-1.5) arc (270:290:1.5);
                \draw[blue!60, fill=blue!60,very thick] (0.95,0.95) rectangle ++(0.2,.2);
                \draw[blue!60, fill=blue!60,very thick] (0.95,0.95) rectangle ++(0.2,.2);
                \draw[blue!60, fill=blue!60,very thick] (-0.95,0.95) rectangle ++(-0.2,.2);
                \draw[blue!60, fill=blue!60,very thick] (-0.95,-0.95) rectangle ++(-0.2,-.2);
                \draw[blue!60, fill=blue!60,very thick] (0.95,-0.95) rectangle ++(0.2,-.2);
			\draw [fill = red!60,red!60] (0.7,0) circle (0.15);
               \draw [fill = red!60,red!60] (-0.7,0) circle (0.15);
               \draw [fill = red!60,red!60] (0,0.7) circle (0.15);
               \draw [fill = red!60,red!60] (0,-0.7) circle (0.15);
			%\node[ line width=0.2pt, dashed, draw opacity=0.5] (a) at (0.7,0){$x_1$};
		\end{tikzpicture}
	\end{aligned} \quad \quad  \Leftrightarrow \quad 
     \begin{aligned}
		\begin{tikzpicture}
               % \filldraw[color=blue!60, fill=red!5, very thick](0,0) circle (1.5);
                \draw[dotted,gray,very thick] (0.95,0.95) -- (0,0.7) -- (-0.95,0.95)-- (-0.7,0) --(-0.95,-0.95)--(0,-0.7)--(0.95,-0.95)--(0.7,0)--(0.95,0.95);
			%\draw[-latex,red!60,very thick] (0,0) -- (0,1.5);
			%\draw[-latex,red!60,very thick] (0,0) -- (0,-1.5);
   			%\draw[-latex,red!60,very thick] (0,0) -- (-1.5,0); 
         	%\draw[-latex,red!60,very thick] (0,0) -- (1.5,0); 
                %\draw[blue!60, very thick, -latex] (1.5,0) arc (0:20:1.5);
                %\draw[blue!60, very thick, -latex] (0,1.5) arc (90:110:1.5);
                %\draw[blue!60, very thick, -latex] (-1.5,0) arc (180:200:1.5);
                %\draw[blue!60, very thick, -latex] (0,-1.5) arc (270:290:1.5);
                \draw[blue!60, fill=blue!60,very thick] (0.95,0.95) rectangle ++(0.2,.2);
                \draw[blue!60, fill=blue!60,very thick] (0.95,0.95) rectangle ++(0.2,.2);
                \draw[blue!60, fill=blue!60,very thick] (-0.95,0.95) rectangle ++(-0.2,.2);
                \draw[blue!60, fill=blue!60,very thick] (-0.95,-0.95) rectangle ++(-0.2,-.2);
                \draw[blue!60, fill=blue!60,very thick] (0.95,-0.95) rectangle ++(0.2,-.2);
			\draw [fill = red!60,red!60] (0.7,0) circle (0.15);
               \draw [fill = red!60,red!60] (-0.7,0) circle (0.15);
               \draw [fill = red!60,red!60] (0,0.7) circle (0.15);
               \draw [fill = red!60,red!60] (0,-0.7) circle (0.15);
			%\node[ line width=0.2pt, dashed, draw opacity=0.5] (a) at (0.7,0){$x_1$};
		\end{tikzpicture}
	\end{aligned}
\end{equation}
The solid lattice represents the quantum double lattice, while the dotted lattice represents the cluster lattice, where the qudits are placed on the vertices. In the context of the cluster state model, qudits are typically placed on the vertices of the cluster lattice. To distinguish these from the vertices of the quantum double lattice, we refer to them as sites. Even sites are depicted as blue squares, and odd sites as red dots (see also Figure~5 of Ref.~\cite{jia2024generalized} for an illustration of how the 1d cluster lattice is folded into a 2d configuration).
Each bicycle-wheel quantum double lattice corresponds to a 1d cluster lattice, and every 1d cluster lattice can be transformed into a bicycle-wheel quantum double lattice. For convenience, we will use these two lattices interchangeably.

The symmetry operators can be easily understood from the perspective of the bicycle-wheel lattice. In this framework, the rough boundary corresponds to a vertex in the bicycle-wheel lattice. Consider a closed ribbon $\sigma$ surrounding the vertex $w$, consisting exclusively of dual triangles (illustrated in gray as follows):

\begin{equation*}
    \begin{aligned}
\begin{tikzpicture}
    % Draw the circle
    \filldraw[color=blue!60, fill=green!20, very thick](0,0) circle (1.5);
    
    % Draw the rectangle with dotted cyan boundary
    \filldraw[fill=gray!20, draw=black, dotted, thick] (-0.6,-0.6) rectangle (0.6,0.6);
    
    % Draw the dotted cyan lines from rectangle vertices to the center
    \draw[dotted, cyan, thick] (-0.6,-0.6) -- (0,0);
    \draw[dotted, cyan, thick] (-0.6,0.6) -- (0,0);
    \draw[dotted, cyan, thick] (0.6,0.6) -- (0,0);
    \draw[dotted, cyan, thick] (0.6,-0.6) -- (0,0);
    
    % Draw the reversed red arrows
    \draw[latex-,red!60,very thick] (0,0) -- (0,-1.5); % Reversed red arrow
    \draw[latex-,red!60,very thick] (0,0) -- (0,1.5);  % Reversed red arrow
    \draw[latex-,red!60,very thick] (0,0) -- (1.5,0);  % Reversed red arrow
    \draw[latex-,red!60,very thick] (0,0) -- (-1.5,0); % Reversed red arrow
    
    % Draw the blue arcs
    \draw[blue!60, very thick, -latex] (1.5,0) arc (0:20:1.5);
    \draw[blue!60, very thick, -latex] (0,1.5) arc (90:110:1.5);
    \draw[blue!60, very thick, -latex] (-1.5,0) arc (180:200:1.5);
    \draw[blue!60, very thick, -latex] (0,-1.5) arc (270:290:1.5);
    
    % Draw the counterclockwise thick dark gray arrow around the center with smaller radius
    \draw[gray, very thick, -to] (0.4,0) arc[start angle=0, end angle=290, radius=0.4];
\end{tikzpicture}
	\end{aligned}
\end{equation*}
The corresponding ribbon operator is simply the vertex operator 
\begin{equation}
    F^{h, \hat{1}}(\sigma) = \Av_w^h,
\end{equation}
as there is no direct triangle. Thus, the operator is independent of the unit $\hat{1} \in \hat{H}$ (we retain this $\hat{1}$ to emphasize that this is a special case of a more general situation, which will be discussed later). In the finite group case, this operator reduces to the Pauli $X$-string operator. Note that this symmetry operator acts only on the red edges, which correspond to the odd sites of the cluster lattice.

It is clear $[\Av_w^h, \Av_v] = 0$ for all vertex operators $A_v$ on the smooth boundary (depicted as the blue circle). This is because, for the two endpoints of an edge, one operator acts as $\XR$ and the other as $\XL$, which commute with each other. 
Using a result from \cite[Corollary 1]{chang2014kitaev}, if $h$ is cocommutative, then $[\Av_w^h, \Bf_f] = 0$ for all face operators. This implies that the symmetry arising from the rough boundary is 
\begin{equation}
    \operatorname{Sym}_{\mathrm{rough}} = \operatorname{Cocom}(H),
\end{equation}
namely, the set of all cocommutative elements in $H$. $h\in \operatorname{Cocom}(H)$
For $H = \mathbb{C}[G]$, it follows that $\operatorname{Cocom}(\mathbb{C}[G]) = \mathbb{C}[G]$, and thus the symmetry is $\mathbb{C}[G]$. This conclusion also aligns with a macroscopic perspective, where we assert that the rough boundary symmetry is $\operatorname{Rep}(\hat{H})$, which can be embedded into $H$. Note that all characters are cocommutative, as $\chi_\Gamma(fg) = \chi_\Gamma(gf)$ for all $\Gamma \in \operatorname{Rep}(\hat{H})$. Here, the comultiplication of $\chi_\Gamma$ is defined in Definition~\ref{def:pairingDef}.

For the smooth boundary, we can further deform the bicycle-wheel lattice into a cone as follows:
\begin{equation*}
    \begin{tikzpicture}
    % Draw the base of the cone
    \fill[gray!70] (0,0) ellipse (1.5 and 0.5); % Base ellipse (perspective)
   
    \draw[blue!60, very thick, -latex] (1.5,0) arc[start angle=0, end angle=60, x radius=1.5, y radius=0.5];

    \draw[blue!60, very thick, -latex] (1.5,0) arc[start angle=0, end angle=130, x radius=1.5, y radius=0.5];
    \draw[blue!60, very thick, -latex] (1.5,0) arc[start angle=0, end angle=220, x radius=1.5, y radius=0.5];
    \draw[blue!60, very thick, -latex] (1.5,0) arc[start angle=0, end angle=320, x radius=1.5, y radius=0.5];
    \draw[blue!60, very thick] (0,0) ellipse (1.5 and 0.5); % Outline of the base

    % Apex of the cone
    \coordinate (Apex) at (0,1.5); % Define the apex

    % Draw the surface of the cone
    \fill[green!20, opacity=0.7] (Apex) -- (-1.5,0) arc[start angle=180, end angle=360, x radius=1.5, y radius=0.5] -- cycle;

    % Draw arrowed lines connecting the apex to the base (arrows in the middle)
    \draw[thick, red!60] (1.52,0) -- node[midway] {} (Apex);  % Right side with arrow pointing to Apex
    \draw[thick, red!60] (-1.52,0) -- node[midway] {} (Apex); % Left side with arrow pointing to Apex
    \draw[thick, red!60] (0.3,0.45) -- node[midway] {} (Apex);   % Top-center with arrow pointing to Apex
    \draw[thick, red!60] (-0.3,-0.53) -- node[midway] {} (Apex);  % Bottom-center with arrow pointing to Apex

    % Right side (half length)
    \coordinate (RightMid) at (1.52,0); % Midpoint of the original line
    \draw[thick, red!60, -latex] (RightMid) -- node[midway, left] {} ($ (Apex)!0.5!(1.52,0) $);  % Right side (half length) with arrow pointing to Apex

    % Left side (half length)
    \coordinate (LeftMid) at (-1.52,0); % Midpoint of the original line
    \draw[thick, red!60, -latex] (LeftMid) -- node[midway, left] {} ($ (Apex)!0.5!(-1.52,0) $);  % Left side (half length) with arrow pointing to Apex

    % Top-center (half length)
    \coordinate (TopMid) at (0.3,0.45); % Midpoint of the original line
    \draw[thick, red!60, -latex] (TopMid) -- node[midway, right] {} ($ (Apex)!0.5!(0.3,0.45) $);  % Top-center (half length) with arrow pointing to Apex

    % Bottom-center (half length)
    \coordinate (BottomMid) at (-0.3,-0.53); % Midpoint of the original line
    \draw[thick, red!60, -latex] (BottomMid) -- node[midway, right] {} ($ (Apex)!0.5!(-0.3,-0.53) $);  % Bottom-center (half length) with arrow pointing to Apex

    % Draw cyan dotted lines from the center of the base to the bottom ends of red lines
    \draw[cyan, dotted] (0,0) -- (1.52,0);  % Right bottom end
    \draw[cyan, dotted] (0,0) -- (-1.52,0); % Left bottom end
    \draw[cyan, dotted] (0,0) -- (0.3,0.53);  % Bottom-center
    \draw[cyan, dotted] (0,0) -- (-0.3,-0.53);  % Bottom-center
    \end{tikzpicture}
\end{equation*}
On the base face, we also have a ribbon $\xi$ consisting only of direct triangles surrounding the face $f_s$. The corresponding symmetry ribbon operator is simply the face operator 
\begin{equation}
    F^{1,\varphi}(\xi) = \Bf_{f_s}^\varphi,
\end{equation}
as there are no dual triangles. Thus, the operator does not depend on elements in $H$. 

The symmetry operator $\Bf_{f_s}^\varphi$ commutes with all face operators of the cluster state model because, on the two sides of an edge, the operators are chosen as $\ZR$ and $\ZL$ respectively, which always commute. Moreover, if $\varphi$ is cocommutative, a result from Ref.~\cite[Corollary 1]{chang2014kitaev} shows that $\Bf_{f_s}^\varphi$ also commutes with all vertex operators of the cluster state model. 

This implies that the symmetry arising from the smooth boundary is 
\begin{equation}
    \operatorname{Sym}_{\mathrm{smooth}} = \Cocom(\hat{H}),
\end{equation}
namely, the subalgebra of all cocommutative elements in $\hat{H}$. 
From the macroscopic analysis, all characters in $\operatorname{Rep}(H)$ are cocommutative elements in $\hat{H}$, which matches well with the result obtained from the lattice here.

The symmetry operators arising from the rough and smooth boundaries commute with each other, as they act non-trivially on different weak Hopf qudits. This is analogous to the $\mathbb{Z}_2$ cluster state (as well as finite-group or Hopf cluster states), where one symmetry operator acts on odd sites and the other acts on even sites.

\begin{theorem}
    The weak Hopf cluster state model on a closed manifold exhibits the symmetry
    \begin{equation}
        \operatorname{Sym}_{\rm closed\, cluster} = \Cocom(H) \times \Cocom(\hat{H}),
    \end{equation}
    which is a subset of \( H \times \hat{H} \). 
    Moreover, note that \( \operatorname{Gr}(\Rep(H)) \) is a subalgebra of \( \Cocom(\hat{H}) \) in the sense that fusion coefficients are non-negative integers.  Consequently, this symmetry contains a weak Hopf sub-symmetry \( \Cocom(H) \times \Rep(H) \).
\end{theorem}

We should also consider the locality of the symmetry operators, namely, they should preserve locality by mapping local operators to nearby local operators. 
In our model, this locality condition is indeed satisfied. 
This follows directly from the fact that the symmetry operators are a special class of ribbon operators. 
According to Ref.~\cite[Lemma~5]{Jia2023weak}, ribbon operators map local operators to local operators. 
A subtlety arises from the fact that the symmetry operator is not on-site in the usual sense; its locality-preserving property should be understood in a generalized sense, involving the comultiplication structure.

\paragraph{Symmetry for open 1d ladder}
For an open 1d manifold, the corresponding weak Hopf cluster state model exhibits a larger weak Hopf symmetry: $H \times \hat{H}$. Since the symmetry ribbons on the rough and smooth boundaries are not required to be closed ribbons, the elements assigned to the ribbon operators are no longer constrained to be cocommutative. Nonetheless, the symmetry ribbon operators still commute with the vertex and face operators in the cluster state model \cite{Jia2023weak,chang2014kitaev}. 

In the case of a finite group, this is analogous to the group symmetry and dual group symmetry $G \times \hat{G}$ \cite{Ji2020categoricalsym, Kong2020algebraic}. Using our notation, $\Cocom(\mathbb{C}[G]) = \mathbb{C}[G]$ and $\Cocom(\mathbb{C}[G]^{\vee}) = \Rep(G)$, indicating that both the closed ladder and open ladder share the same symmetry, given by $\mathbb{C}[G] \times \Rep(G)$.

When acting on the ground state of the cluster state model, the symmetry ribbon operators create edge modes at the two ends of the open 1D lattice. Consider a more general ribbon operator as illustrated below:

\begin{equation*}
   \begin{aligned}
\begin{tikzpicture}
    % Define the number of rungs in the ladder
    \def\n{5}
    % Define the size of each square
    \def\s{1}
        % Draw the shaded background lattice
    \fill[green!20] (0, 0) rectangle (\n*\s+\s, \s); % Rectangle covering the whole background
    % Draw the ladder with arrows in the middle of each edge
    \foreach \i in {0,...,\n} {
        % Draw solid bottom edges with arrows in the middle pointing right
        \draw[-stealth, line width=1.0pt,blue, midway] (\i*\s, 0) -- (\i*\s+\s, 0);
        % Draw dotted top edges with arrows in the middle pointing right
        \draw[dotted, line width=1.0pt,red, midway] (\i*\s, \s) -- (\i*\s+\s, \s);
        % Draw upward ladder edges with arrows in the middle pointing up
        \draw[-stealth,line width=1.0pt, midway] (\i*\s, 0) -- (\i*\s, \s);
    }
    % Draw the right-most vertical ladder edge with an arrow in the middle pointing up
    \draw[-stealth, midway,line width=1.0pt] (\n*\s+\s, 0) -- (\n*\s+\s, \s);

\filldraw[gray!10, opacity=0.7] (0,0) -- (-0.5,0.5) -- (6.5, 0.5) -- (6,0) -- cycle;

   \draw[black,dashed, line width = 1pt] (-0.5,0.5) -- (6.5,0.5);
   \draw[cyan,dotted, line width = 1pt] (0,0) -- (-0.5,0.5);
   \draw[cyan,dotted, line width = 1pt] (0,0) -- (0.5,0.5);
  \draw[cyan,dotted, line width = 1pt] (1,0) -- (0.5,0.5);
  \draw[cyan,dotted, line width = 1pt] (1,0) -- (1.5,0.5);
    \draw[cyan,dotted, line width = 1pt] (2,0) -- (1.5,0.5);
        \draw[cyan,dotted, line width = 1pt] (2,0) -- (2.5,0.5);
          \draw[cyan,dotted, line width = 1pt] (3,0) -- (2.5,0.5);
                \draw[cyan,dotted, line width = 1pt] (3,0) -- (3.5,0.5);
 \draw[cyan,dotted, line width = 1pt] (4,0) -- (3.5,0.5);
\draw[cyan,dotted, line width = 1pt] (4,0) -- (4.5,0.5);
\draw[cyan,dotted, line width = 1pt] (5,0) -- (4.5,0.5);
\draw[cyan,dotted, line width = 1pt] (5,0) -- (5.5,0.5);
\draw[cyan,dotted, line width = 1pt] (6,0) -- (5.5,0.5);
\draw[cyan,dotted, line width = 1pt] (6,0) -- (6.5,0.5);
\end{tikzpicture}
\end{aligned} 
\end{equation*}
In this case, the ribbon operator $F^{h,\varphi}(\rho)$ is obtained. A result in \cite[Proposition 3]{Jia2023weak} guarantees that $F^{h,\varphi}(\rho)$ commutes with all vertex and face operators of the cluster state model. At the two ends of the boundary, it creates edge modes.

\begin{theorem}
    The weak Hopf cluster state model on an open manifold exhibits the weak Hopf symmetry
    \begin{equation}
        \operatorname{Sym}_{\rm open\, cluster} = H \times \hat{H}.
    \end{equation}
    This symmetry contains a weak Hopf sub-symmetry \( H \times \Rep(H) \).
\end{theorem}

\begin{remark}[String-net realization]
In the multifusion string-net framework \cite{Levin2005, Kitaev2012boundary, jia2024weakTube}, the input data is a unitary multifusion category $\eC = \Rep(H)$ associated with a weak Hopf gauge symmetry $H$. The smooth boundary is described by the $\eC$-module category $\eM_s = \eC$, whose boundary excitations are given by 
\[
\eB_s = \mathsf{Fun}_{\eC}(\eM_s, \eM_s) = \eC,
\]
where $\mathsf{Fun}_{\eC}(\eC, \eC)$ denotes the category of all $\eC$-module functors from $\eM_s$ to itself.  

For the ordinary Hopf case, the rough boundary is described by the $\eC$-module category $\eM = \mathsf{Vect}_{\Cbb}$, where $\mathsf{Vect}_{\Cbb}$ denotes the category of all finite-dimensional complex vector spaces. The corresponding rough boundary excitations are
\[
\eB_r = \mathsf{Fun}_{\eC}(\mathsf{Vect}_{\Cbb}, \mathsf{Vect}_{\Cbb}) = \Rep(\hat{H}).
\]
Using the Kitaev–Kong construction, one can also obtain a lattice realization for the cluster state model within this framework.  

However, for a general anomalous weak Hopf symmetry, $\mathsf{Vect}_{\Cbb}$ is not a module category over $\Rep(H)$—equivalently, $\Cbb$ is not a comodule algebra over $\Rep(H)$. Nevertheless, the cluster state model can still be constructed by removing the degrees of freedom on the rough boundary. The resulting phase can then be understood as a partially symmetry-broken phase.
\end{remark}

%\subsection{Weal Hopf Cluster state model can realize all fusion $\eS$ SymTFT}

\section{Weak Hopf cluster ladder model II: general construction}
\label{sec:latticeII}
In this section, we extend the cluster state model to the more general cluster ladder model, where the topological boundary conditions on the two boundaries can be chosen arbitrarily.

\subsection{Weak Hopf cluster ladder model}

For the most general case, consider a sandwich manifold \(\Sigma \times [0,1]\). For the bulk \(\Sigma \times (0,1)\), we assign a weak Hopf algebra. On the symmetry boundary \(\Sigma \times \{0\}\), we assign a weak Hopf comodule algebra \(K\). Depending on the orientation of the boundary, we can choose either left or right comodule algebras as input data \cite{jia2023boundary, Jia2023weak}. Similarly, on the physical boundary \(\Sigma \times \{1\}\), we assign a weak Hopf comodule algebra \(J\).

The compactification over the interval \([0,1]\) can be realized by considering an ultra-thin sandwich lattice, which is a ladder lattice denoted by \(\mathbb{M}^1\):
\begin{equation}\label{eq:ClusterLadder}
\begin{aligned}
\begin{tikzpicture}
    % Define the number of rungs in the ladder
    \def\n{5}
    % Define the size of each square
    \def\s{1}
    
    % Draw the shaded background lattice
    \fill[green!20] (0, 0) rectangle (\n*\s+\s, \s); % Rectangle covering the whole background

    % Draw the ladder with arrows in the middle of each edge
    \foreach \i in {0,...,\n} {
        % Draw solid bottom edges with arrows in the middle pointing right
        \draw[-stealth, line width=1.0pt,blue, midway] (\i*\s, 0) -- (\i*\s+\s, 0);
        % Draw dotted top edges with arrows in the middle pointing right
        \draw[-stealth, line width=1.0pt,red, midway] (\i*\s, \s) -- (\i*\s+\s, \s);
        % Draw upward ladder edges with arrows in the middle pointing up
        \draw[-stealth,line width=1.0pt, midway] (\i*\s, 0) -- (\i*\s, \s);
    }
    % Draw the right-most vertical ladder edge with an arrow in the middle pointing up
    \draw[-stealth, midway,line width=1.0pt] (\n*\s+\s, 0) -- (\n*\s+\s, \s);
\end{tikzpicture}
\end{aligned}
\end{equation}
The edges for the symmetry boundary are drawn in blue, the edges for the physical boundary are drawn in red, and the bulk edges are drawn in black. Note that in the above lattice, the bulk is positioned on the left side of the symmetry boundary and on the right side of the physical boundary when traversing along the positive direction of the boundary. Consequently, the corresponding comodule algebras \(K\) and \(J\) must be chosen as a left \(H\)-comodule algebra and a right \(H\)-comodule algebra, respectively. 

We will adopt Sweedler's notation for the left $H$-comodule algebra \(K\) as \(\beta_K(x) = \sum_{[x]} x^{[-1]} \otimes x^{[0]}\), where \(x^{[-1]} \in H\) and \(x^{[0]} \in K\). Similarly, for the right $H$-comodule algebra \(J\), we write \(\beta_J(x) = \sum_{[x]} x^{[0]} \otimes x^{[1]}\), where \(x^{[0]} \in J\) and \(x^{[1]} \in H\).

For the symmetry boundary, we assign comodule algebra $\mathcal{H}_{e_s}=K$ on each boundary edge.
To build the boundary local stabilizers, we need the notion of \emph{symmetric separability idempotent} of an algebra $K$ (see \cite{aguiar2000note,koppen2020defects,jia2023boundary,Jia2023weak}).
By definition, it is an element $\zeta = \sum_{\langle \zeta \rangle} \zeta^{\langle 1 \rangle} \otimes \zeta^{\langle 2 \rangle} \in K \otimes K$ such that
\begin{enumerate}
    \item $\sum_{\langle \zeta \rangle} x \zeta^{\langle 1 \rangle} \otimes \zeta^{\langle 2 \rangle} = \sum_{\langle \zeta \rangle} \zeta^{\langle 1 \rangle} \otimes \zeta^{\langle 2 \rangle} x$ for all $x \in K$;
    \item $\sum_{\langle \zeta \rangle} \zeta^{\langle 1 \rangle} \zeta^{\langle 2 \rangle} = 1$;
    \item $\sum_{\langle \zeta \rangle} \zeta^{\langle 1 \rangle} \otimes \zeta^{\langle 2 \rangle} = \sum_{\langle \zeta \rangle} \zeta^{\langle 2 \rangle} \otimes \zeta^{\langle 1 \rangle}$.
\end{enumerate}
We have used ``$\sum_{\langle \zeta \rangle}$'' to represent that $\zeta$ is a linear combination of elements in $K \otimes K$, in order to make the notation consistent with that for comultiplication.
It can be shown that $\zeta$ is an idempotent element of the enveloping algebra $K \otimes K^{\rm op}$. If $K$ is a weak Kac algebra \footnote{Recall that weak Kac algebra is a weak Hopf algebra such that $S^2=\id$, see Refs.~\cite{yamanouchi1994duality,bohm2000weakII}. All finite dimensional $C^*$ Hopf algebras are weak Kac algebras.} with Haar integral $\lambda_K$, then it is straightforward to verify that $\zeta = \sum_{(\lambda_K)} \lambda_K^{(1)} \otimes S(\lambda_K^{(2)})$ serves as a symmetric separability idempotent. The existence and uniqueness of the symmetric separability idempotent for a finite-dimensional semisimple algebra over an algebraically closed field of characteristic zero are established in \cite[Corollary 3.1]{aguiar2000note}.

For the symmetry boundary, as depicted in Eq.~\eqref{eq:ClusterLadder}, with the bulk on the left-hand side of the boundary, we define the following operator for \(z \otimes w \in K \otimes K^{\mathrm{op}}\):
\begin{align}
 %    \begin{aligned}
  %   \begin{tikzpicture}
   % \begin{scope}[rotate=90] % Rotate the entire scope by 90 degrees
    %    \fill[green!20] (0, 0) rectangle (1,2); % Rectangle 
     %   \draw[-stealth,blue,line width = 1.6pt] (0,1) -- (0,0);
      %  \draw[cyan, dotted, line width = 1pt] (0,1) -- (0.5,0.5); 
%         \draw[-stealth,blue,line width = 1.6pt] (0,2) -- (0,1); 
%         \draw[-stealth,black] (0,1) -- (1,1); 
%        	\node[ line width=0.2pt, dashed, draw opacity=0.5] (a) at (-0.4,0.5){$y$};
%     		\node[ line width=0.2pt, dashed, draw opacity=0.5] (a) at (-0.4,1.5){$x$};
%             \node[ line width=0.2pt, dashed, draw opacity=0.5] (a) at (0.5,1.3){$h$};
%             \node[ line width=0.2pt, dashed, draw opacity=0.5] (a) at (0.7,0.5){$s$}; 
%     \end{scope}
% \end{tikzpicture}
%     \end{aligned} & 
%  \quad \quad  
%  \Av^{z\otimes w}(s)|x,y,h\rangle=\sum_{[w]} |w^{[0]} x , y z, h S^{-1}(w^{[-1]})\rangle, \label{eq:Avs1}\\
       \begin{aligned}
           \begin{tikzpicture}
    \begin{scope}[rotate=90] % Rotate the entire scope by 90 degrees
        \fill[green!20] (0, 0) rectangle (1,2); % Rectangle 
        \draw[-stealth,blue,line width = 1.6pt] (0,1) -- (0,0);
      %   \draw[cyan,dotted, line width = 1pt] (0,1) -- (0.5,1.5); 
        \draw[-stealth,blue,line width = 1.6pt] (0,2) -- (0,1); 
        \draw[-stealth,black] (0,1) -- (1,1); 
       	\node[ line width=0.2pt, dashed, draw opacity=0.5] (a) at (-0.4,0.5){$y$};
    		\node[ line width=0.2pt, dashed, draw opacity=0.5] (a) at (-0.4,1.5){$x$};
            \node[ line width=0.2pt, dashed, draw opacity=0.5] (a) at (0.5,0.8){$h$};
             %\node[ line width=0.2pt, dashed, draw opacity=0.5] (a) at (0.7,1.4){$s$};
    \end{scope}
\end{tikzpicture}
 \end{aligned} & \quad \quad \Av^{z\otimes w}_{v_s}|x,y,h\rangle=\sum_{[z]} | w x, yz^{[0]}, hz^{[-1]}\rangle. \label{eq:Avs2}
\end{align}
Now we define the boundary operator as 
\begin{equation}
    \Av^K_{v_s} = \Av^{\zeta_K}_{v_s} = \sum_{\langle \zeta_K \rangle} \Av^{\zeta_K^{\langle 1 \rangle} \otimes \zeta_K^{\langle 2 \rangle}}_{v_s}
\end{equation}
where \(\zeta_K\) is the symmetric separability idempotent of \( K \). 
The Hamiltonian for the symmetry boundary is
\begin{equation}
    \Hbb_{\rm sym} = -\sum_{v_s: \rm sym} \Av_{v_s}^K.
\end{equation}

\begin{remark}
When \( K \) is a weak Kac subalgebra of \( H \), it is also a left \( H \)-comodule algebra. In this case, the symmetric separability idempotent of \( K \) is given by
\[
\zeta_K = \sum_{(\lambda_K)} \lambda_K^{(1)} \otimes S(\lambda_K^{(2)}) = \sum_{(\lambda_K)} \lambda_K^{(2)} \otimes S(\lambda_K^{(1)}) = \sum_{(\lambda_K)} S^{-1}(\lambda_K^{(2)}) \otimes \lambda_K^{(1)} = \sum_{(\lambda_K)} S^{-1}(\lambda_K^{(1)}) \otimes \lambda_K^{(2)}
\]
where \( \lambda_K \) denotes the Haar integral of \( K \). Here, we have used the properties \( S(\lambda_K) = S^{-1}(\lambda_K) = \lambda_K \), the cocommutativity of \( \lambda_K \), and the fact that the antipode \( S: K \to K \) acts as an anti-homomorphism on the coalgebra.
In this case, the boundary vertex operator in Eq.~\eqref{eq:Avs2} becomes:
\begin{equation}
    \begin{aligned}
     \begin{tikzpicture}
    \begin{scope}[rotate=90] % Rotate the entire scope by 90 degrees
        \fill[green!20] (0, 0) rectangle (1,2); % Rectangle 
        \draw[-stealth,blue,line width = 1.6pt] (0,1) -- (0,0);
       %  \draw[cyan,dotted, line width = 1pt] (0,1) -- (0.5,1.5); 
        \draw[-stealth,blue,line width = 1.6pt] (0,2) -- (0,1); 
        \draw[-stealth,black] (0,1) -- (1,1); 
       	\node[ line width=0.2pt, dashed, draw opacity=0.5] (a) at (-0.4,0.5){$y$};
    		\node[ line width=0.2pt, dashed, draw opacity=0.5] (a) at (-0.4,1.5){$x$};
            \node[ line width=0.2pt, dashed, draw opacity=0.5] (a) at (0.5,0.8){$h$};
           %  \node[ line width=0.2pt, dashed, draw opacity=0.5] (a) at (0.7,1.4){$s$};
    \end{scope}
\end{tikzpicture} \end{aligned} 
 \quad 
 \begin{aligned}
    & \Av^K_{v_s}=\XR_{\lambda_K^{\cone}}\otimes \XL_{\lambda_K^{\ctwo}}\otimes \XL_{\lambda_K^{\cthree}},\\
    & \Av^K_{v_s}|x,y,h\rangle= |\lambda_K^{\cone}x, yS^{-1}(\lambda_K^{\ctwo}), hS^{-1}(\lambda_K^{\cthree})\rangle.
 \end{aligned}
\end{equation} 
The cocommutativity of \(\lambda\) implies that the operator \(\Av^K_{v_s}\) is independent of the specific starting link \(s = (v_s, f)\) chosen at the vertex \(v_s\). Consequently, \(\Av^K_{v_s}\) depends solely on the vertex \(v_s\) itself.
\end{remark}

For the physical boundary as drawn in Eq.~\eqref{eq:ClusterLadder}, the bulk is on the right-hand side of the boundary, and \( J \) is a right \( H \)-comodule algebra. We define the following operator for \( z \otimes w \in J \otimes J^{\rm op} \):
\begin{align}
       \begin{aligned}
\begin{tikzpicture}
    \begin{scope}[rotate=-90] % Rotate the entire scope by -90 degrees
        \fill[green!20] (0, 0) rectangle (1,2); % Rectangle 
        \draw[-stealth,red,line width = 1.6pt] (0,0) -- (0,1);
      %  \draw[cyan,dotted, line width = 1pt] (0,1) -- (0.5,1.5); 
        \draw[-stealth,red, line width = 1.6pt] (0,1) -- (0,2); 
        \draw[-stealth,black] (1,1) -- (0,1); 
        \node[ line width=0.2pt, dashed, draw opacity=0.5] (a) at (-0.4,0.5){$x$};
        \node[ line width=0.2pt, dashed, draw opacity=0.5] (a) at (-0.4,1.5){$y$};
       % \node[ line width=0.2pt, dashed, draw opacity=0.5] (a) at (0.7,1.4){$s_b$};
        \node[ line width=0.2pt, dashed, draw opacity=0.5] (a) at (0.5,0.75){$h$};
    \end{scope}
\end{tikzpicture}
\end{aligned} & \quad \quad \Av^{z\otimes w}_{v_p}|x,h,y\rangle=\sum_{[z]} | w^{[0]}x, w^{[1]}h, y z\rangle.  \label{eq:bdd-sta-A2}
\end{align}
The boundary operator is defined similarly as that for symmetry boundary 
\begin{equation}
    \Av^J_{v_p} = \Av^{\zeta_J}_{v_p} = \sum_{\langle \zeta_J \rangle} \Av^{\zeta_J^{\langle 1 \rangle} \otimes \zeta_J^{\langle 2 \rangle}}_{v_p}
\end{equation}
where \(\zeta_J\) is the symmetric separability idempotent of \( J \). 
The Hamiltonian for the symmetry boundary is
\begin{equation}
    \Hbb_{\rm phys} = -\sum_{v_p:\rm phys} \Av_{v_s}^J.
\end{equation}

\begin{remark}
When \( J \) is chosen as a weak Kac subalgebra of \( H \), it is also a right \( H \)-comodule algebra. In this case, the symmetric separability idempotent of \( J \) is given by
$\zeta_J = \sum_{(\lambda_J)} \lambda_J^{(1)} \otimes S(\lambda_J^{(2)})$ with \( \lambda_J \) denotes the Haar integral of \( J \).
In this case, the boundary vertex operator becomes:
\begin{equation}
    \begin{aligned}
    \begin{tikzpicture}
    \begin{scope}[rotate=-90] % Rotate the entire scope by -90 degrees
        \fill[green!20] (0, 0) rectangle (1,2); % Rectangle 
        \draw[-stealth,red,line width = 1.6pt] (0,0) -- (0,1);
      %  \draw[cyan,dotted, line width = 1pt] (0,1) -- (0.5,1.5); 
        \draw[-stealth,red, line width = 1.6pt] (0,1) -- (0,2); 
        \draw[-stealth,black] (1,1) -- (0,1); 
        \node[ line width=0.2pt, dashed, draw opacity=0.5] (a) at (-0.4,0.5){$x$};
        \node[ line width=0.2pt, dashed, draw opacity=0.5] (a) at (-0.4,1.5){$y$};
       % \node[ line width=0.2pt, dashed, draw opacity=0.5] (a) at (0.7,1.4){$s_b$};
        \node[ line width=0.2pt, dashed, draw opacity=0.5] (a) at (0.5,0.75){$h$};
    \end{scope}
\end{tikzpicture} \end{aligned} 
 \quad 
 \begin{aligned}
    & \Av^J_{v_p}=\XR_{\lambda_J^{\cone}}\otimes \XR_{\lambda_J^{\ctwo}}\otimes \XL_{\lambda_J^{\cthree}},\\
    & \Av^J_{v_p}|x,h,y\rangle= |\lambda_J^{\cone}x, \lambda_J^{\ctwo} h, yS^{-1}(\lambda_J^{\cthree})\rangle.
 \end{aligned}
\end{equation} 
The cocommutativity of \(\lambda\) implies that the operator \(\Av^J_{v_p}\) is independent of the specific starting link \(s = (v_p, f)\) chosen at the vertex \(v_s\). Consequently, \(\Av^J_{v_p}\) depends solely on the vertex \(v_s\) itself.
\end{remark}

The bulk stabilizers are face operators \( \Bf_f^{H} \), whose expression is the same as that in Eq.~\eqref{eq:ABdef} but with a different interpretation. For the symmetry boundary \( K \), we need to introduce
\[
\ZL^K_{\psi} | y \rangle = \sum_{[y]} \psi(S(y^{[-1]})) |y^{[0]}\rangle , \quad y\in K,
\]
where there is no direct pairing between \( K \) and \( \hat{H} \); instead, we use the pairing between \( H \) and \( \hat{H} \), with \( y^{[-1]} \in H \).
For the physical boundary, we similarly introduce \( \ZR^J_{\psi} \) as
\[
\ZR^J_{\psi} | x \rangle = \sum_{[x]} \psi(x^{[1]}) | x^{[0]} \rangle, \quad x\in J.
\]
Using these two edge operators, we can define face operator as 
\begin{equation}
\begin{aligned}
    \begin{tikzpicture}
    \begin{scope}[rotate=-90] % Rotate the entire scope by -90 degrees
        \fill[green!20] (0, 0) rectangle (1,1); % Rectangle 
        \draw[-stealth,red,line width = 1.6pt] (0,0) -- (0,1);
       \draw[cyan,dotted, line width = 1pt] (0,0) -- (0.5,0.5); 
        %\draw[-stealth,red, line width = 1.6pt] (0,1) -- (0,2); 
        \draw[-stealth,blue,line width = 1.6pt] (1,0) -- (1,1);
        \draw[-stealth,black] (1,1) -- (0,1); 
                \draw[-stealth,black] (1,0) -- (0,0); 
        \node[ line width=0.2pt, dashed, draw opacity=0.5] (a) at (-0.3,0.5){$x$};
        \node[ line width=0.2pt, dashed, draw opacity=0.5] (a) at (1.4,0.5){$y$};
       % \node[ line width=0.2pt, dashed, draw opacity=0.5] (a) at (0.7,1.4){$s_b$};
        \node[ line width=0.2pt, dashed, draw opacity=0.5] (a) at (0.5,1.25){$h$};
            \node[ line width=0.2pt, dashed, draw opacity=0.5] (a) at (0.5,-.25){$g$};
    \end{scope}
\end{tikzpicture}   
\end{aligned}
\quad
\begin{aligned}
   & \Bf_f^{\psi}=\ZR_{\psi^{\cone}}\otimes \ZL^K_{\psi^{\ctwo}} \otimes \ZL_{\psi^{\cthree}} \otimes \ZR_{\psi^{\cfour}}^J\\
   &   \Bf_f^{\psi} |g,y,h,x\rangle= \sum \psi(g^{\ctwo} S(y^{[-1]}) S(h^{(1)}) x^{[1]}) |g^{\cone},y^{[0]},h^{\ctwo},x^{[0]}\rangle
\end{aligned}
\end{equation}
By setting $\psi$ as Haar measure $\Lambda\in \hat{H}$, we obtain face stabilizer operator $\Bf_f^H=\Bf_f^{\Lambda}$, the corresponding bulk Hamiltonian reads
\begin{equation}
    \Hbb_{\rm bk} =-\sum_f \Bf_f^{H}.
\end{equation}

The total weak Hopf cluster ladder Hamiltonian is of the form
\begin{equation}\label{eq:clusterLadderHam}
    \mathbb{H}[H,K,J]=\mathbb{H}_{\rm bk}+\mathbb{H}_{\rm sym} + \mathbb{H}_{\rm phys} 
\end{equation}
which realize the SymTFT for a given weak Hopf gauge symmetry $H$.

The bulk phase is given by \( \Rep(D(H)) \simeq \mathcal{Z}(\Rep(H)) \) as we have discussed before. The symmetry boundary is equivalently characterized by the category \( {_{K}\Mod} \) of all left \( K \)-modules, and it can be checked that this forms a left module category over \( \Rep(H) \).
The symmetry is thus given by the category of endofunctor category
\begin{equation}
    \eS =\eB_{\rm sym} \simeq \Fun_{\Rep(H)}({_{K}\Mod}, {_{K}\Mod}). 
\end{equation}
For the physical boundary, since \( J \) is a right \( H \)-comodule algebra, the boundary is equivalently characterized by \( \Mod_{J} \), which is a right module category over \( \Rep(H) \). The physical boundary excitation is given by  the category of endofunctor category
\begin{equation}
    \eB_{\rm phys} = \Fun_{\Rep(H)}(\Mod_{J}, \Mod_{J}).
\end{equation}

\begin{remark}[String-net realization]
  From a macroscopic point of view, the bulk of the ladder is characterized by the multifusion category \( \Rep(H) \). The symmetry boundary is characterized by the category of left \( K \)-modules \( \eM_K = {_K}\mathsf{Mod} \), which is a left module category over \( \Rep(H) \). Similarly, the physical boundary is given by the module category \( \eM_J = \Mod_J \) over \( \Rep(H) \). The category \( \eC = \Rep(H) \) can be treated as the bulk input data, while \( \eM_K \) and \( \eM_J \) serve as two boundary input data. We can also construct a string-net lattice model that realizes the corresponding cluster ladder model.
\end{remark}

\subsection{Ground state degeneracy}

On a ladder lattice with periodic boundary condition, the ground state degeneracy can be calculated using a macroscopic argument. 
The bulk topological phase is given by the Drinfeld center $\mathsf{Bulk}(H)=\mathcal{Z}(\Rep(H))$ of $\Rep(H)$, which is braided monoidal equivalent to the representation category of charge symmetry, Drinfeld double $D(H)$ of gauge symmetry $H$.
The bulk phase is a unitary modular tensor category, from the anyon condensation point of view, the topological boundary is determined by Lagrangian algebras in $\mathsf{Bulk}(H)$.
For each boundary labeled by $H$ comodule algebra $K$, there is a corresponding Lagrangian algebra $\mathcal{A}_K$ in $\mathsf{Bulk}(H)$.

Before discussing the ground state degeneracy of the weak Hopf ladder model, let us first consider a more general case. Take a closed sphere $\mathbb{S}^2$ and punch $n$ holes in it. For each hole, assign a topological boundary condition denoted by $K_i$, where $i = 1, \ldots, n$.
For example, when $n=3$ we have the following manifold:  
\begin{equation*}
\begin{tikzpicture}[
  tqft,
  every outgoing boundary component/.style={fill=blue!50},
  outgoing boundary component 3/.style={fill=none,draw=black},
  every incoming boundary component/.style={fill=yellow!10},
  every lower boundary component/.style={draw,black},
  every upper boundary component/.style={draw,black},
  cobordism/.style={fill=gray!40},
  cobordism edge/.style={draw,black},
  genus=0,
  hole 2/.style={draw},
  view from=incoming,
  anchor=between incoming 1 and 2
]
\pic[name=a,tqft,
    incoming boundary components=3,
  %  skip incoming boundary components={2,4},
    outgoing boundary components=0,
    %skip outgoing boundary components={2,3,5},
    offset=-.5];
\end{tikzpicture}
\end{equation*}
Then for each boundary there is a Lagrangian algebra $A_{K_i}$, which can be regarded as a charge in the bulk. The sphere with $n$ gapped holes is equivalent to a sphere with $n$ charges. In this case the ground state degeneracy is just the number of independent fusion channel between these charges and the vacuum charge $\mathbb{1}$ of the bulk phase:
\begin{equation}\label{eq:GSD-Lag}
  \mathrm{GSD}= \operatorname{dim}  \operatorname{Hom}(\mathbb{1},\mathcal{A}_{K_1}\otimes \cdots \otimes \mathcal{A}_{K_n}).
\end{equation}

If we replace the sphere with a more complicated $g$-genus surface and punch $n$ holes, the ground state degeneracy can also be derived from the Moore-Seiberg formula \cite{Moore1989,ritzzwilling2023topological}.
Consider an orientable $g$-genus surface $\mathbb{M}_{g,n}$ with $n$ holes, the anyons for there holes are still $\cA_{K_1},\cdots,\cA_{K_k}$. 
This surface can be decomposed into surfaces of pants, which are then glued together. It is important to note that this decomposition is generally not unique, but the associativity of anyon fusion ensures that the final ground state degeneracy remains consistent:
\begin{equation}\label{eq:GSD-TQFT}
    \operatorname{GSD}[\mathbb{M}_{g,n},\cA_{K_1},\cdots,\cA_{K_k}]=\sum_{Y\in \Irr(\Rep(D(H)))}\left(\prod_{j=1}^k S_{\cA_{K_j},Y}\right)S_{\mathbb{1},Y}^{2-2g-n},
\end{equation}
where $S$ is the S-matrix for the bulk phase.

The weak Hopf ladder model corresponds to the case where the manifold is a cylinder, which can be thought of as a sphere with two punctures. In this model, the boundary conditions are labeled as follows: the symmetry boundary is associated with $\cA_K$, and the physical boundary is associated with $\cA_J$. By applying the Eq.~\eqref{eq:GSD-Lag} or Eq.~\eqref{eq:GSD-TQFT}, one can determine the ground state degeneracy for this configuration.

\subsection{Solving the model via tensor network}

The cluster ladder model can be solved using the weak Hopf tensor network introduced in Section~\ref{sec:tensor-network}.
There will be four types of local tensors we need to construct the cluster ladder ground state based on a tensor network construction:

(1) For symmetry boundary edge, the corresponding local tensor is the comultiplication of Haar integral $\lambda_K$ of weak Kac subalgebra $K\leq H$:
\begin{equation}
     \Delta(\lambda_K)= \begin{aligned}			
\begin{tikzpicture}
    % Draw the central blue box with \lambda
    \node[draw, fill=blue!20, minimum width=0.6cm, minimum height=0.6cm] (center) at (0,0) {$\scriptstyle\lambda_K$};
    % Draw the edges with labels a and b
    \draw[line width=1.0pt] (center.north) -- ++(0,.5) node[above] {$\scriptstyle\lambda_K^{\ctwo}$};
    \draw[line width=1.0pt ] (center.south) -- ++(0,-.5) node[below] {$\scriptstyle \lambda_K^{\cone}$};
\end{tikzpicture}
\end{aligned} :=  \begin{aligned}
				\begin{tikzpicture}
					\draw[black, line width=1.0pt]  (-0.5, 1) .. controls (-0.4, 0) and (0.4, 0) .. (0.5, 1);
					\draw[black, line width=1.0pt]  (0,0.23)--(0,-0.23);
	 % \filldraw[green!30, draw=black] 
  %       (-0.3, 0) -- (0.3, 0) -- (0, -0.4) -- cycle;
        \filldraw[green!30, draw=black] 
       (0,-0.2) circle (0.2);
    % Place the label at the centroid of the triangle
    \node at (0, -0.2) {\textcolor{black}{$\scriptstyle \lambda_K$}};
                   \node at (-0.5, 1.3) {$\scriptstyle \lambda_K^{\cone}$};
	                   \node at (0.5, 1.3) {$\scriptstyle \lambda_K^{\ctwo}$};			
    \end{tikzpicture}
\end{aligned}.
\end{equation}
The leg of the tensor that will be contracted in the tensor network  is \( \lambda_K^{\ctwo} \).

(2) For physical boundary edge, the corresponding local tensor is the comultiplication of Haar integral $\lambda_J$ of weak Kac subalgebra $J\leq H$:
\begin{equation}
     \Delta(\lambda_J)= \begin{aligned}			
\begin{tikzpicture}
    % Draw the central blue box with \lambda
    \node[draw, fill=cyan!20, minimum width=0.6cm, minimum height=0.6cm] (center) at (0,0) {$\scriptstyle\lambda_J$};
    % Draw the edges with labels a and b
    \draw[line width=1.0pt] (center.north) -- ++(0,.5) node[above] {$\scriptstyle\lambda_J^{\cone}$};
    \draw[line width=1.0pt ] (center.south) -- ++(0,-.5) node[below] {$\scriptstyle \lambda_J^{\ctwo}$};
\end{tikzpicture}
\end{aligned} :=  \begin{aligned}
				\begin{tikzpicture}
					\draw[black, line width=1.0pt]  (-0.5, 1) .. controls (-0.4, 0) and (0.4, 0) .. (0.5, 1);
					\draw[black, line width=1.0pt]  (0,0.23)--(0,-0.23);
	 % \filldraw[green!30, draw=black] 
  %       (-0.3, 0) -- (0.3, 0) -- (0, -0.4) -- cycle;
        \filldraw[green!30, draw=black] 
       (0,-0.2) circle (0.2);
    % Place the label at the centroid of the triangle
    \node at (0, -0.2) {\textcolor{black}{$\scriptstyle \lambda_J$}};
                   \node at (-0.5, 1.3) {$\scriptstyle \lambda_J^{\cone}$};
	                   \node at (0.5, 1.3) {$\scriptstyle \lambda_J^{\ctwo}$};			
    \end{tikzpicture}
\end{aligned}.
\end{equation}
Notice that we have put \( \lambda_J^{\cone} \) at the top; the leg of the tensor that will be contracted in the tensor network is \( \lambda_J^{\ctwo} \).

(3) For the bulk edge, the local tensor corresponds to $\Delta_2(\lambda)$:
\begin{equation}
     \Delta_2(\lambda)= \begin{aligned}			
\begin{tikzpicture}
    \draw[line width=1.0pt] (center.north) (0,0) -- ++(0.8,0) node[right] {$\scriptstyle\lambda^{\ctwo}$};
    \draw[line width=1.0pt] (center.north) (0,0) -- ++(-0.8,0) node[left] {$\scriptstyle\lambda^{\cthree}$};
    \draw[line width=1.0pt ] (center.south) -- ++(0,-.5) node[below] {$\scriptstyle \lambda^{\cone}$};
    \node[draw, fill=gray!20, minimum width=0.6cm, minimum height=0.6cm] (center) at (0,0) {$\scriptstyle\lambda$};
\end{tikzpicture}
\end{aligned} :=  \begin{aligned}
				\begin{tikzpicture}
					\draw[black, line width=1.0pt]  (-0.5, 1) .. controls (-0.4, 0) and (0.4, 0) .. (0.5, 1);
					\draw[black, line width=1.0pt]  (0,1)--(0,-0.23);
				 	 % \filldraw[green!30, draw=black] 
       %  (-0.3, 0) -- (0.3, 0) -- (0, -0.4) -- cycle;
             \filldraw[green!30, draw=black] 
       (0,-0.2) circle (0.2);
    % Place the label at the centroid of the triangle
    \node at (0, -0.2) {\textcolor{black}{$\scriptstyle \lambda$}};
                   \node at (-0.7, 1.3) {$\scriptstyle \lambda^{\cone}$};
	                   \node at (0.7, 1.3) {$\scriptstyle \lambda^{\cthree}$};	
                     \node at (0, 1.3) {$\scriptstyle \lambda^{\ctwo}$};	
    \end{tikzpicture}
\end{aligned}.
\end{equation}
The legs of the tensor that will be contracted in the tensor network  are \( \lambda^{\ctwo}, \lambda^{\cthree} \).

(4) For face, a local tensor corresponds to the Haar measure $\Lambda\in \hat{H}$ is needed to glue the edge tensors together:
\begin{equation}
    ( \id \otimes \id\otimes \hat{S} \otimes \hat{S}) \circ \hat{\Delta}(\Lambda)=
    \begin{aligned}			
\begin{tikzpicture}
\draw[line width=1.0pt] (center.north) (0,0) -- ++(0.8,0) node[right] {$\scriptstyle\Lambda^{\ctwo}$};
\draw[line width=1.0pt] (center.north) (0,0) -- ++(0,0.8) node[above] {$\scriptstyle \hat{S}(\Lambda^{\cthree})$};
\draw[line width=1.0pt] (center.north) (0,0) -- ++(-0.8,0) node[left] {$\scriptstyle \hat{S}(\Lambda^{\cfour})$};
\draw[line width=1.0pt ] (center.south) -- ++(0,-.5) node[below] {$\scriptstyle \Lambda^{\cone}$};
\node[draw, fill=yellow!20, minimum size=0.6cm, shape=circle] (center) at (0,0) {$\scriptstyle\Lambda$};
 \filldraw[black] (-0.5,0) circle (2pt);  % Shaded dot on the line
  \filldraw[black] (0,0.5) circle (2pt);  % Shaded dot on the line
\end{tikzpicture}
\end{aligned} 
    =
    \begin{aligned}
			\begin{tikzpicture}
				 \draw[black, line width=1.0pt]  (-0.5, 0) .. controls (-0.4, 1) and (0.4, 1) .. (0.5, 0);
                 \draw[black, line width=1.0pt]  (-1.3, 0) .. controls (-1.2, 1) and (1.2, 1) .. (1.3, 0);
				 \draw[black, line width=1.0pt]  (0,0.75)--(0,1);
			 % \filldraw[red!30, draw=black] 
    %     (-0.3, 1) -- (0.3, 1) -- (0, 1.4) -- cycle;
             \filldraw[red!30, draw=black] 
        (0,1.2) circle (0.2);
         \node at (0, 1.2) {$\scriptstyle \Lambda$};
          \node at (-0.7, -.3) {$\scriptstyle \Lambda^{\ctwo}$};
           \node at (-1.4, -.3) {$\scriptstyle \Lambda^{\cone}$};
              \node at (0.4, -.3) {$\scriptstyle \hat{S}( \Lambda^{\cthree})$};
               \node at (1.5, -.3) {$\scriptstyle \hat{S}(\Lambda^{(4)})$};
               \filldraw[black] (0.4,0.375) circle (2pt);  % Shaded dot on the line
             \filldraw[black] (1.12,0.375) circle (2pt);  % Shaded dot on the line

				\end{tikzpicture}
			\end{aligned}.
\end{equation}
The antipode \(\hat{S}\) is introduced because the direction of the lattice is chosen as in Eq.~\eqref{eq:ClusterLattice}. If a different lattice configuration is chosen, the antipode may be placed on different legs of the tensor.

\begin{theorem}
The cluster ladder state $|\Psi\rangle_{\rm cluster}$ is defined as the following tensor network state:
\begin{equation}
    \begin{aligned}			
\begin{tikzpicture}
    \draw[line width=1.0pt] (-0.5,0) -- ++(11.3,0) ;
        %\draw[dotted,green, line width=1.0pt] (-0.5,-1) -- ++(11.3,0) ;
          %  \draw[dotted,green, line width=1.0pt] (-0.5,1) -- ++(11.3,0) ;
    \draw[line width=1.0pt, red] (0,0) -- ++(0,-.8);
        \draw[line width=1.0pt, red] (2,0) -- ++(0,-.8);
    \draw[line width=1.0pt, red] (4,0) -- ++(0,-.8);
    \draw[line width=1.0pt, red] (6,0) -- ++(0,-.8);
    \draw[line width=1.0pt, red] (8,0) -- ++(0,-.8);
    \draw[line width=1.0pt, red] (10,0) -- ++(0,-.8);
    %%%%%
     \filldraw[black] (0.5,0) circle (2pt);  % Shaded dot on the line
      \filldraw[black] (2.5,0) circle (2pt);
       \filldraw[black] (4.5,0) circle (2pt);
        \filldraw[black] (6.5,0) circle (2pt);
         \filldraw[black] (8.5,0) circle (2pt);
          \filldraw[black] (10.5,0) circle (2pt);
    %%%%%%
  \filldraw[black] (1,0.5) circle (2pt);  % Shaded dot on the line
  \filldraw[black] (3,0.5) circle (2pt);  % Shaded dot on the line
  \filldraw[black] (5,0.5) circle (2pt);  % Shaded dot on the line
  \filldraw[black] (7,0.5) circle (2pt);  % Shaded dot on the line
  \filldraw[black] (9,0.5) circle (2pt);  % Shaded dot on the line
    %%%%%%%%
    \draw[line width=1.0pt] (1,-1) -- ++(0,1) ;
    \draw[line width=1.0pt, red] (1,-1) -- ++(0,-.8);
    \node[draw, fill=blue!20, minimum width=0.6cm, minimum height=0.6cm] (center) at (1,-1) {$\scriptstyle\lambda_K$};
     %%%%
     \draw[line width=1.0pt] (3,-1) -- ++(0,1) ;
    \draw[line width=1.0pt, red] (3,-1) -- ++(0,-.8);
    \node[draw, fill=blue!20, minimum width=0.6cm, minimum height=0.6cm] (center) at (3,-1) {$\scriptstyle\lambda_K$};
    %%%%%%
         \draw[line width=1.0pt] (5,-1) -- ++(0,1) ;
    \draw[line width=1.0pt, red] (5,-1) -- ++(0,-.8);
    \node[draw, fill=blue!20, minimum width=0.6cm, minimum height=0.6cm] (center) at (5,-1) {$\scriptstyle\lambda_K$};
    %%%%%%
         \draw[line width=1.0pt] (7,-1) -- ++(0,1) ;
    \draw[line width=1.0pt, red] (7,-1) -- ++(0,-.8);
    \node[draw, fill=blue!20, minimum width=0.6cm, minimum height=0.6cm] (center) at (7,-1) {$\scriptstyle\lambda_K$};
    %%%%%%
         \draw[line width=1.0pt] (9,-1) -- ++(0,1) ;
    \draw[line width=1.0pt, red] (9,-1) -- ++(0,-.8);
    \node[draw, fill=blue!20, minimum width=0.6cm, minimum height=0.6cm] (center) at (9,-1) {$\scriptstyle\lambda_K$};
    %%%%%%
    %%%%%%
    \draw[line width=1.0pt] (1,1) -- ++(0,-1) ;
    \draw[line width=1.0pt, red] (1,1) -- ++(0,.8);
    \node[draw, fill=cyan!20, minimum width=0.6cm, minimum height=0.6cm] (center) at (1,1) {$\scriptstyle\lambda_J$};
        %%%%%%
    \draw[line width=1.0pt] (3,1) -- ++(0,-1) ;
    \draw[line width=1.0pt, red] (3,1) -- ++(0,.8);
    \node[draw, fill=cyan!20, minimum width=0.6cm, minimum height=0.6cm] (center) at (3,1) {$\scriptstyle\lambda_J$};
            %%%%%%
    \draw[line width=1.0pt] (5,1) -- ++(0,-1) ;
    \draw[line width=1.0pt, red] (5,1) -- ++(0,.8);
    \node[draw, fill=cyan!20, minimum width=0.6cm, minimum height=0.6cm] (center) at (5,1) {$\scriptstyle\lambda_J$};
            %%%%%%
    \draw[line width=1.0pt] (7,1) -- ++(0,-1) ;
    \draw[line width=1.0pt, red] (7,1) -- ++(0,.8);
    \node[draw, fill=cyan!20, minimum width=0.6cm, minimum height=0.6cm] (center) at (7,1) {$\scriptstyle\lambda_J$};
            %%%%%%
    \draw[line width=1.0pt] (9,1) -- ++(0,-1) ;
    \draw[line width=1.0pt, red] (9,1) -- ++(0,.8);
    \node[draw, fill=cyan!20, minimum width=0.6cm, minimum height=0.6cm] (center) at (9,1) {$\scriptstyle\lambda_J$};
    
    \node[draw, fill=gray!20, minimum width=0.6cm, minimum height=0.6cm] (center) at (0,0) {$\scriptstyle\lambda$};
    \node[draw, fill=yellow!20, minimum size=0.6cm, shape=circle] (center) at (1,0) {$\scriptstyle\Lambda$};
     \node[draw, fill=gray!20, minimum width=0.6cm, minimum height=0.6cm] (center) at (2,0) {$\scriptstyle\lambda$};
      \node[draw, fill=yellow!20, minimum size=0.6cm, shape=circle] (center) at (3,0) {$\scriptstyle\Lambda$};
        \node[draw, fill=gray!20, minimum width=0.6cm, minimum height=0.6cm] (center) at (4,0) {$\scriptstyle\lambda$};
      \node[draw, fill=yellow!20, minimum size=0.6cm, shape=circle] (center) at (5,0) {$\scriptstyle\Lambda$};
        \node[draw, fill=gray!20, minimum width=0.6cm, minimum height=0.6cm] (center) at (6,0) {$\scriptstyle\lambda$};
      \node[draw, fill=yellow!20, minimum size=0.6cm, shape=circle] (center) at (7,0) {$\scriptstyle\Lambda$};
        \node[draw, fill=gray!20, minimum width=0.6cm, minimum height=0.6cm] (center) at (8,0) {$\scriptstyle\lambda$};
      \node[draw, fill=yellow!20, minimum size=0.6cm, shape=circle] (center) at (9,0) {$\scriptstyle\Lambda$};
        \node[draw, fill=gray!20, minimum width=0.6cm, minimum height=0.6cm] (center) at (10,0) {$\scriptstyle\lambda$};
\end{tikzpicture}
\end{aligned} 
\end{equation}
where we have drawn the free leg (corresponding to physical degrees of freedom) in red.
The cluster state $|\Psi\rangle_{\rm cluster}$  is the ground state of the cluster ladder model defined in Eq.~\eqref{eq:clusterLadderHam}, namely, 
\begin{equation}
    \Av_{v_s}^{K}|\Psi\rangle_{\rm cluster} =   \Av_{v_p}^{J}|\Psi\rangle_{\rm cluster} =|\Psi\rangle_{\rm cluster}= \Bf_{f}^{H}|\Psi\rangle_{\rm cluster},
\end{equation}
for all $v_s,v_p,f$.
\end{theorem}

\begin{proof}
This can be proved using the techniques developed in Ref.~\cite[Section 7]{Jia2023weak} and the property of pairing, $\lambda^2=\lambda$, $\Lambda^2=\Lambda$, along with the cyclic property of the comultiplication of the Haar integral and Haar measure.
\end{proof}

\section{Examples of weak Hopf cluster ladder model}
\label{sec:exampleCluster}
In this section, let us examine some examples of our lattice model.
An example for Haagerup fusion category symmetry is given in Ref.~\cite{jia2024quantumclusterstatemodel} and a the model with Ising fusion category symmetry is given in Ref.~\cite{jia2025Ising}. Here we provide more examples.

\subsection{$H_8$ model}

The first example we consider is the Kac-Paljutkin algebra $H_8$ \cite{kac1966finite}, which is generated by three elements $x, y, z$ with the constraints 
\begin{gather*}
        x^2 = y^2 = 1, \quad z^2 = \frac{1}{2}(1+x+y-xy), \\
        xy = yx,\quad zx=yz,\quad zy = xz. 
\end{gather*}
It's clear that $\dim H_8 =8$ and the basis is $\{1,x,y,xy,z,zx,zy,zxy\}$. 
The coalgebra structure and the antipode are determined by 
\begin{gather*}
    \Delta(x) = x\otimes x,\quad \Delta(y) = y\otimes y, \\
    \Delta(z) = \frac{1}{2}(1\otimes 1+y\otimes 1 + 1\otimes x - y\otimes x)(z\otimes z), \\
    \varepsilon(x) = \varepsilon(y) = \varepsilon(z) = 1, \\
    S(x) = x,\quad S(y) = y, \quad S(z) = z,
\end{gather*}
with linear extension. 
$H_8$ is a $C^*$ Hopf algebra with $*$-operation given by
\begin{equation}
    x^*=x,\quad y^*=y,\quad z^*=z^{-1}=z^3=\frac{1}{2}(z+zx + zy -zxy).
\end{equation}
The Haar integral is given by 
\begin{equation}
    \lambda =\frac{1}{8}(1 + x + y + xy + z + zx + zy + zxy).
\end{equation}
And Haar measure (Haar integral of $\hat{H}_8$) is given by
\begin{equation}
    \Lambda =\delta_1(\bullet).
\end{equation}
Since $H_8$ is neither commutative nor cocommutative, it can not be represented as a group algebra.

The quantum group $H_8$ has four invertible one-dimensional irreducible representations $\Gamma_{\one}, \Gamma_{a}, \Gamma_{b}, \Gamma_{ab}$ and one two-dimensional irreducible representation $\Gamma_{\sigma}$.  
Since $\Rep(H_8)$ is equivalent to the Tambara--Yamagami fusion category $\mathsf{TY}(\Zbb_2\times \Zbb_2, \chi_{\rm diag}, \epsilon=+1)$ \cite{tambara1998tensor}, they share the same fusion rules:
\begin{equation}
\begin{aligned}
     &\Gamma_{a}\otimes \Gamma_{b} = \Gamma_{ab}, \qquad 
     \Gamma_{\sigma} \otimes \Gamma_{\sigma} = \Gamma_{\one} \oplus \Gamma_a \oplus \Gamma_b \oplus \Gamma_{ab},\\
     &\Gamma_a \otimes \Gamma_{\sigma} = \Gamma_b \otimes \Gamma_{\sigma} = \Gamma_{ab} \otimes \Gamma_{\sigma} = \Gamma_{\sigma},
\end{aligned}
\end{equation}
where only the nontrivial fusion rules are written explicitly. We also note that the fusion in $\Rep(H_8)$ is commutative.

Consider a ladder lattice (Eq.~\ref{eq:ClusterLattice}), on smooth boundary, the vertex operator is given by 
\begin{equation}
    \mathbb{A}_v^{\lambda}=\frac{1}{8} ( \mathbb{A}_v^{1}+ \mathbb{A}_v^{x}+\mathbb{A}_v^{y} + \mathbb{A}_v^{xy} +\mathbb{A}_v^{z} +\mathbb{A}_v^{zx} +\mathbb{A}_v^{zy} +\mathbb{A}_v^{zxy})
\end{equation}
where the sum runs over all basis elements. For each basis element $g$, $\Av^g= \XR_{g^{\cone}} \otimes \XL_{g^{\ctwo}} \otimes \XL_{g^{\cthree}}$ (we omit the summation for comultiplication there):
\begin{align}
       \begin{aligned}
           \begin{tikzpicture}
    \begin{scope}[rotate=90] % Rotate the entire scope by 90 degrees
        \fill[green!20] (0, 0) rectangle (1,2); % Rectangle 
        \draw[-stealth,blue,line width = 1.6pt] (0,1) -- (0,0);
      %   \draw[cyan,dotted, line width = 1pt] (0,1) -- (0.5,1.5); 
        \draw[-stealth,blue,line width = 1.6pt] (0,2) -- (0,1); 
        \draw[-stealth,black] (0,1) -- (1,1); 
       	\node[ line width=0.2pt, dashed, draw opacity=0.5] (a) at (-0.4,0.5){$b$};
    		\node[ line width=0.2pt, dashed, draw opacity=0.5] (a) at (-0.4,1.5){$a$};
            \node[ line width=0.2pt, dashed, draw opacity=0.5] (a) at (0.5,0.8){$c$};
             %\node[ line width=0.2pt, dashed, draw opacity=0.5] (a) at (0.7,1.4){$s$};
    \end{scope}
\end{tikzpicture}
 \end{aligned} & \quad \quad \Av^{g}_{v}|a,b,c\rangle=\sum_{(g)} | g^{(1)} a, b S(g^{(2)}), c S(g^{(3)})\rangle. 
\end{align}
The face operator is given by $\Bf_f^{\Lambda}=\ZR_{\Lambda^{\cone}}\otimes \ZL_{\hat{S}(\Lambda^{\ctwo} )}  \otimes \ZL_{\hat{S}(\Lambda^{\ctwo} )}$, it's action reads
\begin{equation}
\begin{aligned}
    \begin{tikzpicture}
    \begin{scope}[rotate=-90] % Rotate the entire scope by -90 degrees
        \fill[green!20] (0, 0) rectangle (1,1); % Rectangle 
       % \draw[-stealth,red,line width = 1.6pt] (0,0) -- (0,1);
     %  \draw[cyan,dotted, line width = 1pt] (0,0) -- (0.5,-0.5); 
        %\draw[-stealth,red, line width = 1.6pt] (0,1) -- (0,2); 
        \draw[-stealth,blue,line width = 1.6pt] (1,0) -- (1,1);
        \draw[-stealth,black] (1,1) -- (0,1); 
                \draw[-stealth,black] (1,0) -- (0,0); 
        %\node[ line width=0.2pt, dashed, draw opacity=0.5] (a) at (-0.3,0.5){$$};
        \node[ line width=0.2pt, dashed, draw opacity=0.5] (a) at (1.4,0.5){$s$};
       % \node[ line width=0.2pt, dashed, draw opacity=0.5] (a) at (0.7,1.4){$s_b$};
        \node[ line width=0.2pt, dashed, draw opacity=0.5] (a) at (0.5,1.25){$t$};
            \node[ line width=0.2pt, dashed, draw opacity=0.5] (a) at (0.5,-.25){$r$};
    \end{scope}
\end{tikzpicture}   
\end{aligned}
\quad
\begin{aligned}
   &   \Bf_f^{\Lambda} |r,s,t\rangle= \sum \delta_1(r^{\ctwo} S(s^{\cone}) S(t^{(1)})) |r^{\cone},s^{\ctwo},t^{\ctwo}\rangle
\end{aligned}.
\end{equation}

On a closed manifold $\mathbb{S}^1$, the corresponding cluster state model possesses the symmetry
\begin{equation}
    \Cocom(H_8) \times \Cocom(\hat{H}_8).
\end{equation}
For any $h \in \Cocom(H_8)$, the associated symmetry operator is given by
\begin{equation}
    S_h = \XR_{h^{\cone}} \otimes \XR_{h^{\ctwo}} \otimes \cdots \otimes \XR_{h^{(n)}}= 
    \begin{aligned}
\begin{tikzpicture}
    % Define the number of rungs in the ladder
    \def\n{5}
    % Define the size of each square
    \def\s{1}
        % Draw the shaded background lattice
    \fill[green!20] (0, 0) rectangle (\n*\s+\s, \s); % Rectangle covering the whole background
    % Draw the ladder with arrows in the middle of each edge
    \foreach \i in {0,...,\n} {
        % Draw solid bottom edges with arrows in the middle pointing right
        \draw[-stealth, line width=1.0pt,blue, midway] (\i*\s, 0) -- (\i*\s+\s, 0);
        % Draw dotted top edges with arrows in the middle pointing right
        \draw[dotted, line width=1.0pt,red, midway] (\i*\s, \s) -- (\i*\s+\s, \s);
        % Draw upward ladder edges with arrows in the middle pointing up
        \draw[-stealth,line width=1.0pt, midway] (\i*\s, 0) -- (\i*\s, \s);
    }
    % Draw the right-most vertical ladder edge with an arrow in the middle pointing up
    \draw[-stealth, midway,line width=1.0pt] (\n*\s+\s, 0) -- (\n*\s+\s, \s);
    \node[ line width=0.2pt, dashed, draw opacity=0.5] (a) at (.45,0.5){$\XR_{h^{\cone}}$};
    \node[ line width=0.2pt, dashed, draw opacity=0.5] (a) at (1.45,0.5){$\XR_{h^{\ctwo}}$};
    \node[ line width=0.2pt, dashed, draw opacity=0.5] (a) at (2.45,0.5){$\XR_{h^{\cthree}}$};
        \node[ line width=0.2pt, dashed, draw opacity=0.5] (a) at (4.45,0.5){$\cdots$};
     \node[ line width=0.2pt, dashed, draw opacity=0.5] (a) at (5.45,0.5){$\XR_{h^{(n)}}$};
      \draw[-stealth, white, line width=2pt, midway] (6, 0) -- (6, 1.02);
\end{tikzpicture}
\end{aligned}
\end{equation}
which acts on all vertex edges of the lattice (notice that we have assume the periodic boundary condition in the above lattice).
For $\psi\in \Cocom(\hat{H}_8)$, the symmetry operator is given by 
\begin{equation}
    T_{\psi}= \ZR_{\psi^{\cone}} \otimes \ZR_{\psi^{\ctwo}} \otimes \cdots \otimes \ZR_{\psi^{(n)}}= 
    \begin{aligned}
\begin{tikzpicture}
    % Define the number of rungs in the ladder
    \def\n{5}
    % Define the size of each square
    \def\s{1}
        % Draw the shaded background lattice
    \fill[green!20] (0, 0) rectangle (\n*\s+\s, \s); % Rectangle covering the whole background
    % Draw the ladder with arrows in the middle of each edge
    \foreach \i in {0,...,\n} {
        % Draw solid bottom edges with arrows in the middle pointing right
        \draw[-stealth, line width=1.0pt,blue, midway] (\i*\s, 0) -- (\i*\s+\s, 0);
        % Draw dotted top edges with arrows in the middle pointing right
        \draw[dotted, line width=1.0pt,red, midway] (\i*\s, \s) -- (\i*\s+\s, \s);
        % Draw upward ladder edges with arrows in the middle pointing up
        \draw[-stealth,line width=1.0pt, midway] (\i*\s, 0) -- (\i*\s, \s);
    }
    % Draw the right-most vertical ladder edge with an arrow in the middle pointing up
    \draw[-stealth, midway,line width=1.0pt] (\n*\s+\s, 0) -- (\n*\s+\s, \s);
    \node[ line width=0.2pt, dashed, draw opacity=0.5] (a) at (.5,-0.5){\color{blue}$\ZR_{\psi^{\cone}}$};
    \node[ line width=0.2pt, dashed, draw opacity=0.5] (a) at (1.5,-0.5){\color{blue}$\ZR_{\psi^{\ctwo}}$};
    \node[ line width=0.2pt, dashed, draw opacity=0.5] (a) at (2.5,-0.5){\color{blue}$\ZR_{\psi^{\cthree}}$};
    \node[ line width=0.2pt, dashed, draw opacity=0.5] (a) at (4.5,-0.5){\color{blue}$\cdots$};
     \node[ line width=0.2pt, dashed, draw opacity=0.5] (a) at (5.5,-0.5){\color{blue}$\ZR_{\psi^{(n)}}$};
      \draw[-stealth, white, line width=2pt, midway] (6, 0) -- (6, 1.02);
\end{tikzpicture}
\end{aligned}
\end{equation}
which acts on all horizontal edges.
When for $\Gamma\in \Rep(H_8)$, its character $\chi_{\Gamma} \in \Cocom(H_8)$, thus the $H_8$ cluster state model also have $\Rep(H_8)$ as a sub-symmetry.

\subsection{Model from a boundary tube algebra $\mathbf{Tube}({_{\eC}}\eM)$}

The \emph{boundary tube algebra} provides a general method for constructing weak Hopf cluster state models with non-invertible symmetries.  
Starting from a $\eC$-module category $\eM$, the corresponding boundary tube algebra $\mathbf{Tube}({_{\eC}}\eM)$ satisfies  
\begin{equation}
    \Rep\!\left(\mathbf{Tube}({_{\eC}}\eM)\right) \;\simeq\; \Fun_{\eC}(\eM,\eM)^{\rm rev},
\end{equation}
where the right-hand side is the unitary fusion category of $\eC$-module functors from $\eM$ to itself, which characterizes the boundary topological excitations.  
The superscript ``rev'' denotes the reversed tensor product, defined by $X \otimes_{\rm rev} Y := Y \otimes X$.  
This reversal arises because the comultiplication of the boundary tube algebra is conventionally taken downward.  
One can instead adopt the coopposite comultiplication to remove this reversal.  
Since different authors follow different conventions, both versions---with and without the ``rev''---appear in the literature.

If we set $\eM=\eC$, we obtain the \emph{smooth boundary} in the Levin--Wen string-net model.  
In this case, one finds
\begin{equation}
    \Rep\!\left(\mathbf{Tube}({_{\eC}}\eC)\right) 
    \;\simeq\; \Fun_{\eC}(\eC,\eC)^{\rm rev} 
    \;\simeq\; \eC.
\end{equation}
Thus, 
\(
    \mathcal{T}_{\eC} := \mathbf{Tube}({_{\eC}}\eC)
\)
is the weak Hopf symmetry underlying the symmetry boundary, and we take this as the input data for the cluster state model.  
An explicit construction for the Haagerup fusion category symmetry is given in Ref.~\cite{jia2024quantumclusterstatemodel}, and a model with Ising fusion category symmetry is presented in Ref.~\cite{jia2025Ising}.  
Here, we provide an example for the Fibonacci fusion category symmetry.

\subsubsection{Fibonacci fusion category symmetry}

Notice that the Fibonacci category is given by 
\(
\mathsf{Fib}=\{\one,\tau\},
\) 
with the (non-trivial) fusion rule
\(
\tau \otimes \tau \;=\; \one \oplus \tau.
\)   
where $\varphi=\frac{1+\sqrt{5}}{2}$ is the golden ratio.
All nonzero $F$-symbols are equal to $1$, except for
\begin{equation}
   [F^{\tau,\tau,\tau}_{\tau}]_{\one}^{\one}=\varphi^{-1}, 
\qquad 
[F^{\tau,\tau,\tau}_{\tau}]_{\tau}^{\one}
= [F^{\tau,\tau,\tau}_{\tau}]_{\one}^{\tau} = \varphi^{-1/2}, 
\qquad 
[F^{\tau,\tau,\tau}_{\tau}]_{\tau}^{\tau}=-\varphi^{-1}. 
\end{equation}
The topological spins are 
\(
\theta_{\one}=1, \theta_{\tau}=e^{4\pi i/5},
\) 
and the quantum dimensions are 
\(
d_{\one}=1, d_{\tau}=\varphi.
\)

The boundary tube algebra $\mathcal{T}_{\mathsf{Fib}}$ is $13$-dimensional. 
A basis of $\mathcal{T}_{\mathsf{Fib}}$ is given by
\begin{align*}
    &\begin{aligned}
    \begin{tikzpicture}[scale=0.5]
     \begin{scope}
            \fill[gray!20]
                (0,1.5) arc[start angle=90, end angle=270, radius=1.5] -- 
                (0,-0.5) arc[start angle=270, end angle=90, radius=0.5] -- cycle;
        \end{scope}
         \draw[dotted,line width=.7pt,black] (0,0.5)--(0,0.8);
         \draw[dotted,line width=.7pt,black] (0,0.8)--(0,1.2);
         \draw[dotted,line width=.7pt,black] (0,1.2)--(0,1.5);
         \draw[dotted,line width=.7pt,black] (0,-0.5)--(0,-0.8);
         \draw[dotted,line width=.7pt,black] (0,-0.8)--(0,-1.2);
         \draw[dotted,line width=.7pt,black] (0,-1.2)--(0,-1.5);
         \draw[dotted,line width=.7pt,red] (0,0.8) arc[start angle=90, end angle=270, radius=0.8];
        \node[ line width=0.6pt, dashed, draw opacity=0.5] (a) at (-1,0){$\scriptstyle \one$};
        \node[ line width=0.6pt, dashed, draw opacity=0.5] (a) at (-0.3,1.3){$\scriptstyle \one$};
        \node[ line width=0.6pt, dashed, draw opacity=0.5] (a) at (-0.3,-1.3){$\scriptstyle \one$};
        \end{tikzpicture}
    \end{aligned} \;\; 
    \begin{aligned}
    \begin{tikzpicture}[scale=0.5]
     \begin{scope}
            \fill[gray!20]
                (0,1.5) arc[start angle=90, end angle=270, radius=1.5] -- 
                (0,-0.5) arc[start angle=270, end angle=90, radius=0.5] -- cycle;
        \end{scope}
         \draw[line width=.7pt,black] (0,0.5)--(0,1.5);
         \draw[dotted,line width=.7pt,black] (0,-0.5)--(0,-0.8);
         \draw[dotted,line width=.7pt,black] (0,-0.8)--(0,-1.2);
         \draw[dotted,line width=.7pt,black] (0,-1.2)--(0,-1.5);
         \draw[dotted,line width=.7pt,red] (0,0.8) arc[start angle=90, end angle=270, radius=0.8];
         \node[ line width=0.6pt, dashed, draw opacity=0.5] (a) at (-1,0){$\scriptstyle \one$};
        \node[ line width=0.6pt, dashed, draw opacity=0.5] (a) at (-0.3,1.3){$\scriptstyle \tau$};
        \node[ line width=0.6pt, dashed, draw opacity=0.5] (a) at (-0.3,-1.3){$\scriptstyle \one$};
        \end{tikzpicture}
    \end{aligned} \;\; 
     \begin{aligned}
    \begin{tikzpicture}[scale=0.5]
     \begin{scope}
            \fill[gray!20]
                (0,1.5) arc[start angle=90, end angle=270, radius=1.5] -- 
                (0,-0.5) arc[start angle=270, end angle=90, radius=0.5] -- cycle;
        \end{scope}
         \draw[dotted,line width=.7pt,black] (0,0.5)--(0,1.5);
         \draw[line width=.7pt,black] (0,-0.5)--(0,-0.8);
         \draw[line width=.7pt,black] (0,-0.8)--(0,-1.2);
         \draw[line width=.7pt,black] (0,-1.2)--(0,-1.5);
         \draw[dotted,line width=.7pt,red] (0,0.8) arc[start angle=90, end angle=270, radius=0.8];
         \node[ line width=0.6pt, dashed, draw opacity=0.5] (a) at (-1,0){$\scriptstyle \one$};
        \node[ line width=0.6pt, dashed, draw opacity=0.5] (a) at (-0.3,1.3){$\scriptstyle \one$};
        \node[ line width=0.6pt, dashed, draw opacity=0.5] (a) at (-0.3,-1.3){$\scriptstyle \tau$};
        \end{tikzpicture}
    \end{aligned} \;\; 
    \begin{aligned}
    \begin{tikzpicture}[scale=0.5]
     \begin{scope}
            \fill[gray!20]
                (0,1.5) arc[start angle=90, end angle=270, radius=1.5] -- 
                (0,-0.5) arc[start angle=270, end angle=90, radius=0.5] -- cycle;
        \end{scope}
         \draw[line width=.7pt,black] (0,0.5)--(0,1.5);
         \draw[line width=.7pt,black] (0,-0.5)--(0,-0.8);
         \draw[line width=.7pt,black] (0,-0.8)--(0,-1.2);
         \draw[line width=.7pt,black] (0,-1.2)--(0,-1.5);
         \draw[dotted,line width=.7pt,red] (0,0.8) arc[start angle=90, end angle=270, radius=0.8];
         \node[ line width=0.6pt, dashed, draw opacity=0.5] (a) at (-1,0){$\scriptstyle \one$};
        \node[ line width=0.6pt, dashed, draw opacity=0.5] (a) at (-0.3,1.3){$\scriptstyle \tau$};
        \node[ line width=0.6pt, dashed, draw opacity=0.5] (a) at (-0.3,-1.3){$\scriptstyle \tau$};
        \end{tikzpicture}
    \end{aligned}
\end{align*}

%%%%%%%%%%%%%%%%%%%%%%%%%%%
\begin{align*}
    &\begin{aligned}
    \begin{tikzpicture}[scale=0.5]
     \begin{scope}
            \fill[gray!20]
                (0,1.5) arc[start angle=90, end angle=270, radius=1.5] -- 
                (0,-0.5) arc[start angle=270, end angle=90, radius=0.5] -- cycle;
        \end{scope}
         \draw[dotted,line width=.7pt,black] (0,0.5)--(0,0.8);
         \draw[line width=.7pt,black] (0,0.8)--(0,1.2);
         \draw[line width=.7pt,black] (0,1.2)--(0,1.5);
         \draw[dotted,line width=.7pt,black] (0,-0.5)--(0,-0.8);
         \draw[line width=.7pt,black] (0,-0.8)--(0,-1.2);
         \draw[line width=.7pt,black] (0,-1.2)--(0,-1.5);
         \draw[line width=.7pt,red] (0,0.8) arc[start angle=90, end angle=270, radius=0.8];
        \node[ line width=0.6pt, dashed, draw opacity=0.5] (a) at (-1,0){$\scriptstyle \tau$};
        \node[ line width=0.6pt, dashed, draw opacity=0.5] (a) at (-0.3,1.3){$\scriptstyle \tau$};
        \node[ line width=0.6pt, dashed, draw opacity=0.5] (a) at (-0.3,0.3){$\scriptstyle \one$}; 
        \node[ line width=0.6pt, dashed, draw opacity=0.5] (a) at (-0.3,-0.3){$\scriptstyle \one$}; 
        \node[ line width=0.6pt, dashed, draw opacity=0.5] (a) at (-0.3,-1.3){$\scriptstyle \tau$};
        \end{tikzpicture}
    \end{aligned} \;\; 
    \begin{aligned}
    \begin{tikzpicture}[scale=0.5]
     \begin{scope}
            \fill[gray!20]
                (0,1.5) arc[start angle=90, end angle=270, radius=1.5] -- 
                (0,-0.5) arc[start angle=270, end angle=90, radius=0.5] -- cycle;
        \end{scope}
         \draw[line width=.7pt,black] (0,0.5)--(0,0.8);
         \draw[dotted,line width=.7pt,black] (0,0.8)--(0,1.2);
         \draw[dotted,line width=.7pt,black] (0,1.2)--(0,1.5);
         \draw[dotted,line width=.7pt,black] (0,-0.5)--(0,-0.8);
         \draw[line width=.7pt,black] (0,-0.8)--(0,-1.2);
         \draw[line width=.7pt,black] (0,-1.2)--(0,-1.5);
         \draw[line width=.7pt,red] (0,0.8) arc[start angle=90, end angle=270, radius=0.8];
        \node[ line width=0.6pt, dashed, draw opacity=0.5] (a) at (-1,0){$\scriptstyle \tau$};
        \node[ line width=0.6pt, dashed, draw opacity=0.5] (a) at (-0.3,1.3){$\scriptstyle \one$};
        \node[ line width=0.6pt, dashed, draw opacity=0.5] (a) at (-0.3,0.3){$\scriptstyle \tau$}; 
        \node[ line width=0.6pt, dashed, draw opacity=0.5] (a) at (-0.3,-0.3){$\scriptstyle \one$}; 
        \node[ line width=0.6pt, dashed, draw opacity=0.5] (a) at (-0.3,-1.3){$\scriptstyle \tau$};
        \end{tikzpicture}
    \end{aligned} \;\; 
        \begin{aligned}
    \begin{tikzpicture}[scale=0.5]
     \begin{scope}
            \fill[gray!20]
                (0,1.5) arc[start angle=90, end angle=270, radius=1.5] -- 
                (0,-0.5) arc[start angle=270, end angle=90, radius=0.5] -- cycle;
        \end{scope}
         \draw[line width=.7pt,black] (0,0.5)--(0,0.8);
         \draw[line width=.7pt,black] (0,0.8)--(0,1.2);
         \draw[line width=.7pt,black] (0,1.2)--(0,1.5);
         \draw[dotted,line width=.7pt,black] (0,-0.5)--(0,-0.8);
         \draw[line width=.7pt,black] (0,-0.8)--(0,-1.2);
         \draw[line width=.7pt,black] (0,-1.2)--(0,-1.5);
         \draw[line width=.7pt,red] (0,0.8) arc[start angle=90, end angle=270, radius=0.8];
        \node[ line width=0.6pt, dashed, draw opacity=0.5] (a) at (-1,0){$\scriptstyle \tau$};
        \node[ line width=0.6pt, dashed, draw opacity=0.5] (a) at (-0.3,1.3){$\scriptstyle \tau$};
        \node[ line width=0.6pt, dashed, draw opacity=0.5] (a) at (-0.3,0.3){$\scriptstyle \tau$}; 
        \node[ line width=0.6pt, dashed, draw opacity=0.5] (a) at (-0.3,-0.3){$\scriptstyle \one$}; 
        \node[ line width=0.6pt, dashed, draw opacity=0.5] (a) at (-0.3,-1.3){$\scriptstyle \tau$};
        \end{tikzpicture}
    \end{aligned} \;\; 
    \begin{aligned}
    \begin{tikzpicture}[scale=0.5]
     \begin{scope}
            \fill[gray!20]
                (0,1.5) arc[start angle=90, end angle=270, radius=1.5] -- 
                (0,-0.5) arc[start angle=270, end angle=90, radius=0.5] -- cycle;
        \end{scope}
         \draw[dotted,line width=.7pt,black] (0,0.5)--(0,0.8);
         \draw[line width=.7pt,black] (0,0.8)--(0,1.2);
         \draw[line width=.7pt,black] (0,1.2)--(0,1.5);
         \draw[line width=.7pt,black] (0,-0.5)--(0,-0.8);
         \draw[dotted,line width=.7pt,black] (0,-0.8)--(0,-1.2);
         \draw[dotted,line width=.7pt,black] (0,-1.2)--(0,-1.5);
         \draw[line width=.7pt,red] (0,0.8) arc[start angle=90, end angle=270, radius=0.8];
        \node[ line width=0.6pt, dashed, draw opacity=0.5] (a) at (-1,0){$\scriptstyle \tau$};
        \node[ line width=0.6pt, dashed, draw opacity=0.5] (a) at (-0.3,1.3){$\scriptstyle \tau$};
        \node[ line width=0.6pt, dashed, draw opacity=0.5] (a) at (-0.3,0.3){$\scriptstyle \one$}; 
        \node[ line width=0.6pt, dashed, draw opacity=0.5] (a) at (-0.3,-0.3){$\scriptstyle \tau$}; 
        \node[ line width=0.6pt, dashed, draw opacity=0.5] (a) at (-0.3,-1.3){$\scriptstyle \one$};
        \end{tikzpicture}
    \end{aligned} \;\; 
    \begin{aligned}
    \begin{tikzpicture}[scale=0.5]
     \begin{scope}
            \fill[gray!20]
                (0,1.5) arc[start angle=90, end angle=270, radius=1.5] -- 
                (0,-0.5) arc[start angle=270, end angle=90, radius=0.5] -- cycle;
        \end{scope}
         \draw[line width=.7pt,black] (0,0.5)--(0,0.8);
         \draw[dotted,line width=.7pt,black] (0,0.8)--(0,1.2);
         \draw[dotted,line width=.7pt,black] (0,1.2)--(0,1.5);
         \draw[line width=.7pt,black] (0,-0.5)--(0,-0.8);
         \draw[dotted,line width=.7pt,black] (0,-0.8)--(0,-1.2);
         \draw[dotted,line width=.7pt,black] (0,-1.2)--(0,-1.5);
         \draw[line width=.7pt,red] (0,0.8) arc[start angle=90, end angle=270, radius=0.8];
        \node[ line width=0.6pt, dashed, draw opacity=0.5] (a) at (-1,0){$\scriptstyle \tau$};
        \node[ line width=0.6pt, dashed, draw opacity=0.5] (a) at (-0.3,1.3){$\scriptstyle \one$};
        \node[ line width=0.6pt, dashed, draw opacity=0.5] (a) at (-0.3,0.3){$\scriptstyle \tau$}; 
        \node[ line width=0.6pt, dashed, draw opacity=0.5] (a) at (-0.3,-0.3){$\scriptstyle \tau$}; 
        \node[ line width=0.6pt, dashed, draw opacity=0.5] (a) at (-0.3,-1.3){$\scriptstyle \one$};
        \end{tikzpicture}
    \end{aligned} \;\; 
        \begin{aligned}
    \begin{tikzpicture}[scale=0.5]
     \begin{scope}
            \fill[gray!20]
                (0,1.5) arc[start angle=90, end angle=270, radius=1.5] -- 
                (0,-0.5) arc[start angle=270, end angle=90, radius=0.5] -- cycle;
        \end{scope}
         \draw[line width=.7pt,black] (0,0.5)--(0,0.8);
         \draw[line width=.7pt,black] (0,0.8)--(0,1.2);
         \draw[line width=.7pt,black] (0,1.2)--(0,1.5);
         \draw[line width=.7pt,black] (0,-0.5)--(0,-0.8);
         \draw[dotted,line width=.7pt,black] (0,-0.8)--(0,-1.2);
         \draw[dotted,line width=.7pt,black] (0,-1.2)--(0,-1.5);
         \draw[line width=.7pt,red] (0,0.8) arc[start angle=90, end angle=270, radius=0.8];
        \node[ line width=0.6pt, dashed, draw opacity=0.5] (a) at (-1,0){$\scriptstyle \tau$};
        \node[ line width=0.6pt, dashed, draw opacity=0.5] (a) at (-0.3,1.3){$\scriptstyle \tau$};
        \node[ line width=0.6pt, dashed, draw opacity=0.5] (a) at (-0.3,0.3){$\scriptstyle \tau$}; 
        \node[ line width=0.6pt, dashed, draw opacity=0.5] (a) at (-0.3,-0.3){$\scriptstyle \tau$}; 
        \node[ line width=0.6pt, dashed, draw opacity=0.5] (a) at (-0.3,-1.3){$\scriptstyle \one$};
        \end{tikzpicture}
    \end{aligned} \;\; 
            \begin{aligned}
    \begin{tikzpicture}[scale=0.5]
     \begin{scope}
            \fill[gray!20]
                (0,1.5) arc[start angle=90, end angle=270, radius=1.5] -- 
                (0,-0.5) arc[start angle=270, end angle=90, radius=0.5] -- cycle;
        \end{scope}
         \draw[dotted,line width=.7pt,black] (0,0.5)--(0,0.8);
         \draw[line width=.7pt,black] (0,0.8)--(0,1.2);
         \draw[line width=.7pt,black] (0,1.2)--(0,1.5);
         \draw[line width=.7pt,black] (0,-0.5)--(0,-0.8);
         \draw[line width=.7pt,black] (0,-0.8)--(0,-1.2);
         \draw[line width=.7pt,black] (0,-1.2)--(0,-1.5);
         \draw[line width=.7pt,red] (0,0.8) arc[start angle=90, end angle=270, radius=0.8];
        \node[line width=0.6pt, dashed, draw opacity=0.5] (a) at (-1,0){$\scriptstyle \tau$};
        \node[ line width=0.6pt, dashed, draw opacity=0.5] (a) at (-0.3,1.3){$\scriptstyle \tau$};
        \node[ line width=0.6pt, dashed, draw opacity=0.5] (a) at (-0.3,0.3){$\scriptstyle \one$}; 
        \node[ line width=0.6pt, dashed, draw opacity=0.5] (a) at (-0.3,-0.3){$\scriptstyle \tau$}; 
        \node[ line width=0.6pt, dashed, draw opacity=0.5] (a) at (-0.3,-1.3){$\scriptstyle \tau$};
        \end{tikzpicture}
    \end{aligned} \;\; 
                \begin{aligned}
    \begin{tikzpicture}[scale=0.5]
     \begin{scope}
            \fill[gray!20]
                (0,1.5) arc[start angle=90, end angle=270, radius=1.5] -- 
                (0,-0.5) arc[start angle=270, end angle=90, radius=0.5] -- cycle;
        \end{scope}
         \draw[line width=.7pt,black] (0,0.5)--(0,0.8);
         \draw[dotted,line width=.7pt,black] (0,0.8)--(0,1.2);
         \draw[dotted,line width=.7pt,black] (0,1.2)--(0,1.5);
         \draw[line width=.7pt,black] (0,-0.5)--(0,-0.8);
         \draw[line width=.7pt,black] (0,-0.8)--(0,-1.2);
         \draw[line width=.7pt,black] (0,-1.2)--(0,-1.5);
         \draw[line width=.7pt,red] (0,0.8) arc[start angle=90, end angle=270, radius=0.8];
        \node[line width=0.6pt, dashed, draw opacity=0.5] (a) at (-1,0){$\scriptstyle \tau$};
        \node[ line width=0.6pt, dashed, draw opacity=0.5] (a) at (-0.3,1.3){$\scriptstyle \one$};
        \node[ line width=0.6pt, dashed, draw opacity=0.5] (a) at (-0.3,0.3){$\scriptstyle \tau$}; 
        \node[ line width=0.6pt, dashed, draw opacity=0.5] (a) at (-0.3,-0.3){$\scriptstyle \tau$}; 
        \node[ line width=0.6pt, dashed, draw opacity=0.5] (a) at (-0.3,-1.3){$\scriptstyle \tau$};
        \end{tikzpicture}
    \end{aligned} \;\; 
                    \begin{aligned}
    \begin{tikzpicture}[scale=0.5]
     \begin{scope}
            \fill[gray!20]
                (0,1.5) arc[start angle=90, end angle=270, radius=1.5] -- 
                (0,-0.5) arc[start angle=270, end angle=90, radius=0.5] -- cycle;
        \end{scope}
         \draw[line width=.7pt,black] (0,0.5)--(0,0.8);
         \draw[line width=.7pt,black] (0,0.8)--(0,1.2);
         \draw[line width=.7pt,black] (0,1.2)--(0,1.5);
         \draw[line width=.7pt,black] (0,-0.5)--(0,-0.8);
         \draw[line width=.7pt,black] (0,-0.8)--(0,-1.2);
         \draw[line width=.7pt,black] (0,-1.2)--(0,-1.5);
         \draw[line width=.7pt,red] (0,0.8) arc[start angle=90, end angle=270, radius=0.8];
        \node[line width=0.6pt, dashed, draw opacity=0.5] (a) at (-1,0){$\scriptstyle \tau$};
        \node[ line width=0.6pt, dashed, draw opacity=0.5] (a) at (-0.3,1.3){$\scriptstyle \tau$};
        \node[ line width=0.6pt, dashed, draw opacity=0.5] (a) at (-0.3,0.3){$\scriptstyle \tau$}; 
        \node[ line width=0.6pt, dashed, draw opacity=0.5] (a) at (-0.3,-0.3){$\scriptstyle \tau$}; 
        \node[ line width=0.6pt, dashed, draw opacity=0.5] (a) at (-0.3,-1.3){$\scriptstyle \tau$};
        \end{tikzpicture}
    \end{aligned} \;\; 
\end{align*}
Haar integral of $\mathcal{T}_{\mathsf{Fib}}$ is of the form (where $\operatorname{rank} \Fib =2$)
\begin{equation}\label{eq:HaarTube}
    \lambda = \frac{1}{\operatorname{rank} \mathsf{Fib}} \sum_{a,x,y,\mu} 
    \sqrt{\frac{d_a}{d_x^3 d_y}}
        \begin{aligned}
    \begin{tikzpicture}
    % 绘制环形区域背景，使红线居中
        \begin{scope}
            \fill[gray!20]
                (0,1.1) arc[start angle=90, end angle=270, radius=1.1] -- 
                (0,-0.5) arc[start angle=270, end angle=90, radius=0.5] -- cycle;
        \end{scope}
          % \draw[line width=0.6pt,black,->] (0,0.5)--(0,0.7);
         %  \draw[line width=0.6pt,black,->] (0,0.7)--(0,1.0);
           \draw[line width=0.6pt,black] (0,0.5)--(0,1.1);
       %lower line
        \draw[line width=.6pt,black] (0,-0.5)--(0,-1.1);
        %\draw[line width=0.6pt,black,->] (0,-1.1)--(0,-0.9);
       % \draw[line width=0.6pt,black,->] (0,-1.1)--(0,-0.6);
        % 绘制红色半圆
        \draw[red, line width=0.6pt] (0,0.8) arc[start angle=90, end angle=270, radius=0.8];
         %  \draw[red, line width=0.6pt, ->] (0,-0.8) arc[start angle=270, end angle=180, radius=0.8];
        % 添加标记节点
        \node[line width=0.6pt, dashed, draw opacity=0.5] at (0,1.3) {$y$};
        \node[line width=0.6pt, dashed, draw opacity=0.5] at (0,-1.3) {$y$};
        \node[line width=0.6pt, dashed, draw opacity=0.5] at (-1,0) {$a$};
        \node[line width=0.6pt, dashed, draw opacity=0.5] at (0.3,-0.7) {$\mu$};
        \node[line width=0.6pt, dashed, draw opacity=0.5] at (0,-0.3) {$x$};
        \node[line width=0.6pt, dashed, draw opacity=0.5] at (0,0.3) {$x$};
        \node[line width=0.6pt, dashed, draw opacity=0.5] at (0.3,0.7) {$\mu$};
    \end{tikzpicture}
\end{aligned}.
\end{equation}
A detailed proof can be found in Ref.~\cite{jia2025weakhopftubealgebra}. 
The Haar measure (i.e., the Haar integral of the dual algebra) is given by
\begin{equation}\label{eq:HaarMTube}
    \Lambda = \frac{1}{\operatorname{rank} \Fib} \sum_{a,x,y,\mu} 
    \sqrt{\frac{d_a}{d_x^3 d_y}}
        \begin{aligned}
    \begin{tikzpicture}
    % 绘制环形区域背景，使红线居中
        \begin{scope}
            \fill[gray!20]
                (0,1.1) arc[start angle=90, end angle=-90, radius=1.1] -- 
                (0,-0.5) arc[start angle=-90, end angle=90, radius=0.5] -- cycle;
        \end{scope}
           % \draw[line width=0.6pt,black,->] (0,0.5)--(0,0.7);
           % \draw[line width=0.6pt,black,->] (0,0.7)--(0,1.0);
           \draw[line width=0.6pt,black] (0,0.5)--(0,1.1);
       %lower line
        \draw[line width=.6pt,black] (0,-0.5)--(0,-1.1);
        %\draw[line width=0.6pt,black,->] (0,-1.1)--(0,-0.9);
       % \draw[line width=0.6pt,black,->] (0,-1.1)--(0,-0.6);
        % 绘制红色半圆
        \draw[blue, line width=0.6pt] (0,-0.8) arc[start angle=-90, end angle=90, radius=0.8];
        %   \draw[red, line width=0.6pt, ->] (0,-0.8) arc[start angle=270, end angle=180, radius=0.8];
        % 添加标记节点
        \node[line width=0.6pt, dashed, draw opacity=0.5] at (0,1.3) {$y$};
        \node[line width=0.6pt, dashed, draw opacity=0.5] at (0,-1.3) {$y$};
        \node[line width=0.6pt, dashed, draw opacity=0.5] at (1,0) {$a$};
        \node[line width=0.6pt, dashed, draw opacity=0.5] at (-0.3,-0.7) {$\mu$};
        \node[line width=0.6pt, dashed, draw opacity=0.5] at (0,-0.3) {$x$};
        \node[line width=0.6pt, dashed, draw opacity=0.5] at (0,0.3) {$x$};
        \node[line width=0.6pt, dashed, draw opacity=0.5] at (-0.3,0.7) {$\mu$};
    \end{tikzpicture}
\end{aligned}.
\end{equation}
The Haar measure induces an inner product on $\mathcal{T}_{\Fib}$ defined by
\begin{equation}
    \langle x, y \rangle = \Lambda(x^* y), \quad x, y \in \mathcal{T}_{\Fib},
\end{equation}
where the $*$-operation $x^*$ is given in Eq.~\eqref{eq:starope}, 
the product $x^* y$ refers to the multiplication in $\mathcal{T}_{\Fib}$, 
and the evaluation is defined via the following pairing function
\begin{equation} \label{eq:pairing_def}
\begin{aligned}
   p\left(\begin{aligned}
    \begin{tikzpicture}[scale=0.65]
        % 绘制环形区域背景，使红线居中
        \begin{scope}
            \fill[gray!15]
                (0,-1.5) arc[start angle=-90, end angle=90, radius=1.5] -- 
                (0,0.5) arc[start angle=90, end angle=-90, radius=0.5] -- cycle;
        \end{scope}
       % \draw[dotted] (0,1.5) arc[start angle=90, end angle=-90, radius=1.5]; 
        %\draw[dotted] (0,0.5) arc[start angle=90, end angle=-90, radius=0.5];
        \draw[line width=0.6pt,black] (0,0.5)--(0,1.5);
        %lower line
        \draw[line width=.6pt,black] (0,-0.5)--(0,-1.5);
        % blue circle
        \draw[blue, line width=0.6pt] (0,1.1) arc[start angle=90, end angle=-90, radius=1.1];
        % 添加标记节点
        \node[ line width=0.6pt, dashed, draw opacity=0.5] (a) at (0,-1.7){$\scriptstyle s$};
        \node[ line width=0.6pt, dashed, draw opacity=0.5] (a) at (1.3,0){$\scriptstyle b$};
        \node[ line width=0.6pt, dashed, draw opacity=0.5] (a) at (0.25,-0.7){$\scriptstyle t$};
        \node[ line width=0.6pt, dashed, draw opacity=0.5] (a) at (-0.25,-1.15){$\scriptstyle \mu$};
        \node[ line width=0.6pt, dashed, draw opacity=0.5] (a) at (0.25,0.7){$\scriptstyle u$};
        \node[ line width=0.6pt, dashed, draw opacity=0.5] (a) at (0,1.7){$\scriptstyle v$};
        \node[ line width=0.6pt, dashed, draw opacity=0.5] (a) at (-0.25,1.15){$\scriptstyle \gamma$};
    \end{tikzpicture}
\end{aligned}\;,\;
\begin{aligned}
    \begin{tikzpicture}[scale=0.65]
        % 绘制环形区域背景，使红线居中
        \begin{scope}
            \fill[gray!15]
                (0,1.5) arc[start angle=90, end angle=270, radius=1.5] -- 
                (0,-0.5) arc[start angle=270, end angle=90, radius=0.5] -- cycle;
        \end{scope}
      %  \draw[dotted] (0,1.5) arc[start angle=90, end angle=270, radius=1.5]; 
       % \draw[dotted] (0,0.5) arc[start angle=90, end angle=270, radius=0.5];
        \draw[line width=0.6pt,black] (0,0.5)--(0,1.5);
        %lower line
        \draw[line width=.6pt,black] (0,-0.5)--(0,-1.5);
        % 绘制红色半圆
        \draw[red, line width=0.6pt] (0,1.1) arc[start angle=90, end angle=270, radius=1.1];
        % 添加标记节点
        \node[ line width=0.6pt, dashed, draw opacity=0.5] (a) at (-1.3,0){$\scriptstyle a$};
        \node[ line width=0.6pt, dashed, draw opacity=0.5] (a) at (0,-1.7){$\scriptstyle x$};
        \node[ line width=0.6pt, dashed, draw opacity=0.5] (a) at (-0.25,-0.7){$\scriptstyle y$};
        \node[ line width=0.6pt, dashed, draw opacity=0.5] (a) at (0.26,-1.15){$\scriptstyle \nu$};
        \node[ line width=0.6pt, dashed, draw opacity=0.5] (a) at (-0.25,0.7){$\scriptstyle z$};
        \node[ line width=0.6pt, dashed, draw opacity=0.5] (a) at (0,1.7){$\scriptstyle w$};
        \node[ line width=0.6pt, dashed, draw opacity=0.5] (a) at (0.2,1.15){$\scriptstyle \zeta$};
    \end{tikzpicture}
\end{aligned}
    \right) 
     = \frac{\delta_{s,w}\delta_{t,x}\delta_{u,y}\delta_{v,z}}{d_s}\;\begin{aligned}
        \begin{tikzpicture}[scale=0.7]
        \path[black!60, fill=gray!15] (0,0) circle[radius=1.4];
       % \draw[dotted] (0,1.4) arc[start angle=90, end angle=270, radius=1.4];
        % \path[fill=white] (0,0.45) circle[radius=0.35];
        % \path[fill=white] (0,-0.45) circle[radius=0.35];
             \draw[line width=.6pt,black] (0,-1.4)--(0,1.4);
             \draw[red, line width=0.6pt] (0,0.95) arc[start angle=90, end angle=270, radius=0.7];
             \draw[blue, line width=0.6pt] (0,-0.95) arc[start angle=-90, end angle=90, radius=0.7];
             \draw[line width=.6pt,black] (0,1.4) arc[start angle=90, end angle=-90, radius=1.4];
             \node[ line width=0.6pt, dashed, draw opacity=0.5] (a) at (1.2,0){$\scriptstyle \bar{s}$};
            \node[ line width=0.6pt, dashed, draw opacity=0.5] (a) at (-0.9,0.4){$\scriptstyle a$};
            \node[ line width=0.6pt, dashed, draw opacity=0.5] (a) at (0.9,-0.4){$\scriptstyle b$};
            \node[ line width=0.6pt, dashed, draw opacity=0.5] (a) at (-0.2,-0.7){$\scriptstyle t$};
            \node[ line width=0.6pt, dashed, draw opacity=0.5] (a) at (-0.2,-1.1){$\scriptstyle \mu$};
            \node[ line width=0.6pt, dashed, draw opacity=0.5] (a) at (0.2,-0.5){$\scriptstyle \nu$};
            \node[ line width=0.6pt, dashed, draw opacity=0.5] (a) at (-0.2,-0.1){$\scriptstyle u$};
            \node[ line width=0.6pt, dashed, draw opacity=0.5] (a) at (-0.2,0.7){$\scriptstyle v$};
            \node[ line width=0.6pt, dashed, draw opacity=0.5] (a) at (-0.2,0.3){$\scriptstyle \gamma$};
            \node[ line width=0.6pt, dashed, draw opacity=0.5] (a) at (0.2,1){$\scriptstyle \zeta$};
        \end{tikzpicture}
    \end{aligned}\;. 
\end{aligned}
\end{equation}
In Ref.~\cite{jia2025weakhopftubealgebra}, we show that this is a weak Hopf skew-pairing.  Using the above data, we can construct the corresponding Fibonacci cluster state model.

\section{Discussion}
\label{sec:Discussion}
In this paper, we propose a general framework for $(1+1)$D lattice models exhibiting non-invertible symmetries. We demonstrate that these models can be analyzed using weak Hopf algebras. 
The weak Hopf cluster state model and its generalizations are discussed in detail, and we show that the underlying symmetry are subalgebras of $H \times \hat{H}$. Thus, dual symmetry plays a crucial role. In the group case, non-invertibility primarily arises due to the existence of dual symmetry. For the general Hopf and weak Hopf cases, non-invertibility emerges from both components of the $H \times \hat{H}$ symmetry.

Despite the progress made, several important open problems remain. A central question is the classification of all $(1+1)$D phases with weak Hopf symmetry. We conjecture that such phases can be classified within a weak Hopf cohomological framework. Another major challenge is the extension of these results to higher dimensions. For $(2+1)$D and beyond, additional categorical structures must be introduced to fully capture the nature of the symmetry.  
From the viewpoint of Tannaka–Krein duality, there should exist an algebraic object corresponding to a fusion $n$-category. Identifying this structure and constructing the associated lattice realization constitute central problems for future investigation.  
The potential applications of this model to measurement-based quantum computation and other quantum information tasks are also of great interest. We plan to explore these directions in future work.

\begin{acknowledgments}
I sincerely thank Dagomir Kaszlikowski for his support, Liang Kong and Sheng Tan for discussions on SymTFT during my visit to BIMSA.
This work is supported by the National Research Foundation in Singapore, the A*STAR under its CQT Bridging Grant, CQT-Return of PIs EOM YR1-10 Funding and  CQT Young Researcher Career Development Grant.

\end{acknowledgments}

\bibliographystyle{apsrev4-1-title}
\bibliography{Jiabib}

\end{document}

%% file: JiaDef.tex
\usepackage{jiatitle}
\pdfoutput = 1
\usepackage{amsmath}
\usepackage{amssymb}
\usepackage{amsthm}
\usepackage{bm}
\usepackage{braket}
\usepackage{graphicx}
\usepackage[caption=false]{subfig}
\usepackage{adjustbox}

\usepackage{mathdots}

\usepackage{multirow}
\usepackage{tabu}

\usepackage{booktabs} % for better lines
\usepackage{xcolor}   % for defining custom colors
\usepackage{hhline}   % for custom horizontal lines

\usepackage{orcidlink} %%%%%%oricd link%%%%%%%

\usepackage{tikz}
\usetikzlibrary{positioning,intersections}
\usetikzlibrary{calc}
\usetikzlibrary{shapes.geometric}
\usepackage{braids}
\usepackage{tikz-cd}
\usetikzlibrary{shapes.geometric,shapes.misc}
\usetikzlibrary{arrows,matrix,calc,scopes,decorations.markings,snakes}
\usetikzlibrary{arrows.meta}
\usetikzlibrary{intersections, patterns,fit} 
\usetikzlibrary{decorations.pathreplacing,angles,quotes}
\usetikzlibrary{tqft} %for TQFT graph
\newcommand{\emaillogo}{\tikz[baseline=-0.05em,scale=0.4]{
  \draw[thick,line width=0.5pt] (0,0) rectangle (0.5,0.35);
  \draw[thick,line width=0.5pt] (0,0.35) -- (0.25,0.2) -- (0.5,0.35);
}}
\definecolor{darkblue}{RGB}{40, 85, 120} % Darker blue for better contrast

\theoremstyle{theorem}
\newtheorem{theorem}{Theorem}
\newtheorem{lemma}{Lemma}
\newtheorem{definition}{Definition}

\newtheorem{proposition}{Proposition}

\theoremstyle{remark}
\newtheorem{remark}{Remark}[section]

\usepackage[bbgreekl]{mathbbol} %change mathbb style
\usepackage{anyfontsize}
\renewcommand\normalsize{\fontsize{10.5pt}{14pt}\selectfont}

\DeclareMathAlphabet{\mathcal}{OMS}{cmsy}{m}{n}
\usepackage{euscript}

\newcommand\eB           {\EuScript{B}}
\newcommand\eC           {\EuScript{C}}
\newcommand\eD           {\EuScript{D}}

\newcommand\eI           {\EuScript{I}}

\newcommand\eM          {\EuScript{M}}

\newcommand\eS         {\EuScript{S}}

\newcommand{\Av}{\mathbb{A}}
\newcommand{\Bf}{\mathbb{B}}

\newcommand{\Zbb}{\mathbb{Z}}
\newcommand{\Cbb}{\mathbb{C}}

\newcommand{\Hbb}{\mathbb{H}}

\newcommand{\one}{\mathbb{1}}

\newcommand{\Rep}{\mathsf{Rep}}
\newcommand{\Mod}{\mathsf{Mod}}

\newcommand{\Fib}{\mathsf{Fib}}

\newcommand{\XR}{\overset{\rightarrow}{X}}
\newcommand{\XL}{\overset{\leftarrow}{X}}

\newcommand{\ZR}{\overset{\rightarrow}{Z}}
\newcommand{\ZL}{\overset{\leftarrow}{Z}}

\newcommand{\cone}{(1)}
\newcommand{\ctwo}{(2)}
\newcommand{\cthree}{(3)}
\newcommand{\cfour}{(4)}

\newcommand\Tr{\operatorname{Tr}}

\newcommand{\id}{\operatorname{id}}
\newcommand{\Irr}{\operatorname{Irr}}
\newcommand{\Hom}{\operatorname{Hom}}
\newcommand{\End}{\operatorname{End}}

\newcommand{\Cocom}{\operatorname{Cocom}}

\newcommand{\Fun}{\mathsf{Fun}}
\newcommand{\Vect}{\mathsf{Vect}}

\newcommand{\cA}{\text{\usefont{OMS}{cmsy}{m}{n}A}}

\newcommand{\cF}{\text{\usefont{OMS}{cmsy}{m}{n}F}}

\newcommand{\cH}{\text{\usefont{OMS}{cmsy}{m}{n}H}}

\makeatletter
\renewcommand{\tableofcontents}{%
  \begingroup
  \setlength{\parskip}{1pt}%
  \renewcommand{\baselinestretch}{0.9}\normalsize
  \@starttoc{toc}%
  \endgroup
}
\makeatother